\documentclass[9pt]{article}
\usepackage{amsmath,amsfonts,amssymb,amsthm,hyperref,enumerate,multicol,eufrak}
\usepackage{calligra,indentfirst,epsfig,lettrine}
\usepackage{caption}
\usepackage{relsize}
\usepackage{tabls}
\usepackage{array}
\usepackage{longtable}
\calligra

\DeclareFontFamily{U}{mathc}{}
\DeclareFontShape{U}{mathc}{m}{it}%
{<->s*[1.03] mathc10}{}

\DeclareMathAlphabet{\mathscr}{U}{mathc}{m}{it}

\usepackage{comment}
\makeatletter
\newcommand*{\rom}[1]{\expandafter\@slowromancap\romannumeral #1@}
\makeatother
\usepackage{mathtools}
\usepackage{graphicx}
\usepackage{subfigure}
\usepackage{xcolor}
\usepackage{color, soul}
\usepackage{mathrsfs}
\DeclareMathAlphabet{\mathpzc}{OT1}{pzc}{m}{it}
\DeclarePairedDelimiter\ceil{\lceil}{\rceil}
\DeclarePairedDelimiter\floor{\lfloor}{\rfloor}
\usepackage{ upgreek }
\newcommand{\etal}{\textit{et al.}}

\newtheorem{theorem}{Theorem}[section]
\newtheorem{example}{Example}[section]
\newtheorem{lemma}{Lemma}[section]
\newtheorem{remark}{Remark}[section]
\newtheorem{proposition}{Proposition}[section]
\allowdisplaybreaks
\numberwithin{equation}{section}

\begin{document} \title{{A recursive approach to the construction and enumeration of self-orthogonal and self-dual codes over finite commutative chain rings of even characteristic}}
\author{Monika Yadav{\footnote{Email address:~\urlstyle{same}\href{mailto:monikay@iiitd.ac.in}{monikay@iiitd.ac.in}}}  ~and 
  Anuradha Sharma{\footnote{Corresponding Author, Email address:~\urlstyle{same}\href{mailto:anuradha@iiitd.ac.in}{anuradha@iiitd.ac.in}} }\\
 {Department of  Mathematics, IIIT-Delhi}\\{New Delhi 110020, India}}
\date{}
\maketitle
\begin{abstract}
Let $\mathcal{R}_{e,m}$ be a finite commutative chain ring of even characteristic  with maximal ideal $\langle u \rangle$ of nilpotency index $e \geq 2,$  Teichm$\ddot{u}$ller set $\mathcal{T}_{m},$ and residue field $\mathcal{R}_{e,m}/\langle u \rangle$ of order $2^m.$ Suppose that $2 \in \langle u^{\kappa}\rangle \setminus \langle u^{\kappa+1}\rangle$ for some even positive integer $ \kappa \leq  e.$   In this paper,   we  provide a recursive method to construct a self-orthogonal code $\mathcal{C}_e$ of type $\{\lambda_1, \lambda_2, \ldots, \lambda_e\}$ and length $n$ over  $\mathcal{R}_{e,m}$   from a chain \vspace{-1mm}\begin{equation*}\vspace{-1mm}
\mathcal{D}^{(1)}\subseteq \mathcal{D}^{(2)} \subseteq \cdots \subseteq \mathcal{D}^{(\ceil{\frac{e}{2}})} \end{equation*} of self-orthogonal codes of length $n$ over $\mathcal{T}_{m},$ and vice versa, where $\dim \mathcal{D}^{(i)}=\lambda_1+\lambda_2+\cdots+\lambda_i$ for $1 \leq i \leq \ceil{\frac{e}{2}},$  the codes $\mathcal{D}^{(\floor{\frac{e+1}{2}}-\kappa)},\mathcal{D}^{(\floor{\frac{e+1}{2}}-\kappa+1)},\ldots,\mathcal{D}^{(\floor{\frac{e}{2}}-\frac{\kappa}{2})}$  satisfy certain additional conditions, and  $\lambda_1,\lambda_2,\ldots,\lambda_e$ are   non-negative integers  satisfying $2\lambda_1+2\lambda_2+\cdots+2\lambda_{e-i+1}+\lambda_{e-i+2}+\lambda_{e-i+3}+\cdots+\lambda_i \leq n$  for $\ceil{\frac{e+1}{2}}\leq i\leq e,$ (here  
 $\floor{\cdot }$ and $\ceil{\cdot }$ denote the floor and ceiling functions, respectively).  This construction guarantees that $Tor_i(\mathcal{C}_e)=\mathcal{D}^{(i)}$ for $1 \leq i \leq \ceil{\frac{e}{2}}.$
By employing this recursive construction method, together with the results from group theory and finite geometry,  we derive explicit enumeration formulae for all self-orthogonal and self-dual codes of an arbitrary length over $\mathcal{R}_{e,m}.$ We also demonstrate these results through examples.
A subsequent study \cite{YSub2}   addresses the complementary case where $\kappa$ is odd.
\end{abstract}
 \section{Introduction}\label{intro}

Self-orthogonal and self-dual codes are among the most significant and extensively studied classes of linear codes due to their intricate algebraic structures and wide-ranging applications. These codes are deeply intertwined with combinatorial design theory \cite{Gaborit,Kennedy}, modular forms, and unimodular lattices \cite{HaradaBannai,Dough,P}, reflecting their rich mathematical framework. In addition to their theoretical importance, self-orthogonal and self-dual codes are instrumental in several critical areas. They form the foundation of code-based cryptography — a leading candidate for post-quantum secure cryptosystems — where their structural properties have been exploited to enhance efficiency and decoding performance in schemes such as the McEliece and Niederreiter cryptosystems \cite{Bernstein,Mariot}. Furthermore, these codes are essential in quantum error correction, particularly in the construction of stabilizer codes that protect quantum information against decoherence and noise \cite{Ashikhmin,XingJin}.  Motivated by these diverse applications, extensive research has been devoted to understanding the structure and construction of self-orthogonal and self-dual codes \cite{BETTY,Dough,Dougherty}.

The theory of self-orthogonal and self-dual codes has advanced substantially over the past decades. A major breakthrough occurred in the 1990s when it was discovered that many important binary non-linear codes — including Kerdock, Preparata, Goethals, and Delsarte-Goethals codes — can be realized as Gray images of linear codes over the ring 
$\mathbb{Z}_4$	 of integers modulo $4$ \cite{R,sole}. This discovery sparked significant interest in exploring self-orthogonal and self-dual codes over more general algebraic structures, especially finite commutative chain rings \cite{Choi,Dougherty,Zp,GBN,Fidel,AT,K}. 
    In particular, the explicit enumeration and construction of self-orthogonal and self-dual codes over various finite commutative chain rings has gained significant attention \cite{b,Choi,Zp,Jose,Y}. Such enumeration results are not only theoretically valuable but also essential for classifying these codes up to monomial equivalence \cite{ GBN,AT,Y}. Below, we provide a summary of key results in this direction.

Throughout this paper, for a prime number $p$ and positive integers $m$ and  $\mathfrak{s},$ let  $GR(p^{\mathfrak{s}},m)$ denote the Galois ring of characteristic $p^{\mathfrak{s}}$ and cardinality $p^{\mathfrak{s}m}.$  
By Theorem XVII.5 of \cite{Mcdonald}, every finite commutative chain ring is isomorphic to a quotient ring of the form \vspace{-2mm}\begin{equation}\vspace{-1mm}\label{Chain}\mathcal{R}_{e,m}=\frac{GR(p^{\mathfrak{s}},m)[x]}{\langle g(x),p^{\mathfrak{s}-1}x^{\mathtt{t}}\rangle},\end{equation}
   where $g(x)=x^{\kappa}+p(g_{\kappa -1}x^{\kappa-1}+\cdots+g_1x+g_0) \in GR(p^{\mathfrak{s}},m)[x] $ is an Eisenstein polynomial of degree $\kappa \geq 1$ with $g_0$  as a unit in $GR(p^{\mathfrak{s}},m),$  the integer $\mathtt{t}$ satisfies  $1\leq \mathtt{t}\leq \kappa$ when $\mathfrak{s}\geq 2$, while $\mathtt{t}=\kappa$ when $\mathfrak{s}=1,$ and $e=\kappa(\mathfrak{s}-1)+\mathtt{t}.$ The converse also holds.
   
  In the special case  when $\mathfrak{s}=\kappa=\mathtt{t}=1,$ 
the chain ring $\mathcal{R}_{e,m}$ reduces to the finite field 
$\mathbb{F}_{p^m}$ of order $p^m,$  and the enumeration formulae for self-orthogonal and self-dual codes over $\mathbb{F}_{p^m}$ follow from Pless \cite{V}.  For a detailed study on the enumeration of self-orthogonal and self-dual codes over finite commutative chain rings of odd characteristic, \textit{i.e.}, when $p$ is an odd prime,
the reader is referred to \cite{b,BETTY,Zp,GBN,AT,K, Y}  and references therein. 

 When $p=2,$ $m=1,$ $\mathfrak{s}\geq 1$ and $\kappa=1,$ we have $\mathtt{t}=1$ and $\mathcal{R}_{e,m}\simeq \mathbb{Z}_{2^{e}},$ and   the enumeration formula for  self-dual codes over $\mathbb{Z}_{2^{e}}$ is derived by Nagata  {\etal}  \cite{Fidel}.
On the other hand, when $p=2$ and $\mathfrak{s}=1,$ we have $\mathtt{t}=\kappa,$ $GR(2^{\mathfrak{s}},m)=\mathbb{F}_{2^m}$ and $\mathcal{R}_{e,m}\simeq \mathbb{F}_{2^m}[x]/\langle x^{\kappa}\rangle,$  which is a quasi-Galois ring of even characteristic.   Yadav and Sharma \cite{quasi} provided a recursive method to construct
and enumerate self-orthogonal and self-dual codes over quasi-Galois rings of even characteristic using linear codes with a certain structural property $(\ast)$ (see Section 4 
of  \cite{quasi}).

Furthermore, when $p=2,$  $\mathfrak{s}\geq 1$ and $\kappa=1,$ we have $\mathtt{t}=1$ and $\mathcal{R}_{e,m}\simeq GR(2^{e},m).$   In this setting, Yadav and Sharma \cite{Galois} observed that the enumeration techniques used in \cite{Fidel,quasi} and \cite{Y} could not be directly extended  to count self-orthogonal and self-dual codes over $GR(2^e,m).$  To overcome this limitation, they further defined doubly even codes over the Teichm$\ddot{u}$ller set $\mathcal{T}$ of $GR(2^e,m)$ and provided a recursive method to construct and enumerate self-orthogonal and self-dual codes over $GR(2^e,m)$ from doubly even codes over $\mathcal{T}.$   

Nevertheless, as demonstrated in Section 3 of \cite{quasi} and illustrated in Example \ref{EX1}, these existing recursive methods do not extend, in a straightforward manner, to finite commutative chain rings of even characteristic when both  $\mathfrak{s}, \kappa \geq 2.$ In a recent work \cite{YSub2}, we investigate the case where $\kappa$ is odd and present a recursive method for constructing self-orthogonal and self-dual codes over $\mathcal{R}_{e,m}$ from a nested sequence of self-orthogonal codes  over the Teichm$\ddot{u}$ller set  $\mathcal{T}_m$ of $\mathcal{R}_{e,m}$, among which exactly $\lfloor \tfrac{e}{2} \rfloor - \lfloor \tfrac{\kappa}{2} \rfloor$ codes are doubly even, subject to certain additional conditions. 

The primary objective of this paper is to employ the modified recursive method derived in Section 3 of Yadav and Sharma \cite{YSub2}  for the construction and enumeration of all self-orthogonal and self-dual codes over $\mathcal{R}_{e,m}$ in the case when $p=2,$ $e \geq 2,$ and $ \kappa\geq 2$ is an even integer. In this context, we construct  these codes  from a nested sequence of self-orthogonal codes    over the Teichm$\ddot{u}$ller set  $\mathcal{T}_m$ of $\mathcal{R}_{e,m}$, satisfying  conditions \textbf{A1)} - \textbf{A5)} in Lemma \ref{l3.4} when $2\kappa \leq e,$ and   conditions \textbf{B1)} - \textbf{B4)} in Lemma \ref{l3.4a} when $2\kappa > e.$   Together with the results derived in \cite{YSub2}, \cite{Galois} and \cite{Y}, this work completely solves the problem of construction and enumeration of self-orthogonal and self-dual codes over  finite commutative chain rings.  This problem is not only fundamental in the theoretical development of algebraic coding theory  but also has significant practical implications. These enumeration formulae are useful in counting   self-orthogonal and self-dual quasi-abelian and  Galois additive cyclic codes over finite commutative chain rings (as demonstrated in  \cite{Jose, Lavanya}), and other specialized linear codes that can be expressed as direct sums of linear codes over finite commutative chain rings. These enumeration formulae are particularly valuable for classifying codes up to monomial equivalence, which is a crucial aspect of the code equivalence problem. Notably, monomially equivalent codes preserve important invariants such as the homogeneous and overweight distances \cite{Gassner,Ozbudak}.

Throughout this paper, we assume that $p=2,$ $e \geq 2$ and $\kappa \geq 2$ is an even integer, and that the chain ring $\mathcal{R}_{e,m}$ is as defined in \eqref{Chain}. 
This paper is organized as follows: In Section \ref{prelim}, we  present some  preliminaries essential for establishing our main results. In Section \ref{construction}, we provide  a  recursive method for  constructing a self-orthogonal (\textit{resp.} self-dual) code $\mathcal{C}_e$ of type  $\{\lambda_1, \lambda_2, \ldots, \lambda_e\}$ and  length $n$ over  $\mathcal{R}_{e,m}$ satisfying $Tor_i(\mathcal{C}_e)=\mathcal{D}^{(i)}$  from a chain 
\vspace{-1mm}\begin{equation*}\vspace{-1mm}
    \mathcal{D}^{(1)}\subseteq \mathcal{D}^{(2)} \subseteq \cdots \subseteq \mathcal{D}^{(\ceil{\frac{e}{2}})} \end{equation*} of self-orthogonal codes of length $n$ over the Teichm$\ddot{u}$ller set    $\mathcal{T}_m$ of $\mathcal{R}_{e,m}$  with  $\dim \mathcal{D}^{(i)}=\lambda_1+\lambda_2+\cdots+\lambda_i$ for $1 \leq i \leq \ceil{\frac{e}{2}},$ where the codes $\mathcal{D}^{(\floor{\frac{e+1}{2}}-\kappa)},\mathcal{D}^{(\floor{\frac{e+1}{2}}-\kappa+1)},\ldots,\mathcal{D}^{(\floor{\frac{e}{2}}-\frac{\kappa}{2})}$  satisfy conditions \textbf{A1)} - \textbf{A5)} in Lemma \ref{l3.4} when $2\kappa \leq e,$ and   conditions \textbf{B1)} - \textbf{B4)} in Lemma \ref{l3.4a} when $2\kappa > e,$ 
 and $\lambda_1,\lambda_2,\ldots,\lambda_e$ are non-negative integers satisfying $2\lambda_1+2\lambda_2+\cdots+2\lambda_{e-i+1}+\lambda_{e-i+2}+\lambda_{e-i+3}+\cdots+\lambda_i \leq n$ for $\ceil{\frac{e+1}{2}}\leq i\leq e$ (see Propositions \ref{pe=k=2} - \ref{p3.9Keven} and Construction methods (\textbf{X}) and (\textbf{Y})). 
In Section \ref{counting},  leveraging Construction methods (\textbf{X}) and (\textbf{Y}), we derive explicit enumeration formulae for all self-orthogonal and self-dual codes of type $\{\lambda_1,\lambda_2,\ldots, \lambda_e\}$ and length $n$  over  $\mathcal{R}_{e,m}$ (Theorems \ref{t4.4e=k=2} - \ref{t4.2Keven}). Furthermore,  these enumeration formulae, together with Remarks \ref{R2.1} and  \ref{R4.1}, yield enumeration formulae for all self-orthogonal and self-dual codes of length $n$ over $\mathcal{R}_{e,m}.$ Notably, when $\mathfrak{s}=1,$ we have $e=\kappa=\mathtt{t},$ and the ring $\mathcal{R}_{e,m}$ reduces to the quasi-Galois ring $\mathbb{F}_{2^m}[x]/\langle x^{\kappa}\rangle.$  In this particular case, for any even $\kappa,$ the enumeration formulae for self-orthogonal and self-dual codes over $\mathbb{F}_{2^m}[x]/\langle x^{\kappa}\rangle$ can be deduced from Theorems \ref{t4.4e=k=2}, \ref{t4.1Keven} -  \ref{t4.3e=k=2} and \ref{t4.2Keven} by setting $\mathfrak{s}=1$ and $e=\mathtt{t}=\kappa$ (see Remarks \ref{R4.1} - \ref{rem8.2}).

 \section{Some preliminaries}\label{prelim} 
In this section, we  present some basic concepts and notations necessary for  establishing our main results. We begin by summarizing essential properties of finite commutative chain rings, followed by an overview of linear codes over finite commutative chain rings, including duality, self-orthogonality, and torsion codes.
\subsection{Basic properties of finite commutative chain rings}\label{prelim1}

A finite commutative ring with unity is called a chain ring if its ideals form a chain under inclusion. Notable examples of finite commutative chain rings include finite fields, quasi-Galois rings, and Galois rings. In general, we see, by \eqref{Chain}, that every finite commutative chain ring is a quotient ring of the form $\mathcal{R}_{e,m}=\frac{GR(p^{\mathfrak{s}},m)[x]}{\langle g(x),p^{\mathfrak{s}-1}x^{\mathtt{t}}\rangle},$ where  $g(x)$ is an Eisenstein polynomial of  degree $\kappa$ over the Galois ring $GR(p^{\mathfrak{s}},m),$  and  the integer $\mathtt{t}$ satisfies  $1\leq \mathtt{t}\leq \kappa$ when $\mathfrak{s}\geq 2$, whereas  $\mathtt{t}=\kappa$ when $\mathfrak{s}=1.$   
   
    Let us define the coset  $u:= x+\langle g(x),p^{\mathfrak{s}-1}x^{\mathtt{t}}\rangle \in  \mathcal{R}_{e,m}.$ Then the ideal $\langle u \rangle$  is the unique maximal ideal of  $\mathcal{R}_{e,m}$  of nilpotency index
        $e=\kappa (\mathfrak{s}-1)+\mathtt{t}.$ In fact, 
         the ideals of $\mathcal{R}_{e,m}$ form the chain
         $$\{0\}\subsetneq   \langle u^{e-1}\rangle \subsetneq \langle u^{e-2}\rangle \subsetneq \cdots \subsetneq \langle u\rangle \subsetneq \langle 1\rangle=\mathcal{R}_{e,m}.$$  
          By Lemma  XVII.4 of \cite{Mcdonald}, we have  $|\langle u^i\rangle|=p^{m(e-i)}$ for $0\leq i\leq e,$ where $|\cdot|$ denotes the cardinality function.  The residue field $\overline{\mathcal{R}}_{e,m}=\mathcal{R}_{e,m}/\langle u \rangle$ has order $p^m.$ Note that $p \in \langle u^{\kappa}\rangle \setminus \langle u^{\kappa+1}\rangle.$ Moreover, we have $e=\kappa(\mathfrak{s}-1)+\mathtt{t}.$
    
  In addition, by Theorem XVII.5 of \cite{Mcdonald},  there exists an element $\zeta \in \mathcal{R}_{e,m}$ of multiplicative order  $p^m-1$. The cyclic subgroup generated by $\zeta$ is the unique subgroup of the unit group of $\mathcal{R}_{e,m}$, which is isomorphic to the multiplicative group of the residue field $\overline{\mathcal{R}}_{e,m}.$  The Teichm$\ddot{u}$ller set of $\mathcal{R}_{e,m}$ is defined as $\mathcal{T}_{e,m}=\{0,1,\zeta,\zeta^2,\ldots, \zeta^{p^m-2}\}.$  By Lemma  XVII.4 of \cite{Mcdonald}, each element $a\in \mathcal{R}_{e,m}$ admits a unique Teichm$\ddot{u}$ller representation  $a=a_{0}+ua_1+\cdots+u^{e-1}a_{e-1},$ where $a_0,a_1,\ldots,a_{e-1}\in \mathcal{T}_{e,m}.$  Moreover, $a$  is a unit in $\mathcal{R}_{e,m}$ if and only if $a_0\neq0.$ Note that there exists a canonical epimorphism  $^{-} : \mathcal{R}_{e,m}\rightarrow \overline{\mathcal{R}}_{e,m},$ given by  $d\mapsto \bar{d}=d+\langle u\rangle$ for all $d\in \mathcal{R}_{e,m}$. It is straightforward to verify that  the function $^{-}{\restriction_{\mathcal{T}_{e,m}}}:\mathcal{T}_{e,m}\rightarrow \overline{\mathcal{R}}_{e,m}$  is a field isomorphism.  

Next, for $0 \leq i \leq e-1,$ define a map $\pi_i: \mathcal{R}_{e,m} \rightarrow \mathcal{T}_{e,m}$ as $\pi_i(d)=d_i$ for all  $d=d_{0}+ud_1+\cdots+u^{e-1}d_{e-1},$ where $d_0, d_1,\ldots, d_{e-1} \in \mathcal{T}_{e,m}.$  The map $\pi_0$ induces   a binary operation $\oplus$ on $\mathcal{T}_{e,m},$ given by $w\oplus v=\pi_0(w+v)$ for all $w,v \in \mathcal{T}_{e,m}.$ Under the addition operation $\oplus$ and the usual multiplication operation of $\mathcal{R}_{e,m}$, the Teichm$\ddot{u}$ller set $\mathcal{T}_{e,m}$ forms a finite field of order $p^m,$ isomorphic to the residue field $\overline{\mathcal{R}}_{e,m}.$

Additionally, for an integer $\ell$ satisfying  $1 \leq \ell < e,$   the quotient ring $\mathcal{R}_{e,m}/\langle u^{\ell} \rangle $ is also a finite commutative chain ring with the unique maximal ideal  $\langle u+\langle u^{\ell} \rangle \rangle $  of nilpotency index $\ell.$  For convenience, we denote this quotient ring $\mathcal{R}_{e,m}/\langle u^{\ell} \rangle $ by   $\mathcal{R}_{\ell,m}.$  
The element $\zeta_{\ell} := \zeta +\langle u^{\ell}\rangle\in \mathcal{R}_{\ell,m}$ has  multiplicative order $p^m-1$ and  the Teichm$\ddot{u}$ller set of $\mathcal{R}_{\ell,m}$ is given by  $\mathcal{T}_{\ell,m}=\{0,1, \zeta_{\ell}, \zeta_{\ell}^2,\ldots, \zeta_{\ell}^{p^m-2}\}.$  There is a canonical epimorphism from $\mathcal{R}_{e,m}$ onto $\mathcal{R}_{\ell,m},$ given by  $b \mapsto b+\langle u^{\ell}\rangle$ for all $b \in \mathcal{R}_{e,m}.$ For convenience, we will identify each element $a +\langle u^{\ell}\rangle \in \mathcal{R}_{\ell,m}$ with its representative $a \in \mathcal{R}_{e,m},$ performing addition and multiplication in $\mathcal{R}_{\ell,m}$ modulo $u^{\ell}.$ In particular, we will identify the element $\zeta_{\ell}\in \mathcal{T}_{\ell,m}$ with  its representative $\zeta \in \mathcal{T}_{e,m}.$ Under this identification,   we assume, throughout this paper, that 
\vspace{-0.5mm}\begin{equation*}\vspace{-0.5mm}
\mathcal{T}_{1,m}=\mathcal{T}_{2,m}=\cdots=\mathcal{T}_{e,m}=\{0,1,\zeta,{\zeta}^2,\ldots, {\zeta}^{p^m-2}\}=\mathcal{T}_{m}\text{ (say)}\end{equation*} and \vspace{-0.5mm}\begin{equation*}\vspace{-0.5mm}
\mathcal{R}_{1,m}=\overline{\mathcal{R}}_{1,m}=\overline{\mathcal{R}}_{2,m}=\cdots=\overline{\mathcal{R}}_{e,m}=\mathcal{T}_{m} .\end{equation*} 
Hence,  for $1 \leq \ell \leq e,$  each element $a \in \mathcal{R}_{\ell,m}$  can be uniquely expressed  as $a=a_0+ua_1 +u^2a_2 +\cdots+u^{\ell-1}a_{\ell-1},$ where $a_0,a_1,a_2,\ldots,a_{\ell-1} \in \mathcal{T}_{m}.$ Similarly, any matrix $A \in \mathcal{M}_{h_1\times h_2}(\mathcal{R}_{\ell,m}) $ admits a unique representation $A=A_0+uA_1+u^2A_2+\cdots+u^{\ell-1}A_{\ell-1},$ where $A_0,A_1,\ldots,A_{\ell-1} \in \mathcal{M}_{h_1\times h_2}(\mathcal{T}_{m}).$

 \subsection{Linear codes over finite commutative chain rings}\label{linearcodes}
Let $n$ be a positive integer, and let $\mathcal{R}_{e,m}^n$ denote the $\mathcal{R}_{e,m}$-module consisting of all $n$-tuples over $\mathcal{R}_{e,m}.$  A linear code $\mathcal{C}$  of length $n$ over $\mathcal{R}_{e,m}$  is
defined as an $\mathcal{R}_{e,m}$-submodule of $\mathcal{R}_{e,m}^n,$ whose elements 
are referred to as codewords.  A generator matrix of the code $\mathcal{C}$  is a matrix over  $\mathcal{R}_{e,m}$ 	
  whose rows form a minimal generating set for the code $\mathcal{C}$ as an $\mathcal{R}_{e,m}$-module. 
  
Two linear codes of length $n$ over $\mathcal{R}_{e,m}$ are said to be permutation equivalent if one  can be obtained from the other by  permuting its coordinate positions.  
By  Proposition 3.2 of Norton and  S$\check{a}$l$\check{a}$gean \cite{Norton},  every linear code $\mathcal{C}$ of length $n$ over $\mathcal{R}_{e,m}$ is permutation equivalent to a code having a generator matrix in the following standard form:
 \vspace{-1mm}\begin{equation}\label{e1.1} 
 \begin{aligned}
G=\begin{bmatrix}
	\mathtt{I}_{\lambda_1}& \mathtt{A}_{1,1}&\cdots& \mathtt{A}_{1,e-2}&\mathtt{A}_{1,e-1}& \mathtt{A}_{1,e} \\
	\mathbf{0} & u\mathtt{I}_{\lambda_2} &\cdots & u\mathtt{A}_{2,e-2}& u\mathtt{A}_{2,e-1} &u\mathtt{A}_{2,e} \\
	\vdots & \vdots &\vdots &\vdots& \vdots&\vdots\\
	\mathbf{0}&\mathbf{0}&\cdots&u^{e-2}\mathtt{I}_{\lambda_{e-1}}& u^{e-2}\mathtt{A}_{e-1,e-1}&u^{e-2}\mathtt{A}_{e-1,e}\\
	\mathbf{0}&\mathbf{0}&\cdots &\mathbf{0}&u^{e-1}\mathtt{I}_{\lambda_{e}}&u^{e-1}\mathtt{A}_{e,e}
\end{bmatrix}=\begin{bmatrix}
    T^{(e)}_1\\uT^{(e)}_2\\ \vdots\\u^{e-2}T_{e-1}^{(e)}\\u^{e-1}T_{e}^{(e)}
\end{bmatrix},
\end{aligned}\end{equation}
where the columns of $G$  are grouped into blocks of sizes $\lambda_1$, $\lambda_2$, $\ldots$, $\lambda_{e}$ and $ \lambda_{e+1}=n-(\lambda_1+\lambda_2+\cdots+\lambda_e).$ Here, $\mathtt{I}_{\lambda_i}$ denotes the $\lambda_i\times\ \lambda_i$ identity matrix over $\mathcal{R}_{e,m}$ and $\mathtt{A}_{i,j}\in \mathcal{M}_{\lambda_i\times \lambda_{j+1}}(\mathcal{R}_{e,m})$  is considered modulo $u^{j-i+1}$ for $1\leq i\leq j\leq e.$ More precisely,  the matrix  $\mathtt{A}_{i,j}\in \mathcal{M}_{\lambda_i\times \lambda_{j+1}}(\mathcal{R}_{e,m})$ is of the form $\mathtt{A}_{i,j}=\mathtt{A}_{i,j}^{(0)}+u\mathtt{A}_{i,j}^{(1)}+\cdots+u^{j-i}\mathtt{A}_{i,j}^{(j-i)}$ with $\mathtt{A}_{i,j}^{(0)},\mathtt{A}_{i,j}^{(1)},\ldots,\mathtt{A}_{i,j}^{(j-i)}\in \mathcal{M}_{\lambda_i\times \lambda_{j+1}}(\mathcal{T}_{e,m})$  for $1\leq i\leq j\leq e.$ (Throughout this paper, $\mathcal{M}_{\ell\times h}(R)$ denotes the set of all $\ell \times h$ matrices over a ring $R.$)  

A linear code $\mathcal{C}$ of length $n$ over $\mathcal{R}_{e,m}$ is said to be of type $\{\lambda_1,\lambda_2,\ldots,\lambda_{e}\}$ if it is permutation equivalent to a code with a generator matrix in the above standard form  \eqref{e1.1}.

The Euclidean bilinear form on $\mathcal{R}_{e,m}^n$ is a map  $\cdot: \mathcal{R}_{e,m}^n \times \mathcal{R}_{e,m}^n \to \mathcal{R}_{e,m},$ defined as  $ \mathbf{c}\cdot \mathbf{d}=\sum\limits_{i=1}^{n}c_id_i$  for all $\mathbf{c}=(c_1,c_2,\ldots,c_n)$, $\mathbf{d}=(d_1,d_2,\ldots,d_n) \in \mathcal{R}_{e,m}^n.$ This bilinear form is  a non-degenerate and symmetric bilinear form  on $\mathcal{R}_{e,m}^n.$ 
  Given a linear code $\mathcal{C}$ of length $n$ over $\mathcal{R}_{e,m},$  its (Euclidean) dual code $\mathcal{C}^{\perp}$ is defined as $\mathcal{C}^{\perp}=\{\mathbf{w}\in \mathcal{R}_{e,m}^n ~:~ \mathbf{v} \cdot \mathbf{w} =0 ~\text{for all }\mathbf{v} \in \mathcal{C}\}.$ The dual code $\mathcal{C}^{\perp}$ is itself a linear code of  length $n$ over $\mathcal{R}_{e,m}$. By Theorem 3.10 of Norton and  S$\check{a}$l$\check{a}$gean \cite{Norton}, if the code $\mathcal{C}$ is of   type $\{\lambda_1,\lambda_2,\ldots, \lambda_{e}\},$ then its dual code $\mathcal{C}^{\perp}$   is of type $\{n-(\lambda_1+\lambda_2+\cdots+\lambda_e),\lambda_{e},\lambda_{e-1},\ldots,\lambda_2\}.$ 
  The code $\mathcal{C}$ is said to be self-orthogonal if it satisfies $\mathcal{C}\subseteq \mathcal{C}^{\perp},$ while  the code $\mathcal{C}$ is said to be  self-dual if it satisfies  $\mathcal{C}= \mathcal{C}^{\perp}$.    
 The following lemma establishes a necessary and sufficient condition for a linear code of length $n$ over  $\mathcal{R}_{e,m}$ to be self-orthogonal.
\vspace{-1mm}\begin{lemma}\cite{Dougherty, Y}\label{l2.2}   Let  $e \geq 2$ and $n \geq 1$ be integers, and let $\lambda_1,\lambda_2,\ldots,\lambda_{e+1}$ be non-negative integers satisfying $\lambda_1+\lambda_2+\cdots+\lambda_{e+1}=n.$   
Let $\mathcal{C}$ be a linear code of  type $\{\lambda_1,\lambda_2,\ldots,\lambda_e\}$ and  length $n$ over $\mathcal{R}_{e,m}$ with a  generator matrix $G$ in the standard form \eqref{e1.1}. Then, the  code $\mathcal{C}$ is   self-orthogonal  if and only if 
	\vspace{-2mm}\begin{equation*}\label{e2.3a}
	T_i^{(e)}T_j^{(e)t}\equiv 0 \pmod{u^{e-i-j+2}} \text{ ~~for } 1\leq i\leq j\leq e \text{ and }  i+j\leq e+1,
	\vspace{-1mm}\end{equation*}
 where $T_i^{(e)}$ are the block matrices as defined in \eqref{e1.1} and $(\cdot)^t$ denotes the matrix transpose. Moreover,  the code $\mathcal{C}$  is   self-dual  if and only if  the code $\mathcal{C}$ is  self-orthogonal and  $\lambda_i=\lambda_{e-i+2}$ for $1 \leq i \leq e.$ 
		\end{lemma}

The set $\mathcal{T}_{m}^n$ consisting of all $n$-tuples over $\mathcal{T}_{m}$ can be viewed  as an $n$-dimensional vector space over the finite field $\mathcal{T}_{m},$ where vector addition is defined component-wise using the operation $\oplus$ 
and scalar multiplication is also defined component-wise.  Define a map $B_{m}:\mathcal{T}_{m}^n \times \mathcal{T}_{m}^n \rightarrow \mathcal{T}_{m}$  as $B_{m}(\mathbf{a},\mathbf{b})=\pi_0(\mathbf{a}\cdot \mathbf{b})=\pi_0(\sum\limits_{i=1}^{n}a_ib_i)$ for all $\mathbf{a}=(a_1,a_2,\ldots,a_n),$ $\mathbf{b}=(b_1,b_2,\ldots,b_n) \in \mathcal{T}_{m}^n.$  It is straightforward to verify that the map $B_m$  is a symmetric and non-degenerate  bilinear form on $\mathcal{T}_{m}^n.$  A linear code $\mathcal{D}$ of length $n$ over $\mathcal{T}_{m}$ is defined as a $\mathcal{T}_{m}$-linear subspace of $\mathcal{T}_{m}^n.$ The dual code of $\mathcal{D}$ with respect to the bilinear form $B_m$ is given by $\mathcal{D}^{\perp_{B_{m}}}=\{\mathbf{b} \in \mathcal{T}_{m}^n : B_{m}(\mathbf{b},\mathbf{d})=0 \text{ for all } \mathbf{d} \in \mathcal{D}\},$ which itself is a linear code of length $n$ over $\mathcal{T}_{m}.$
  Further, a linear code $\mathcal{D}$ of length $n$ over $\mathcal{T}_{m}$ is said to be (i) self-orthogonal if   $\mathcal{D}\subseteq \mathcal{D}^{\perp_{B_{m}}},$ and (ii) self-dual if   $\mathcal{D}= \mathcal{D}^{\perp_{B_{m}}}.$

  Let $\mathcal{C}$ be a linear code of length $n$ over $\mathcal{R}_{e,m}.$ For  $1 \leq i \leq e,$  the $i$-th torsion code  of $\mathcal{C}$  is defined as
\begin{equation*} Tor_i(\mathcal{C})=\{\mathbf{c} \in \mathcal{T}_m^n: u^{i-1} (\mathbf{c}+u\mathbf{c}^{\prime})\in \mathcal{C} \text{ for some } \mathbf{c}^{\prime}\in \mathcal{R}_{e,m}^n \}.\end{equation*} Each $Tor_i(\mathcal{C})$ is a linear code of length $n$ over $\mathcal{T}_{m}.$  If  $\mathcal{C}$ admits a generator matrix $G$ in the standard form \eqref{e1.1}, then the $i$-th torsion code $Tor_i(\mathcal{C})$  has dimension  $\lambda_1+\lambda_2+\cdots + \lambda_i$  over $\mathcal{T}_m$ and has a generator matrix 
 \vspace{-2mm}\begin{equation}\vspace{-2mm}\label{e1.2}
	\begin{aligned}
	\begin{bmatrix}
	\mathtt{I}_{\lambda_1}& \mathtt{A}_{1,1}^{(0)}& \mathtt{A}_{1,2}^{(0)}&  \cdots& \mathtt{A}_{1,i-1}^{(0)}& \cdots&\mathtt{A}_{1,e-1}^{(0)}&\mathtt{A}_{1,e}^{(0)}\\
	\mathbf{0}&\mathtt{I}_{\lambda_2}& \mathtt{A}_{2,2}^{(0)}&\cdots & \mathtt{A}_{2,i-1}^{(0)}& \cdots&\mathtt{A}_{2,e-1}^{(0)}&\mathtt{A}_{2,e}^{(0)}\\
	\vdots&\vdots&\vdots&\vdots& \vdots& \vdots &\vdots&\vdots\\
	\mathbf{0}&\mathbf{0}&\mathbf{0}&\cdots &\mathtt{I}_{\lambda_i}&\cdots &\mathtt{A}_{i,e-1}^{(0)}&\mathtt{A}_{i,e}^{(0)}
	\end{bmatrix}.
	\end{aligned}
	\end{equation}
	 It is important to note that these torsion codes form a nested chain:  \vspace{-1mm}\begin{equation*}\vspace{-1mm}
	    Tor_1(\mathcal{C})\subseteq Tor_{2}(\mathcal{C})\subseteq \cdots \subseteq Tor_{e}(\mathcal{C})\end{equation*}  and  the size of the code $\mathcal{C}$ satisfies the relation $|\mathcal{C}|=\prod\limits_{i=1}^{e}|Tor_i(\mathcal{C})|.$
	
To enumerate all self-orthogonal and self-dual codes over $\mathcal{R}_{e,m}$ 	
  of a given type, we need the following lemma. 
		\begin{lemma}\cite{Dougherty, Y}\label{l1.2} Let $\mathcal{C}$ be  a  self-orthogonal code  of length $n$ over $\mathcal{R}_{e,m}.$ The following hold.
\begin{enumerate}[(a)]
\vspace{-0.5mm}	\item 	\label{l1.2a}$Tor_i(\mathcal{C})\subseteq Tor_i(\mathcal{C})^{\perp_{B_m}}~\text{for }1\leq i \leq \floor{\frac{e+1}{2}}.$
				\vspace{-1mm}\item  	\label{l1.2b}$Tor_i(\mathcal{C})\subseteq Tor_{e-i+1}(\mathcal{C})^{\perp_{B_m}}~\text{for } \floor{\frac{e+1}{2}}+1\leq i \leq e.$
					\vspace{-1mm}\end{enumerate}
			In particular, if the code $\mathcal{C}$ is self-dual, then  we have $Tor_i(\mathcal{C})=Tor_{e-i+1}(\mathcal{C})^{\perp_{B_m}}$ for $\ceil{\frac{e+1}{2}}\leq i\leq e.$ (Throughout this paper,  $\floor{\cdot }$ and $\ceil{\cdot }$ denote the floor and ceiling functions, respectively.)
	\end{lemma}
\begin{remark}\label{R2.1}By Lemma \ref{l1.2}, if $\mathcal{C}$ is a self-orthogonal code of type $\{\lambda_1, \lambda_2, \ldots, \lambda_e\}$ and length $n$ over $\mathcal{R}_{e,m}$, then for $\ceil{\frac{e+1}{2}} \leq i \leq e$, the inequality
\vspace{-1.5mm}\begin{equation*}\vspace{-1.5mm}
   2(\lambda_1 + \cdots + \lambda_{e - i + 1}) + \lambda_{e - i + 2} + \cdots + \lambda_i \leq n 
\end{equation*}
holds. In particular, this implies that
\vspace{-1.5mm}\begin{equation*}\vspace{-1.5mm}
   n \geq 
\begin{cases}
2(\lambda_1 + \cdots + \lambda_{\frac{e}{2}}) + \lambda_{\frac{e}{2} + 1} & \text{if } e \text{ is even};\\[6pt]
2(\lambda_1 + \cdots + \lambda_{\frac{e+1}{2}}) & \text{if } e \text{ is odd}.
\end{cases} 
\end{equation*}
\end{remark} 
    
From now on, we assume, throughout this paper,  that  $p=2,$ $e \geq 2,$ and $\kappa$ is an even integer satisfying  $2 \leq \kappa\leq e$ and $2 \in \langle u^{\kappa}  \rangle \setminus \langle u^{\kappa+1}\rangle.$ Let us  write \vspace{-1.5mm}\begin{equation*}\vspace{-1.5mm}
    2 \equiv u^{\kappa}(\eta_0+u\eta_1+u^2\eta_2+\cdots+u^{e-1-\kappa}\eta_{e-1-\kappa})\pmod{u^{e}},
\end{equation*} where $\eta_0\in \mathcal{T}_m\setminus \{0\}$ and $\eta_1,\eta_2,\ldots,\eta_{e-1-\kappa}\in \mathcal{T}_{m}.$  For any positive integer $d,$ we define $\mathrm{f}_d=\floor[\big]{\frac{d}{2}},$ $\mathrm{c}_d=\ceil[\big]{\frac{d}{2}}$ and $\theta_d=\mathrm{c}_d-\mathrm{f}_d.$ It is clear that  $\theta_d= 0$ if $d$ is even, whereas $\theta_d =1$ if $d$ is odd. For notational convenience, we set $s=\mathrm{f}_e$ and $\kappa_1=\mathrm{f}_\kappa.$ 
Unless mentioned otherwise, we will adhere to the notations and conventions  introduced in Section \ref{prelim}.
\vspace{-2mm} \section{A recursive method for constructing and enumerating   self-orthogonal and self-dual codes over $\mathcal{R}_{e,m}$}\label{construction}
 In this section, we will first provide an example to illustrate that the recursive construction methods proposed in \cite{Galois, quasi} do not necessarily yield self-orthogonal codes over $\mathcal{R}_{e,m}.$ Subsequently, we will adopt the recursive technique established in \cite[Sec. 3]{YSub2} to construct and enumerate self-orthogonal and self-dual codes over $\mathcal{R}_{e,m},$  using certain chains of self-orthogonal codes over $\mathcal{T}_{m}.$

We begin by presenting an example to illustrate the limitations of the recursive techniques employed in \cite{Galois, quasi}.
    
\begin{example}\label{EX1}
     Let  $\mathcal{R}_{5,2}=\frac{GR(8,2)[x]}{\langle x^2+2,4x\rangle}$  be a finite commutative chain ring with  maximal ideal $\langle u \rangle,$ where $u:= x+\langle x^2+2, 4x \rangle \in \mathcal{R}_{5,2}.$ Note that  the maximal ideal $\langle u\rangle $ has  nilpotency index $5,$ and that $2\equiv  u^2+u^4 \pmod {u^5} .$ 

 Let $n=3,$ $\lambda_1=1$  and $\lambda_2=\lambda_3=\lambda_4=\lambda_5=0.$ Let $\mathcal{D}_0$ be a self-orthogonal  code of length $3$ and dimension $\lambda_1+\lambda_2+\lambda_3=1$ over $\mathcal{T}_2$ with a generator matrix $\begin{bmatrix}1& 1&0\end{bmatrix}.$   Observe that $\pi_{1}(\textbf{v}\cdot \textbf{v})=0$ for all $\textbf{v}\in \mathcal{D}_0,$ where each $\textbf{v}\cdot \textbf{v}$ is viewed as an element of $\mathcal{R}_{5,2}.$
 
 \begin{description}\item[(i)] We will first demonstrate that the code  $\mathcal{D}_0$ can not be lifted to a self-orthogonal code of type $\{1,0,0,0,0\}$ and length $3$ over  $\mathcal{R}_{5,2}$ using the recursive construction method provided in Section 4  of Yadav and Sharma \cite{quasi}. For this,  consider a linear code $\mathcal{C}_3$ of type $\{1,0,0\} $ and length $3$ over $\mathcal{R}_{3,2}$ with a generator matrix $\begin{bmatrix}
    1&1&0
\end{bmatrix}+u\begin{bmatrix}
    0&a_1&b_1
\end{bmatrix}+u^2\begin{bmatrix}
    0&a_2&b_2
\end{bmatrix},  $  where $a_1,b_1,a_2,b_2\in \mathcal{T}_2.$ 
According to the recursive construction provided in Section 4 of \cite{quasi},  the code $\mathcal{D}_0$ can be lifted to a self-orthogonal code of  type $\{1,0,0,0,0\}$ over $\mathcal{R}_{5,2}$ if and only if $\mathcal{C}_3$ is a   self-orthogonal code   over $\mathcal{R}_{3,2}$ satisfying   property $(\ast)$ (see \cite[Sec. 4]{quasi}), which holds if and only if  $a_1^2+b_1^2\equiv 1\pmod u,$ $a_1^2+b_1^2\equiv 0\pmod u$ and $a_2^2+b_2^2\equiv 0\pmod u.$ It is straightforward to see that no choice of 
$a_1,b_1,a_2,b_2\in \mathcal{T}_2$
  can satisfy these conditions simultaneously, and hence  $\mathcal{C}_3$ is not a self-orthogonal code  over $\mathcal{R}_{3,2}$ satisfying property $(\ast)$ for any choice of $a_1,b_1,a_2,b_2\in \mathcal{T}_2.$  This  illustrates that the recursive construction method provided in Yadav and Sharma \cite{quasi} does not always lift  a self-orthogonal code over $\mathcal{T}_m$ to a self-orthogonal code over $\mathcal{R}_{e,m}.$ 

\item[(ii)] We will next demonstrate that the code $\mathcal{D}_0$ can not be lifted to a self-orthogonal code of type $\{1,0,0,0,0\}$ and length $3$ over  $\mathcal{R}_{5,2},$ using the recursive construction method provided in Section 4  of Yadav and Sharma \cite{Galois}. To proceed, we  recall, from \cite[Def. 2.1]{Galois}, that  a self-orthogonal code $\mathscr{C}$ over $\mathcal{T}_2$ is called doubly even if it satisfies $\pi_1(\textbf{c}\cdot \textbf{c})=0$ for all $\textbf{c}\in \mathscr{C},$ where  each $\textbf{c}\cdot \textbf{c}$ is viewed as an element of $\mathcal{R}_{5,2}$. 
Since  $\pi_1(\textbf{d}\cdot \textbf{d})=0$ for all $\textbf{d}\in \mathcal{D}_0,$  the code $\mathcal{D}_0$ is doubly even in  accordance with   \cite[Def. 2.1]{Galois}.   Now, consider a linear code $\mathcal{D}_3$ of type $\{1,0,0\} $ and length $3$ over $\mathcal{R}_{3,2}$ with a generator matrix $\begin{bmatrix}
    1&1&0
\end{bmatrix}+u\begin{bmatrix}
    0&c_1&d_1
\end{bmatrix}+u^2\begin{bmatrix}
    0&c_2&d_2
\end{bmatrix},  $  where $c_1,d_1,c_2,d_2\in \mathcal{T}_2.$ According to the recursive construction method given in Section 4 of Yadav and Sharma \cite{Galois},  the code $\mathcal{D}_3$ can be lifted to a self-orthogonal code of type $\{1,0,0,0,0\}$ over $\mathcal{R}_{5,2}$ if and only if $\mathcal{D}_3$ is a  $\lambda_1$-doubly even and self-orthogonal code   over $\mathcal{R}_{3,2}$ (see \cite[Def. 2.1]{Galois}), which  holds  if and only if $c_1=0$ and $d_1=1.$  In particular, for $c_1=c_2=0$ and  $d_1=d_2=1,$ we observe that  the code $\mathcal{D}_3$  with a generator matrix $\begin{bmatrix}
    1&1&u+u^2
\end{bmatrix}$ is  $\lambda_1$-doubly even and self-orthogonal. Corresponding to this particular choice of the code $\mathcal{D}_3$, consider a linear code $\mathcal{D}_5$ of type $\{1,0,0,0,0\}$ and length $3$ over $\mathcal{R}_{5,2}$ with a generator matrix 
$$\begin{bmatrix}
    1&1&u+u^2
\end{bmatrix}+u^3\begin{bmatrix}
0&c_3&d_3\end{bmatrix}+u^4\begin{bmatrix}
    0&c_4&d_4
\end{bmatrix},  $$  where $c_3,d_3,c_4,d_4\in \mathcal{T}_2.$  
It is easy to see that  the code $\mathcal{D}_5$ is not self-orthogonal for any choice of $c_3,d_3,c_4,d_4\in \mathcal{T}_2.$   This  illustrates that the recursive construction method provided in Section 4 of Yadav and Sharma \cite{Galois}  does not always lift a doubly even code over 
$\mathcal{T}_m$ to a self-orthogonal code over $\mathcal{R}_{e,m}.$
\end{description}
\end{example}

We will now employ the recursive technique proposed in Yadav and Sharma \cite{YSub2} to construct and enumerate self-orthogonal and self-dual codes of length $n$ over $\mathcal{R}_{e,m},$ using certain chains of self-orthogonal codes over $\mathcal{T}_m.$  Throughout this section, we assume, by Remark \ref{R2.1}, that   $n$ is a positive integer, and that   $\lambda_1, \lambda_2, \ldots,  \lambda_{e+1}$ are non-negative integers satisfying $\lambda_1+\lambda_2+\cdots+\lambda_{e+1}=n$ and  $2\lambda_1+2\lambda_2+\cdots+2\lambda_{e-i+1}+\lambda_{e-i+2}+\cdots+\lambda_i \leq n$ for $s+1\leq i\leq e.$       Define $\Lambda_0=0$  and 
$\Lambda_i=\lambda_1+\lambda_2+\cdots+\lambda_i $   for $1\leq i \leq e+1 .$  

Throughout this paper, for  positive integers  $a$ and $b \leq e,$ let  $(\mathtt{R})_{a,b} $ denote the block matrix whose  $(i,j)$-th block is  the matrix $\mathtt{R}_{i,j}\in \mathcal{M}_{\lambda_{i}\times \lambda_{j+1}}(\mathcal{T}_{m})$   for $1 \leq i \leq a$ and $b \leq j \leq e.$  For any square matrix $Y,$
 let $\mathcal{D}iag(Y)$
   denote the diagonal matrix whose principal diagonal entries coincide with those of  the matrix $Y.$ For any positive integer $z,$ let  $[B]_z$  denote the column block matrix whose $j$-th block is  $B_j$ for $1 \leq j \leq z.$ 

We further recall, from Section \ref{prelim}, that $(\mathcal{T}_m,\oplus,\cdot)$ is the finite field of order $2^m.$  Now, let $\mathcal{C}$ be a  self-orthogonal code of type $\{\lambda_1,\lambda_2,\ldots, \lambda_e\}$ and length $n$ over $\mathcal{R}_{e,m}$.   By Lemma \ref{l1.2},  the torsion code $Tor_{s+\theta_e}(\mathcal{C})$ is also a  self-orthogonal code of length $n$  and  dimension $\Lambda_{s+\theta_e}$ over $\mathcal{T}_{m}(=\overline{\mathcal{R}}_{e,m}).$ We further recall the following important observation:
\vspace{-2mm}\begin{remark}\label{r3.1}\cite[Rem. 4.1]{Galois} and \cite[Rem. 4.1]{quasi}
Every  self-orthogonal code of length $n$ and dimension $\Lambda_{s+\theta_e}$ over  $\mathcal{T}_{m}$ is permutation equivalent to a self-orthogonal code with a generator matrix  in the form
\vspace{-1mm}\begin{equation*}\vspace{-1mm}
\mathtt{G}=\begin{bmatrix}
T_1^{(0)}\\T_2^{(0)}\\\vdots\\T_{s+\theta_e}^{(0)}
\end{bmatrix}   =\begin{bmatrix}
\mathtt{I}_{\lambda_1}&\mathtt{A}_{1,1}^{(0)}&\mathtt{A}_{1,2}^{(0)}&\cdots&\mathtt{A}_{1,s+\theta_e-1}^{(0)}&\mathtt{A}_{1,s+\theta_e}^{(0)}&\cdots &\mathtt{A}_{1,e}^{(0)}\\
 \mathbf{0} &\mathtt{I}_{\lambda_2} &\mathtt{A}_{2,2}^{(0)} &\cdots &\mathtt{A}_{2,s+\theta_e-1}^{(0)}&\mathtt{A}_{2,s+\theta_e}^{(0)}&\cdots &\mathtt{A}_{2,e}^{(0)}\\
	\vdots& \vdots &\vdots &\vdots &\vdots &\vdots& \vdots&\vdots\\
 \mathbf{0}& \mathbf{0}& \mathbf{0}&\cdots&  \mathtt{I}_{\lambda_{s+\theta_e}}&\mathtt{A}_{s+\theta_e,s+\theta_e}^{(0)}& \cdots & \mathtt{A}_{s+\theta_e,e}^{(0)}
\end{bmatrix} , \end{equation*} 
where the columns of $\mathtt{G}$ are grouped into blocks of sizes $\lambda_1,\lambda_2,\ldots, \lambda_e$ and $ \lambda_{e+1}.$ Here,  $\mathtt{I}_{\lambda_i}$ denotes the $\lambda_i\times \lambda_i$ identity matrix over $\mathcal{T}_{m}$  and  $\mathtt{A}_{i,j}^{(0)} \in \mathcal{M}_{\lambda_i\times \lambda_{j+1}}(\mathcal{T}_{m})$  for $1 \leq i \leq s+\theta_e$ and $i \leq j \leq e.$ Additionally, each of the matrices 
$(\mathtt{A}^{(0)})_ { s,s+\theta_e+1}$, $(\mathtt{A}^{(0)})_{ s-1,s+\theta_e+2},\ldots, (\mathtt{A}^{(0)})_{2,e-1},$  $\mathtt{A}_{1,e}^{(0)}$ has   full row-rank. \end{remark}

For $ 1\leq \ell \leq e$ and  $0 \leq i \leq \ell-1,$ we recall that the map $\pi_i: \mathcal{R}_{\ell,m} \rightarrow \mathcal{T}_{m}$ is defined as $\pi_i(d)=d_i$ for all  $d=d_0+ud_1+\cdots+u^{\ell-1}d_{\ell-1}\in \mathcal{R}_{\ell,m} ,$ where $d_0,d_1,\ldots,d_{\ell-1} \in \mathcal{T}_{m}.$ Now, if  $\mathcal{C}$ is a self-orthogonal code of length $n$ over $\mathcal{T}_m,$ one can easily see, for all $\textbf{v}\in \mathcal{C},$ that $\pi_{0}(\textbf{v}\cdot \textbf{v})=0$ and $\pi_{2j+1}(\textbf{v}\cdot \textbf{v})=0$ for $ 0\leq j \leq \kappa_1-1,$  where $\kappa_1=\mathrm{f}_\kappa=\frac{\kappa}{2}$ and each $\textbf{v}\cdot \textbf{v}$ is viewed as an element of $\mathcal{R}_{e,m}.$ We next recall the definition of a special class of linear codes of length $n$ over $\mathcal{R}_{\ell,m},$ referred to as linear codes satisfying   property $(\mathfrak{P}),$ as introduced in \cite[Sec. 3]{YSub2} for all $e \geq 3$ and $2 \leq \ell \leq e.$ 

\noindent \textbf{Linear codes over $\mathcal{R}_{\ell,m}$ satisfying  property $(\mathfrak{P})$:} \textit{Let $e, \ell $ be integers such that $e \geq 3 ,$ $2 \leq \ell\leq  e$ and $\ell \equiv e\pmod2.$ Define $\gamma_{\ell}=s-\mathrm{f}_\ell.$   Let $\mathcal{C}_{\ell}$ be a linear code of type $\{\Lambda_{\gamma_{\ell}+1},\lambda_{\gamma_{\ell}+2}, \ldots, \lambda_{\gamma_{\ell}+\ell}\}$ and length $n$ over $\mathcal{R}_{\ell,m}$ with a generator matrix 
\vspace{-2mm}\begin{equation}\label{e3.0}
\mathcal{G}_{\ell}=\begin{bmatrix}
T_1^{(\ell)}\\T_2^{(\ell)} \\ \vdots\\ T_{\gamma_{\ell}+1}^{(\ell)}\vspace{0.5mm}\\ uT_{\gamma_{\ell}+2}^{(\ell)}\\ \vdots\\u^{\ell-1}T_{\gamma_{\ell}+\ell}^{(\ell)}
\end{bmatrix},\end{equation} 
where $T_i^{(\ell)}\in \mathcal{M}_{\lambda_i\times n}(\mathcal{R}_{\ell,m})$ for $1 \leq i \leq \gamma_{\ell}+1$ and  the matrix $T_{\gamma_{\ell}+j}^{(\ell)}\in \mathcal{M}_{\lambda_{\gamma_{\ell}+j}\times n}(\mathcal{R}_{\ell,m})$ is to be considered modulo $u^{\ell-j+1}$ for  $2 \leq j \leq \ell.$ 
\begin{enumerate}
    \item[(i)] When $\ell \leq \min\{\kappa-1,e-\kappa\} ,$ we say that the code $\mathcal{C}_{\ell}$ satisfies  property $(\mathfrak{P})$ if 
   \hspace{-3mm} \begin{eqnarray*}
    \displaystyle \mathcal{D}iag\left([T^{(\ell)}]_{\gamma_{\ell}+1-\mathrm{f}_i}[T^{(\ell)}]_{\gamma_{\ell}+1-\mathrm{f}_i}^t\right) &\equiv & \mathbf{0} \pmod{u^{\ell+i}} \text{ ~~for ~} i\in \{2,4,6,\ldots,\ell-\theta_e\}~\text{ and}\\ 
   \hspace{-2mm} \displaystyle  \pi_{\ell+j}\big( \mathcal{D}iag\big([T^{(\ell)}]_{\gamma_{\ell}+1-\mathrm{f}_{j+2}}[T^{(\ell)}]_{\gamma_{\ell}+1-\mathrm{f}_{j+2}}^t\big) \big)&=& \mathbf{0} \text{ ~for~~ } j \in \{ \ell-1,\ell+1,\ldots, \kappa -1-\theta_e  \}.
    \end{eqnarray*}
     \item[(ii)] When $e-\kappa< \ell \leq \kappa-\mathrm{f}_{2\kappa-e}+1 ,$  we say that the code $\mathcal{C}_{\ell}$ satisfies  property $(\mathfrak{P})$ if 
     \begin{eqnarray*}
   \displaystyle \mathcal{D}iag\left([T^{(\ell)}]_{\gamma_{\ell}+1-\mathrm{f}_i}[T^{(\ell)}]_{\gamma_{\ell}+1-\mathrm{f}_i}^t\right) &\equiv &\mathbf{0} \pmod{u^{\ell+i}} \text{ ~~~for ~} i\in \{2,4,6,\ldots,\ell-\theta_e\}\text{ and}\\ 
     \displaystyle  \pi_{\ell+j}\left( \mathcal{D}iag\big([T^{(\ell)}]_{\gamma_{\ell}+1-\mathrm{f}_{j+2}}[T^{(\ell)}]_{\gamma_{\ell}+1-\mathrm{f}_{j+2}}^t\big)\right)&=& \mathbf{0} \text{ ~ for~ } j \in \{ \ell-1,\ell+1,\ell+3,\ldots, e-\ell-1-\theta_e\}.
    \end{eqnarray*}
     \item[(iii)] When $\kappa \leq \ell \leq e-\kappa,$ we say that the code $\mathcal{C}_{\ell}$ satisfies  property $(\mathfrak{P})$ if 
    \begin{eqnarray*}
   \displaystyle  \mathcal{D}iag\left([T^{(\ell)}]_{\gamma_{\ell}+1-\mathrm{f}_i}[T^{(\ell)}]_{\gamma_{\ell}+1-\mathrm{f}_i}^t\right) &\equiv &\mathbf{0} \pmod{u^{\ell+i}} \text{ ~~for ~} i\in \{2,4,6,\ldots,\kappa\}.
    \end{eqnarray*}
    \item[(iv)] When $\ell > \max\{e-\kappa,\kappa-\mathrm{f}_{2\kappa-e}+1\},$  we say that the code $\mathcal{C}_{\ell}$ satisfies  property $(\mathfrak{P})$ if 
    \begin{eqnarray*}
   \displaystyle  \mathcal{D}iag\left([T^{(\ell)}]_{\gamma_{\ell}+1-\mathrm{f}_i}[T^{(\ell)}]_{\gamma_{\ell}+1-\mathrm{f}_i}^t\right) &\equiv& \mathbf{0} \pmod{u^{\ell+i}} \text{ ~~for ~} i\in \{2,4,6,\ldots,e-\ell\}.
    \end{eqnarray*}(Here, all matrices are considered  over $\mathcal{R}_{e,m}$ and all matrix multiplications are performed over $\mathcal{R}_{e,m}.$) 
    \end{enumerate}}

It is worth noting that when $\mathfrak{s}=1,$ 
linear codes over $\mathcal{R}_{\ell,m}$ satisfying property $(\mathfrak{P})$  match with linear codes over $\mathcal{R}_{\ell,m}$ satisfying property $(\ast)$  as defined in  \cite[Sec. 4]{quasi}.  
    
   Next, let $\mathcal{C}_e$ be a linear code of type $\{\lambda_1,\lambda_2,\ldots,\lambda_e\}$ and length $n$ over $\mathcal{R}_{e,m}$ with a generator matrix $\mathcal{G}_e$ as defined in \eqref{e3.0} for the case  $\ell=e,$    where 
    the block matrix  $T^{(e)}_j \in \mathcal{M}_{\lambda_j\times n }(\mathcal{R}_{e,m})$ is of the form  $T^{(e)}_j=T^{(0)}_j+u\mathtt{U}^{(1)}_j+u^2\mathtt{U}^{(2)}_j+\cdots+u^{e-j}\mathtt{U}^{(e-j)}_j$ with $T^{(0)}_j,\mathtt{U}^{(1)}_j,\mathtt{U}^{(2)}_j,\ldots, \mathtt{U}^{(e-j)}_j\in \mathcal{M}_{\lambda_j\times n}(\mathcal{T}_m)$
     for $1 \leq j \leq e.$ Now, corresponding to the code $\mathcal{C}_e,$     we consider, for each integer $\ell$ satisfying $2 \leq \ell \leq e-2$ and  $\ell\equiv e\pmod 2,$ a linear   code $\mathcal{C}_{\ell}$ of  type $\{\Lambda_{\gamma_{\ell}+1},\lambda_{\gamma_{\ell}+2}, \ldots, \lambda_{\gamma_{\ell}+\ell}\}$ and length $n$ over $\mathcal{R}_{\ell,m}$    with a generator matrix $\mathcal{G}_{\ell}$  (as defined in \eqref{e3.0}),
whose block matrices are of the forms $T_i^{(\ell)} =T^{(0)}_i+u\mathtt{U}_i^{(1)}+\cdots+ u^{\ell-1}\mathtt{U}_i^{(\ell-1)}$ \big(\textit{i.e.,} $T_i^{(\ell)}\equiv T_i^{(e)} \pmod{ u^{\ell}}$\big) for $1 \leq i\leq \gamma_{\ell}+1$ and  $T_{\gamma_{\ell}+j}^{(\ell)}=T^{(0)}_{\gamma_{\ell}+j}+u\mathtt{U}_{\gamma_{\ell}+j}^{(1)}+\cdots+ u^{\ell-j}\mathtt{U}_{\gamma_{\ell}+j}^{(\ell-j)}$ \big(\textit{i.e.,} $T_{\gamma_{\ell}+j}^{(\ell)}\equiv T_{\gamma_{\ell}+j}^{(e)} \pmod{ u^{\ell-j+1}}$\big)  for   $2 \leq j \leq \ell.$ We will refer to $\mathcal{C}_{\ell}$ as the linear code over $\mathcal{R}_{\ell,m}$ corresponding to the code $\mathcal{C}_e.$ We will also say that the code $\mathcal{C}_{e}$ is a lift of the code $\mathcal{C}_{\ell}.$    
Now, the following lemma provides a necessary condition under which a linear code over $\mathcal{R}_{\ell,m}$
can be lifted to a self-orthogonal or  self-dual code over $\mathcal{R}_{e,m},$ where $e \geq 4$ and $\ell$ is an integer satisfying $2 \leq \ell \leq e-2$ and $\ell \equiv e \pmod 2.$  
\begin{lemma}\label{l3.2} Let $e \geq 4$ be an integer, and  let $\mathcal{C}_e$  be a linear code 
of  type $\{\lambda_1,\lambda_2,\ldots,\lambda_e\}$ and length $n$ over $\mathcal{R}_{e,m}$ with a generator matrix $\mathcal{G}_{e}$ as defined in \eqref{e3.0}. For   each integer $\ell$ satisfying $2 \leq \ell \leq e-2$ and $\ell\equiv e \pmod 2,$ let $\mathcal{C}_{\ell}$ be  the linear code 
 over $\mathcal{R}_{\ell,m}$  corresponding to the code  $\mathcal{C}_e.$ If the code $\mathcal{C}_e$ is self-orthogonal (\textit{resp.} self-dual), then the code $\mathcal{C}_{\ell}$  is  also self-orthogonal  (\textit{resp.} self-dual) and satisfies   property $(\mathfrak{P}).$

    \end{lemma}
 \begin{proof} Working  as in Lemma 3.1 of Yadav and Sharma \cite{YSub2}, the desired result follows. \end{proof}

From the above lemma, it follows,  for  $e \geq 4$ and an integer $\ell$ satisfying $2 \leq \ell \leq e-2$ and  $\ell \equiv e\pmod 2 ,$ that a necessary condition for lifting a linear code $\mathcal{C}_{\ell}$ over $\mathcal{R}_{\ell,m}$ to a self-orthogonal (\textit{resp.} self-dual) code $\mathcal{C}_{e}$ over $\mathcal{R}_{e,m}$
 is that the code $\mathcal{C}_{\ell}$
 must be self-orthogonal  (\textit{resp.} self-dual) and  satisfies   property $(\mathfrak{P}).$ Throughout this paper, let $Sym_j(\mathcal{T}_m)$ and $Alt_j(\mathcal{T}_m)$ denote the sets consisting of all $j \times j$ symmetric and alternating matrices over $\mathcal{T}_m,$ respectively. We also need the following lemma.
\begin{lemma}\label{R=2}For an even integer $e\geq 4,$  let us define \vspace{-2mm}\begin{equation*}\label{e5.0a}
[T^{(0)}]_s=\begin{bmatrix}
T_1^{(0)}\\T_2^{(0)}\\\vdots\\T_{s}^{(0)}
\end{bmatrix}  =\begin{bmatrix}
\mathtt{I}_{\lambda_1}&\mathtt{A}_{1,1}^{(0)}& \mathtt{A}_{1,2}^{(0)}&\cdots&\mathtt{A}_{1,s-1}^{(0)}&\cdots & \mathtt{A}^{(0)}_{1,e-1}&\mathtt{A}_{1,e}^{(0)}\vspace{0.5mm}\\
\mathbf{0} &\mathtt{I}_{\lambda_2} & \mathtt{A}_{2,2}^{(0)} &\cdots &\mathtt{A}_{2,s-1}^{(0)}&\cdots & \mathtt{A}_{2,e-1}^{(0)}& \mathtt{A}_{2,e}^{(0)}\\
	\vdots& \vdots  &\vdots &\vdots&\vdots &\vdots& \vdots&\vdots\\
\mathbf{0}&\mathbf{0}&\mathbf{0}&\cdots&\mathtt{I}_{\lambda_{s}}& \cdots & \mathtt{A}_{s,e-1}^{(0)}& \mathtt{A}_{s,e}^{(0)}
\end{bmatrix} \vspace{-1.5mm}\end{equation*} and \begin{equation*}\vspace{-1.5mm}[X]_{ s}=\begin{bmatrix}
X_1\\X_2\\ \vdots\\ X_{s}
\end{bmatrix}= \begin{bmatrix}
\mathbf{0}&\mathbf{0}&\mathtt{A}_{1,2}^{(1)}& \mathtt{A}_{1,3}^{(1)}&\cdots&\mathtt{A}_{1,s+1}^{(1)}&\cdots &  \mathtt{A}_{1,e}^{(1)}\\
\mathbf{0} & \mathbf{0}&\mathbf{0}&\mathtt{A}_{2,3}^{(1)}&  \cdots &\mathtt{A}_{2,s+1}^{(1)}&\cdots &\mathtt{A}_{2,e}^{(1)}\vspace{0.5mm}\\
	\vdots& \vdots  &\vdots &\vdots &\vdots& \vdots&\vdots&\vdots\\
\mathbf{0}&\mathbf{0}&\mathbf{0}&\mathbf{0}&\cdots&\mathtt{A}_{s,s+1}^{(1)}&\cdots &\mathtt{A}_{s,e}^{(1)}\vspace{0.5mm}
\end{bmatrix},\end{equation*} 
where columns of the matrices $[T^{(0)}]_{s}$ and $[X]_s$  are grouped  into blocks of sizes $\lambda_1,\lambda_2,\ldots, \lambda_e $ and $\lambda_{e+1}, $  $\mathtt{I}_{\lambda_i}$ denotes the $\lambda_i\times \lambda_i$ identity matrix over $\mathcal{T}_{m},$  $\mathtt{A}_{i,j}^{(0)} \in \mathcal{M}_{\lambda_i\times \lambda_{j+1}}(\mathcal{T}_{m})$  for $1 \leq i \leq s$ and $i \leq j \leq e,$ $\mathtt{A}_{g,h}^{(1)} \in \mathcal{M}_{\lambda_g \times \lambda_{h+1}}(\mathcal{T}_{m})$ for $1 \leq g \leq s$ and $g < h \leq e,$ and   each of the matrices  $(\mathtt{A}^{(0)})_ { s,s+1} $, $(\mathtt{A}^{(0)})_{ s-1,s+2},\ldots , (\mathtt{A}^{(0)})_{2,e-1}$, $\mathtt{A}_{1,e}^{(0)}$ has full row-rank. Let us suppose that the matrix $[T^{(0)}]_s$ satisfies 
\vspace{-2mm}\begin{eqnarray*}[T^{(0)}]_s[T^{(0)}]_s^t & \equiv & \mathbf{0}\pmod u \text{ ~and}\\
~[T^{(0)}]_s[T^{(0)}]_s^t &\equiv &u\mathtt{H}^{(1)}+u^2\mathtt{H}^{(2)}+\cdots+u^{\kappa +1}\mathtt{H}^{(\kappa+1)} \pmod{u^{\kappa+2}},\end{eqnarray*} where  $\mathtt{H}^{(1)},\mathtt{H}^{(2)},\ldots,\mathtt{H}^{(\kappa+1)}  \in Sym_{\Lambda_s}(\mathcal{T}_{m})$  and $\mathtt{H}^{(1)},\mathtt{H}^{(3)},\mathtt{H}^{(5)},\ldots,\mathtt{H}^{(\kappa-1)}\in Alt_{\Lambda_s}(\mathcal{T}_{m}).$ Let $\mathtt{F}^{(2)}$ and $\mathtt{F}^{(\kappa+1)}$ be  matrices over $\mathcal{T}_m$ whose rows consist of the first $\Lambda_{s-1}$ and $\Lambda_{s-\kappa_1}$  rows of the matrices $\mathtt{H}^{(2)}$ and  $\mathtt{H}^{(\kappa+1)},$  respectively. Let us now consider the following system of matrix equations in the unknown matrix $[X]_s\in \mathcal{M}_{\Lambda_s\times n}(\mathcal{T}_m)$: 
\vspace{-2mm}\begin{eqnarray}\label{el1.0}
 \begin{split}
  \left.\begin{array}{rrl}
[T^{(0)}]_{s}[X]_{ s}^t+[X]_{ s}[T^{(0)}]_{s}^t& \equiv & \mathtt{H}^{(1)} \pmod{u},\\
~ [X]_{s-1}[X]_{s-1}^t &\equiv& \mathtt{F}^{(2)} \pmod u,\\
\eta_0[T^{(0)}]_{s-\kappa_1}[X]_{s-\kappa_1}^t &\equiv & \mathtt{F}^{(\kappa+1)}\pmod u.~~\end{array}\right\} 
 \end{split}\vspace{-4mm}
\end{eqnarray}
 The following hold.
\vspace{-1mm}\begin{itemize}
    \vspace{-1mm}\item[(a)] If $\textbf{1}$ does not belong to the $\mathcal{T}_m$-span of the rows of the matrix $[T^{(0)}]_{s-\kappa_1}$, then the system \eqref{el1.0} always has a solution. (Throughout this paper, $\textbf{1}$ denotes the all-one vector $(1,1,\ldots,1)\in \mathcal{T}_m^n.$) 
   \vspace{-1mm} \item[(b)]  On the other hand, if $\textbf{1}$  belongs to the $\mathcal{T}_m$-span of the rows of the matrix $[T^{(0)}]_{s-\kappa_1},$ then the system \eqref{el1.0} has a solution if and only if  one of the following conditions holds: (i)  $n \equiv 0 , 4 \pmod{8}$, or (ii) $n \equiv 2,6 \pmod{8},$ $\kappa = 2$ and $(\eta_0)^{3/2} = \eta_1$, or (iii) $n \equiv 2,6 \pmod{8},$  $\kappa \geq 4$ and $\eta_1 = 0.$ 
\end{itemize}
Furthermore, if the system \eqref{el1.0} has a solution, then the number of its solutions is given by
\vspace{-2mm} \begin{equation*}
\displaystyle (2^m)^{\sum\limits_{i=3}^{s+1}\lambda_i\Lambda_{i-2}+\Lambda_{s}(n-\Lambda_{s+1})- \Lambda_{s-1}-\Lambda_{s-\kappa_1}-\frac{\Lambda_{s}(\Lambda_{s}-1)}{2}+\omega},
\vspace{-2mm}\end{equation*}
where $\omega=1$ if   $\textbf{1}$ belongs to the $\mathcal{T}_m$-span of the rows of the matrix $[T^{(0)}]_{s-\kappa_1},$ while $\omega =0$  otherwise.
\end{lemma}
\begin{proof}To prove the result, let us suppose that  $[T^{(0)}]_s=(\textbf{a}_i)$ and $[X]_s=(\textbf{x}_j),$ where $\textbf{a}_i$\rq{}s and $\textbf{x}_j$\rq{}s denote the rows of the matrices $[T^{(0)}]_s$ and $[X]_s,$ respectively. Let $\mathtt{H}^{(z)}_{i,j}\in \mathcal{T}_{m}$ denote the $(i,j)$-th entry of the matrix  $\mathtt{H}^{(z)}$ for $1 \leq i, j \leq \Lambda_s,$ where $1 \leq z \leq \kappa+1.$ Note that $\mathtt{H}^{(1)},\mathtt{H}^{(3)},\mathtt{H}^{(5)},\ldots,\mathtt{H}^{(\kappa-1)}\in Alt_{\Lambda_s}(\mathcal{T}_{m}), $
 which implies that $\mathtt{H}^{(1)}_{i,i}=\mathtt{H}^{(3)}_{i,i}=\mathtt{H}^{(5)}_{i,i}=\cdots=\mathtt{H}^{(\kappa-1)}_{i,i}=0$ for $1\leq i \leq \Lambda_s .$  We next observe that the system \eqref{el1.0}  of  matrix equations  is  equivalent to the following system of equations in unknowns $ \textbf{x}_1, \textbf{x}_2,\ldots, \textbf{x}_{\Lambda_{s}}$ over $\mathcal{T}_{m}$:  \vspace{-2mm}\begin{eqnarray}
\textbf{a}_i\cdot \textbf{x}_j+\textbf{a}_j \cdot \textbf{x}_i  &\equiv& \mathtt{H}^{(1)}_{i,j} \pmod{u} \text{ ~~for } 1 \leq i <j \leq \Lambda_s,\label{e3.6}\\ 
\textbf{1}\cdot \textbf{x}_i & \equiv & (\mathtt{H}^{(2)}_{i,i})^{2^{m-1}} \pmod{u} \text{ ~ for } 1 \leq i \leq \Lambda_{s-1},\label{e3.7}\\
\textbf{a}_i \cdot \textbf{x}_i & \equiv & \eta_0^{-1}\mathtt{H}_{i,i}^{(\kappa+1)} \pmod{u} \text{ ~~for } 1 \leq i \leq \Lambda_{s-\kappa_1}. \label{e3.8}
\end{eqnarray} 
Further, for each integer $j$ satisfying $1 \leq j \leq \Lambda_{s},$ we see that there exists a unique integer $z_{j}$ satisfying $1 \leq z_{j}\leq s$ and $\Lambda_{z_{j}-1}+1 \leq j \leq \Lambda_{z_{j}}.$ Thus, the corresponding unknown vector $\textbf{x}_{j}$ can be written as $\textbf{x}_{j}=(\textbf{0}~ \textbf{x}_{j}^{n-\Lambda_{z_{j}+1}}),$ where $\textbf{0}$ denotes the zero vector of 
length $\Lambda_{z_{j}+1}$ and $\textbf{x}_{j}^{n-\Lambda_{z_{j}+1}}$ denotes the vector of length $n-\Lambda_{z_{j}+1}$ obtained   by removing the first $\Lambda_{z_{j}+1}$ coordinates of $\textbf{x}_{j}$. This implies, for $\Lambda_{z_{j}-1}+1 \leq j \leq \Lambda_{z_{j}}$, that the first $\Lambda_{z_{j}+1}$ coordinates of $\textbf{x}_{j}$ are zero, leaving  $n-\Lambda_{z_{j}+1}$   variables
 in $\textbf{x}_{j}.$  For $1 \leq j \leq \Lambda_{s},$ let $\widetilde{\textbf{x}}_{j}=\textbf{x}_{j}^{n-\Lambda_{z_{j}+1}}$ 
and  $\widetilde{\textbf{a}}_{j}=\textbf{a}_{j}^{n-\Lambda_{z_{j}+1}}$ be the  vectors of length $n-\Lambda_{z_{j}+1}$ obtained from $\textbf{x}_{j}$ and  $\textbf{a}_{j}$ by deleting their first $\Lambda_{z_{j}+1}$ coordinates, respectively. 
Consequently, the system   of equations  \eqref{e3.6}-\eqref{e3.8} is equivalent to the following system of equations in unknowns $\widetilde{\textbf{x}}_1, \widetilde{\textbf{x}}_2,\ldots, \widetilde{\textbf{x}}_{\Lambda_{s}}$ over $\mathcal{T}_{m}$:   
\vspace{-2mm}\begin{eqnarray}
\widetilde{\textbf{a}}_i\cdot \widetilde{\textbf{x}}_j+\widetilde{\textbf{a}}_j \cdot \widetilde{\textbf{x}}_i  &\equiv& \mathtt{H}^{(1)}_{i,j} \pmod{u} \text{ ~~for } 1 \leq i <j \leq \Lambda_s,\label{e3.9}\\ 
\widetilde{\textbf{1}}\cdot \widetilde{\textbf{x}}_i & \equiv & (\mathtt{H}^{(2)}_{i,i})^{2^{m-1}} \pmod{u} \text{ ~ for } 1 \leq i \leq \Lambda_{s-1},\label{e3.10}\\
\widetilde{\textbf{a}}_i \cdot \widetilde{\textbf{x}}_i   & \equiv & \eta_0^{-1}\mathtt{H}_{i,i}^{(\kappa+1)} \pmod{u} \text{ ~~for } 1 \leq i \leq \Lambda_{s-\kappa_1}, \label{e3.11}
\vspace{-2mm} \end{eqnarray}
 where $\widetilde{\textbf{1}}$ denotes the all-one vector having the same length as that of $\widetilde{\textbf{x}}_i$ for each $i.$ Now, the above system of equations \eqref{e3.9}-\eqref{e3.11} can be expressed as the following matrix equation:
\begin{equation}\label{e3.12}\mathcal{N}\begin{bmatrix}
	\widetilde{\textbf{x}}_1^{t}\\ \widetilde{\textbf{x}}_2^{t}\\ \widetilde{\textbf{x}}_3^{t}\\ \vdots\\ \widetilde{\textbf{x}}_{\Lambda_{s-\kappa_1}-1}^{t}\\ \widetilde{\textbf{x}}_{\Lambda_{s-\kappa_1}}^{t}\vspace{0.50mm}\\ \vdots\\ \widetilde{\textbf{x}}_{\Lambda_{s-1}+1}^{t} \\ \vdots\\ \widetilde{\textbf{x}}_{\Lambda_{s}-2}^{t}\\ \widetilde{\textbf{x}}_{\Lambda_{s}-1}^{t} \\ \vspace{0.75mm}\widetilde{\textbf{x}}_{\Lambda_{s}}^{t}
	\end{bmatrix}\equiv\begin{bmatrix}
	(\mathtt{H}_{1,1}^{(2)})^{2^{m-1}}\\ \vdots\\ (\mathtt{H}_{\Lambda_{s-1},\Lambda_{s-1}}^{(2)})^{2^{m-1}}\vspace{0.5mm}\\ \eta_0^{-1}\mathtt{H}^{(\kappa+1)}_{1,1}\vspace{-0.25mm}\\ \vdots\vspace{-0.25mm}\\ \eta_0^{-1}\mathtt{H}^{(\kappa+1)}_{\Lambda_{s-\kappa_1},\Lambda_{s-\kappa_1}}\vspace{0.5mm}\\\mathtt{H}^{(1)}_{1,2}\\\vdots\\\mathtt{H}^{(1)}_{1,\Lambda_s}\vspace{-0.25mm}\\ \vdots\vspace{-0.25mm} \\\mathtt{H}^{(1)}_{\Lambda_{s}-1,\Lambda_s}
	\end{bmatrix}\pmod u,\vspace{-2mm}\end{equation}\vspace{-2mm}
where the matrix $\mathcal{N}$ has order $\left( \Lambda_{s-1}+\Lambda_{s-\kappa_1}+\frac{\Lambda_{s}(\Lambda_{s}-1)}{2}\right) \times \left(\sum\limits_{i=3}^{s+1}\lambda_i\Lambda_{i-2}+\Lambda_{s}(n-\Lambda_{s+1})\right)$ and is given by \begin{equation*}
  \mathcal{N}=\begin{bmatrix}
	\widetilde{\textbf{1}} &&&\\
    	&\widetilde{\textbf{1}} &&&\\
	& & \ddots  &  \\
	& & & \widetilde{\textbf{1}}\\
    	& & &&\ddots  &  \\
    & & & && \widetilde{\textbf{1}}\\
	\widetilde{\textbf{a}}_1 &&&\\
    	&\widetilde{\textbf{a}}_2 &&&\\
	& & \ddots  &  \\
	& & &  \widetilde{\textbf{a}}_{\Lambda_{s-\kappa_1}}&\\
\widetilde{\textbf{a}}_{2}&\widetilde{\textbf{a}}_{1}&&& \\
\vdots&\vdots  & \vdots &\vdots\\
\widetilde{\textbf{a}}_{\Lambda_{s-1}}&&&&&\widetilde{\textbf{a}}_{1}&&& \\
	\vdots&\vdots  & \vdots &\vdots\\
\widetilde{\textbf{a}}_{\Lambda_{s}}&&&&&&\cdots&&\widetilde{\textbf{a}}_{1} \\
	&&&\vdots&\vdots  & \vdots&\cdots&\vdots&\vdots \\
    	
	& && &&&\cdots& \widetilde{\textbf{a}}_{\Lambda_{s}}& \widetilde{\textbf{a}}_{\Lambda_{s}-1}\\
	\end{bmatrix}.\vspace{-2mm}
\end{equation*}
Further, working as in Lemma 4.1 of Yadav and Sharma \cite{Galois}, we see that the rows of the matrix $\mathcal{N}$  are linearly dependent over $\mathcal{T}_{m}$  if and only if $\textbf{1}$ belongs to the $\mathcal{T}_{m}$-span of the rows of the matrix $[T^{(0)}]_{s-\kappa_1}.$
Accordingly, we will distinguish the following two cases: (i) $\textbf{1}$ does not belong to the $\mathcal{T}_{m}$-span of the rows of the matrix $[T^{(0)}]_{s-\kappa_1},$ and (ii) $\textbf{1}$ belongs to the $\mathcal{T}_{m}$-span of the rows of the matrix $[T^{(0)}]_{s-\kappa_1}.$
\begin{itemize}
\item[(i)] Let us first assume that  $\textbf{1}$ does not belong to the $\mathcal{T}_{m}$-span of the rows of the matrix $[T^{(0)}]_{s-\kappa_1}.$  Here, the  rows of the matrix $\mathcal{N}$ are linearly independent over $\mathcal{T}_{m},$  and thus  the matrix $\mathcal{N}$ has row-rank  $\Lambda_{s-1}+\Lambda_{s-\kappa_1}+\frac{\Lambda_s(\Lambda_s-1)}{2}.$ This implies that the matrix equation \eqref{e3.12} always has a solution. From this, it follows that there exists a matrix $[X]_{s} \in \mathcal{M}_{\Lambda_{s}\times n}(\mathcal{T}_{m})$ satisfying the system \eqref{el1.0} of matrix equations  and that the matrix $[X]_{s} $ has precisely 
\vspace{-1.5mm}\begin{equation*}\vspace{-1.5mm}
\displaystyle (2^m)^{\sum\limits_{i=3}^{s+1}\lambda_i\Lambda_{i-2}+\Lambda_{s}(n-\Lambda_{s+1})- \Lambda_{s-1}-\Lambda_{s-\kappa_1}-\frac{\Lambda_{s}(\Lambda_{s}-1)}{2}}
\end{equation*} 
distinct choices.
\item[(ii)] Next, let us assume that $\textbf{1}$ belongs to the $\mathcal{T}_{m}$-span of the rows of the matrix $[T^{(0)}]_{s-\kappa_1}.$ Here, $\textbf{1}$ belongs to the $\mathcal{T}_m$-span of the linearly independent vectors $\textbf{a}_1,\textbf{a}_2,\ldots, \textbf{a}_{\Lambda_{s-\kappa_1}}.$ Thus, there exist unique scalars $\beta_{1},\beta_{2},\ldots,\beta_{\Lambda_{s-\kappa_1}} \in \mathcal{T}_{m}$ such that $\beta_{1} \textbf{a}_{1}+\beta_{2} \textbf{a}_{2}+\cdots+\beta_{\Lambda_{s-\kappa_1}} \textbf{a}_{\Lambda_{s-\kappa_1}} \equiv  \textbf{1}\pmod u.$ 
We further note that all  rows of the matrix $\mathcal{N}$ except the first row are linearly independent over $\mathcal{T}_{m}.$ This implies that the row-rank of the matrix $\mathcal{N}$ is  $\Lambda_{s-\kappa_1}+\Lambda_{s-1}+\frac{\Lambda_s(\Lambda_s-1)}{2}-1.$ It is easy to see that  the matrix equation \eqref{e3.12} has a solution if and only if 
\vspace{-1.5mm}\begin{equation}\label{e3.13}\vspace{-1.5mm}
   \displaystyle   \sum\limits_{h=1}^{\Lambda_{s-\kappa_1}}\beta_h(\mathtt{H}^{(2)}_{h,h})^{2^{m-1}} +\eta_0^{-1}\sum\limits_{g=1}^{\Lambda_{s-\kappa_1}}\beta_g^2 \mathtt{H}^{(\kappa+1)}_{g,g} + \hspace{-4mm}\sum\limits_{~~~~~1\leq i< j \leq \Lambda_{s-\kappa_1}}\hspace{-4mm}\beta_i \beta_j \mathtt{H}_{i,j}^{(1)}~ \equiv ~0\pmod{u}.
 \end{equation}
We next assert that the equation \eqref{e3.13}  holds (\textit{i.e.,} the matrix equation \eqref{e3.12}  has a solution) if and only if one of the following three conditions holds: (i) $n\equiv 0,4\pmod 8,$ or (ii) $n\equiv 2,6 \pmod 8, $  $\kappa=2$ and $(\eta_0)^{\frac{3}{2}}=\eta_1,$  or (iii) $n\equiv 2,6 \pmod 8,$  $\kappa\geq 4$ and $\eta_1=0.$  To prove this assertion, let us suppose  that $\textbf{1} \equiv \sum\limits_{i=1}^{\Lambda_{s-\kappa_1}} \beta_i \textbf{a}_i+ u\textbf{d}_1+u^2\textbf{d}_2+\cdots+u^{\kappa+1} \textbf{d}_{\kappa+1} \pmod{u^{\kappa+2}}$ for some $\textbf{d}_1,\textbf{d}_2,\ldots,\textbf{d}_{\kappa+1} \in \mathcal{T}_{m}^n.$ As  $\textbf{1}\cdot \textbf{1} \equiv n \pmod{ u^{\kappa+2}},$ we see that
$\textbf{1}\cdot \textbf{1} \equiv 0 \pmod{ u^{\kappa+2}} $ if $ n\equiv 0,4\pmod 8,$ whereas 
$\textbf{1}\cdot \textbf{1} \equiv u^{\kappa}(\eta_0+u\eta_1) \pmod{ u^{\kappa+2}} $ if $ n\equiv 2,6\pmod 8.$
From this, we get  
\vspace{-1.5mm}\begin{equation}\label{eta_0}\vspace{-1.5mm}
\sum\limits_{h=1}^{\Lambda_{s-\kappa_1}}\beta_h^2 \mathtt{H}_{h,h}^{(2)}+\textbf{d}_1\cdot \textbf{d}_1  \equiv \epsilon_1 \eta_0 \pmod{u}  \end{equation}       and \vspace{-1.5mm}\begin{equation}\label{eta_1}\vspace{-1.5mm}
   \eta_0^{-1} \sum\limits_{g=1}^{\Lambda_{s-\kappa_1}}\beta_g^2 \mathtt{H}^{(\kappa+1)}_{g,g}+ \hspace{-5mm}\sum\limits_{~~~~~1\leq i< j \leq \Lambda_{s-\kappa_1}}\hspace{-5mm}\beta_i \beta_j \mathtt{H}_{i,j}^{(1)}+\textbf{1}\cdot \textbf{d}_1 \equiv \epsilon_2 \eta_0^{-1}\eta_1\pmod{u},  
 \end{equation} 
 where  \vspace{-1mm}\begin{equation*}\vspace{-1mm}\epsilon_1=\left\{\begin{array}{ll}0 & \text{if either }  n\equiv 0,4\pmod{8} \text{ or }n\equiv 2,6 \pmod 8 \text{ and }\kappa \geq 4;\\ 1 & \text{if }n\equiv 2,6 \pmod 8 \text{ and }\kappa=2\end{array}\right.\end{equation*}
 and \vspace{-1mm}\begin{equation*}\vspace{-1mm}
     \epsilon_2=\left\{\begin{array}{ll}0 & \text{if }  n\equiv 0,4\pmod{8} ;\\ 1 & \text{if }n\equiv 2,6 \pmod 8 .\end{array}\right.
 \end{equation*} 
By \eqref{eta_0} and \eqref{eta_1},  we see that  the matrix equation \eqref{e3.12}  has a solution if and only if either (i) $n\equiv 0,4\pmod 8,$ or (ii) $n\equiv 2,6 \pmod 8, $  $\kappa=2$ and $(\eta_0)^{\frac{3}{2}}=\eta_1,$ or (iii) $n\equiv 2,6 \pmod 8,$  $\kappa\geq 4$ and $\eta_1=0$ holds. Furthermore,  the number of solutions of the system \eqref{e3.12}, and hence the number of choices for the matrix  $[X]_s \in \mathcal{M}_{\Lambda_s \times n } (\mathcal{T}_{m})$ is given by 
 \vspace{-2mm}\begin{equation*}\vspace{-2mm}
\displaystyle (2^m)^{\sum\limits_{i=3}^{s+1}\lambda_i\Lambda_{i-2}+\Lambda_{s}(n-\Lambda_{s+1})- \Lambda_{s-1}-\Lambda_{s-\kappa_1}-\frac{\Lambda_{s}(\Lambda_{s}-1)}{2}+1}.
\end{equation*}  
\end{itemize}
This proves the lemma.
\vspace{-1mm}\end{proof}

Now, in the following proposition,  we assume that $e=\kappa=2,$  and  we consider a self-orthogonal code $\mathcal{D}^{(1)}$   of length $n$ and dimension $\lambda_1$ over $\mathcal{T}_m.$ 
Here, we establish the existence of a self-orthogonal code  $\mathcal{C}_2$ of  type $\{\lambda_{1},\lambda_{2}\}$ and length $n$ over $\mathcal{R}_{2,m}$ satisfying   $Tor_1(\mathcal{C}_{2}) =\mathcal{D}^{(1)}.$    We also count the choices for such a code $\mathcal{C}_2$ obtained from a given   self-orthogonal code $\mathcal{D}^{(1)}$ over $\mathcal{T}_m.$  
   \begin{proposition}\label{pe=k=2} Let $e=\kappa=2.$   Let $\mathcal{D}^{(1)}$  be a  self-orthogonal code of length $n$ and dimension $\lambda_1$  over $\mathcal{T}_{m}.$    The following hold.
 \begin{enumerate}\vspace{-2mm}\item[(a)]  There exists a  self-orthogonal code $\mathcal{C}_{2}$ of  type $\{\lambda_{1},\lambda_{2}\}$ and length $n$ over $\mathcal{R}_{2,m}$ such that  $Tor_1(\mathcal{C}_2)=\mathcal{D}^{(1)}.$    
\item[(b)] Furthermore,  each such code $\mathcal{D}^{(1)}$  over $\mathcal{T}_m$ gives rise to precisely 	\vspace{-1.5mm}\begin{equation*}\vspace{-1.5mm}
\displaystyle (2^m)^{\lambda_1(n-\Lambda_2)-\frac{\lambda_1(\lambda_1-1)}{2}}{\lambda_{2}+n-\Lambda_{2}-\Lambda_1 \brack \lambda_{2}}_{2^m}\vspace{-2mm}
\end{equation*} distinct   self-orthogonal codes $\mathcal{C}_{2}$ of  type $\{\lambda_{1},\lambda_{2}\}$ and length $n$ over $\mathcal{R}_{2,m}$  satisfying $Tor_1(\mathcal{C}_2)=\mathcal{D}^{(1)}.$    \end{enumerate} \end{proposition}
 \begin{proof} Working as in Theorem 4.1 of Yadav and Sharma \cite{Galois} and by applying Lemma 4.1 of Yadav and Sharma \cite{quasi},  the desired result follows immediately. 
     \end{proof}

In the following proposition, we assume that  $e$ and $ \kappa $ are even integers satisfying $2 \leq \kappa \leq e - 1,$ and we consider two nested self-orthogonal codes $ \mathcal{D}^{(s)} \supseteq \mathcal{D}^{(s - \kappa_1)} $ of length $ n $ over $\mathcal{T}_m $, with dimensions $\dim \mathcal{D}^{(s)} = \Lambda_s $ and $\dim \mathcal{D}^{(s - \kappa_1)} = \Lambda_{s - \kappa_1} .$ Under these assumptions, we establish the existence of a self-orthogonal code $ \mathcal{C}_2 $ of type $ \{ \Lambda_s, \lambda_{s+1} \} $ and length $ n $ over $ \mathcal{R}_{2,m} ,$ satisfying property $ (\mathfrak{P})$ and  $\operatorname{Tor}_1(\mathcal{C}_2) = \mathcal{D}^{(s)}.$ Furthermore, we  explicitly count  the choices of such codes $ \mathcal{C}_2$ that can be obtained from a given pair $ (\mathcal{D}^{(s)}, \mathcal{D}^{(s - \kappa_1)}) $ of nested self-orthogonal codes over $ \mathcal{T}_m .$

\begin{proposition}\label{p3.1}
  Let $e$ and $\kappa$ be  even integers satisfying $2 \leq \kappa \leq e-1.$  Let $\mathcal{D}^{(s)}\supseteq \mathcal{D}^{(s-\kappa_1)}$ be two self-orthogonal codes of length $n$  over $\mathcal{T}_m,$ where $\dim \mathcal{D}^{(s)}=\Lambda_s$ and   $\dim \mathcal{D}^{(s-\kappa_1)}=\Lambda_{s-\kappa_1}.$ 
 \begin{enumerate}\vspace{-2mm}\item[(a)]  There exists a self-orthogonal code $\mathcal{C}_{2}$ of type $\{\Lambda_{s},\lambda_{s+1}\}$ and length $n$ over $\mathcal{R}_{2,m}$ satisfying  property  $(\mathfrak{P})$  and $Tor_1(\mathcal{C}_2)=\mathcal{D}^{(s)},$ where  such a code exists: \vspace{-1mm}\begin{itemize}
     \item unconditionally, if $\textbf{1}\notin \mathcal{D}^{(s-\kappa_1)};$
 \item if and only if either (i)  $n \equiv 0 , 4 \pmod{8}$, or (ii) $n \equiv 2,6 \pmod{8},$ $\kappa = 2$ and $(\eta_0)^{3/2} = \eta_1$, or (iii) $n \equiv 2,6 \pmod{8},$  $\kappa \geq 4$ and $\eta_1 = 0,$
 when $\textbf{1}\in\mathcal{D}^{(s-\kappa_1)}.$ \end{itemize}  
\vspace{-2mm}\item[(b)] Furthermore,  each such pair   $(\mathcal{D}^{(s-\kappa_1)},\mathcal{D}^{(s)})$ of codes over $\mathcal{T}_m$ gives rise to precisely\vspace{-1.5mm}\begin{equation*}\vspace{-1.5mm}
\displaystyle (2^m)^{\sum\limits_{i=3}^{s+1}\lambda_i\Lambda_{i-2}+\Lambda_{s}(n-\Lambda_{s+1})- \Lambda_{s-1}-\Lambda_{s-\kappa_1}-\frac{\Lambda_{s}(\Lambda_{s}-1)}{2}+\omega_2}{\lambda_{s+1}+n-\Lambda_{s+1}-\Lambda_s \brack \lambda_{s+1}}_{2^m}
\end{equation*} distinct   self-orthogonal codes $\mathcal{C}_{2}$ 
of  type $\{\Lambda_{s},\lambda_{s+1}\}$ and length $n$ over $\mathcal{R}_{2,m}$ satisfying  property $(\mathfrak{P})$  and $Tor_1(\mathcal{C}_2)=\mathcal{D}^{(s)},$    where $\omega_2 =1$ if   $\textbf{1} \in \mathcal{D}^{(s-\kappa_1)},$     while $\omega_2 =0$ otherwise.
 \end{enumerate} \end{proposition}
\begin{proof} Working as in Proposition 4.1 of Yadav and Sharma \cite{Galois} and by applying Lemma \ref{R=2}, the desired result follows. \vspace{-2mm}\end{proof}

Now, in the following proposition,  we assume that $e=3$  and $\kappa=2,$  and  we consider two nested self-orthogonal codes $\mathcal{D}^{(2)}\supseteq \mathcal{D}^{(1)}$   of length $n$ over $\mathcal{T}_m,$ where $\dim \mathcal{D}^{(2)}=\Lambda_{2}$ and $\dim \mathcal{D}^{(1)}=\Lambda_1.$    
Here, we establish the existence of a self-orthogonal code  $\mathcal{C}_3$ of  type $\{\lambda_{1},\lambda_{2},\lambda_{3}\}$ and length $n$ over $\mathcal{R}_{3,m},$ such that   $Tor_1(\mathcal{C}_{3}) =\mathcal{D}^{(1)}$ and $Tor_2(\mathcal{C}_{3}) =\mathcal{D}^{(2)}.$   We also count the choices for such a code $\mathcal{C}_3$ obtained from a given  pair $(\mathcal{D}^{(2)},\mathcal{D}^{(1)})$ of nested self-orthogonal codes over $\mathcal{T}_m.$  
   \begin{proposition}\label{p3.3be=3} Let $e=3$  and $\kappa=  2.$  Let $\mathcal{D}^{(2)}\supseteq \mathcal{D}^{(1)}$ be two nested  self-orthogonal codes of length $n$   over $\mathcal{T}_{m},$  where $\dim \mathcal{D}^{(2)}=\Lambda_{2}$ and  $ \mathcal{D}^{(1)}=\Lambda_1.$   The following hold.
 \begin{enumerate}\vspace{-2mm}\item[(a)]  There exists a  self-orthogonal code $\mathcal{C}_{3}$ of  type $\{\lambda_{1},\lambda_{2},\lambda_{3}\}$ and length $n$ over $\mathcal{R}_{3,m},$ such that  $Tor_1(\mathcal{C}_3)=\mathcal{D}^{(1)}$   and  $Tor_2(\mathcal{C}_3)=\mathcal{D}^{(2)}.$ 
\item[(b)] Furthermore,  each such pair  $(\mathcal{D}^{(2)},\mathcal{D}^{(1)})$ of codes over $\mathcal{T}_m$ gives rise to precisely 	\vspace{-1.5mm}\begin{equation*}\vspace{-1.5mm}
\displaystyle (2^m)^{\lambda_3\lambda_1+(\Lambda_1+\Lambda_{2})(n-\Lambda_{3}-\Lambda_1)+\Lambda_1^2}{\lambda_{3}+n-\Lambda_{3}-\Lambda_1 \brack \lambda_{3}}_{2^m}\vspace{-2mm}
\end{equation*} distinct   self-orthogonal codes $\mathcal{C}_{3}$ of  type $\{\lambda_{1},\lambda_{2},\lambda_{3}\}$ and length $n$ over $\mathcal{R}_{3,m}$  satisfying $Tor_1(\mathcal{C}_3)=\mathcal{D}^{(1)}$   and $Tor_2(\mathcal{C}_3)=\mathcal{D}^{(2)}.$   
 \end{enumerate} \end{proposition}
 \begin{proof} Working as in Proposition 3.2 of Yadav and Sharma \cite{YSub2} and by applying Lemma 4.1 of Yadav and Sharma \cite{quasi},  the desired result follows immediately. 
     \end{proof}

Now, in the following proposition,  we assume that $e\geq 5$ is an odd integer and $\kappa=2,$  and  we consider a chain $\mathcal{D}^{(s+1)}\supseteq \mathcal{D}^{(s)}\supseteq  \mathcal{D}^{(s-1)}$ of  self-orthogonal codes of length $n$ over $\mathcal{T}_m,$ where $\dim \mathcal{D}^{(s+1)}=\Lambda_{s+1},$ $\dim \mathcal{D}^{(s)}=\Lambda_s$ and $\dim \mathcal{D}^{(s-1)}=\Lambda_{s-1}.$  
Here, we establish the existence of a self-orthogonal code  $\mathcal{C}_3$ of  type $\{\Lambda_{s},\lambda_{s+1},\lambda_{s+2}\}$ and length $n$ over $\mathcal{R}_{3,m}$ satisfying  property $(\mathfrak{P})$ and   $Tor_1(\mathcal{C}_{3}) =\mathcal{D}^{(s)}$ and $Tor_2(\mathcal{C}_{3}) =\mathcal{D}^{(s+1)}.$   We also count the choices for such a code $\mathcal{C}_3$ obtained from a given  triplet $(\mathcal{D}^{(s+1)},\mathcal{D}^{(s)},\mathcal{D}^{(s-1)})$ of codes over $\mathcal{T}_m.$  
   \begin{proposition}\label{p3.3b} Let $e\geq 5$ be an odd integer, and let $\kappa=  2.$  Let $\mathcal{D}^{(s+1)}\supseteq \mathcal{D}^{(s)}\supseteq \mathcal{D}^{(s-1)}$ be a chain of  self-orthogonal codes of length $n$   over $\mathcal{T}_{m},$  where $\dim \mathcal{D}^{(s+1)}=\Lambda_{s+1},$ $ \mathcal{D}^{(s)}=\Lambda_s$ and $\mathcal{D}^{(s-1)}=\Lambda_{s-1}.$  
   The following hold.
 \begin{enumerate}\vspace{-2mm}\item[(a)]  There exists a  self-orthogonal code $\mathcal{C}_{3}$ of type $\{\Lambda_{s},\lambda_{s+1},\lambda_{s+2}\}$ and length $n$ over $\mathcal{R}_{3,m}$  satisfying  property $(\mathfrak{P})$ and  $Tor_1(\mathcal{C}_3)=\mathcal{D}^{(s)}$   and  $Tor_2(\mathcal{C}_3)=\mathcal{D}^{(s+1)},$ where such a code  exists: \vspace{-1mm}\begin{itemize}
     \item unconditionally, if $\textbf{1}\notin \mathcal{D}^{(s-1)};$
 \item if and only if either   $n \equiv 0  \pmod{8}$, or $n \equiv 4 \pmod{8}$ with $m$ even, 
 when $\textbf{1}\in\mathcal{D}^{(s-1)}.$ \end{itemize}
\item[(b)] Furthermore,  each such triplet $(\mathcal{D}^{(s+1)},\mathcal{D}^{(s)}, \mathcal{D}^{(s-1)} )$ of codes over $\mathcal{T}_m$ gives rise to precisely 	\vspace{-1mm}\begin{equation*}
\displaystyle 2^{\epsilon}(2^m)^{\sum\limits_{i=3}^{s+2}\lambda_i\Lambda_{i-2}+\sum\limits_{j=4}^{s+2}\lambda_j\Lambda_{j-3}+(\Lambda_s+\Lambda_{s+1})(n-\Lambda_{s+2}-\Lambda_s)+\Lambda_s^2-2\Lambda_{s-1}+\epsilon}{\lambda_{s+2}+n-\Lambda_{s+2}-\Lambda_s \brack \lambda_{s+2}}_{2^m}\vspace{-2mm}
\end{equation*} distinct   self-orthogonal codes $\mathcal{C}_{3}$ of type $\{\Lambda_{s},\lambda_{s+1},\lambda_{s+2}\}$ and length $n$ over $\mathcal{R}_{3,m}$ satisfying  property $(\mathfrak{P})$ and $Tor_1(\mathcal{C}_3)=\mathcal{D}^{(s)}$   and $Tor_2(\mathcal{C}_3)=\mathcal{D}^{(s+1)},$    where $\epsilon=1$ if $\textbf{1}\in \mathcal{D}^{(s-1)},$ while $\epsilon=0$ otherwise.
 \end{enumerate} \end{proposition}
 \begin{proof} Working as in Proposition 4.2 of Yadav and Sharma  \cite{Galois},  the desired result follows immediately. 
     \end{proof}  

 In the following proposition,  we  assume that  $e \geq 5$ is an odd integer and $\kappa$ is an even integer satisfying $4\leq \kappa \leq e-1,$  and  we consider a chain $\mathcal{D}^{(s+1)}\supseteq \mathcal{D}^{(s)}\supseteq  \mathcal{D}^{(s-\kappa_1)}$  of  self-orthogonal codes of length $n$ over $\mathcal{T}_m,$ where $\dim \mathcal{D}^{(s+1)}=\Lambda_{s+1},$ $\dim \mathcal{D}^{(s)}=\Lambda_s$ and $\dim \mathcal{D}^{(s-\kappa_1)}=\Lambda_{s-\kappa_1}.$  
Here, we establish the existence of a self-orthogonal code  $\mathcal{C}_3$ of type $\{\Lambda_{s},\lambda_{s+1},\lambda_{s+2}\}$ and length $n$ over $\mathcal{R}_{3,m}$ satisfying  property $(\mathfrak{P})$ and   $Tor_1(\mathcal{C}_{3}) =\mathcal{D}^{(s)}$ and $Tor_2(\mathcal{C}_{3}) =\mathcal{D}^{(s+1)}.$   We also count the choices for such a code $\mathcal{C}_3$ obtained from a given  triplet $(\mathcal{D}^{(s+1)},\mathcal{D}^{(s)},\mathcal{D}^{(s-\kappa_1)})$ of codes over $\mathcal{T}_m.$ 
\begin{proposition}\label{p3.3c} Let $e \geq 5$ be an odd integer, and let $\kappa$ be an even integer satisfying $4 \leq \kappa\leq  e-1 .$ Let $\mathcal{D}^{(s+1)}\supseteq \mathcal{D}^{(s)}\supseteq \mathcal{D}^{(s-\kappa_1)}$ be a chain of  self-orthogonal codes of length $n$   over $\mathcal{T}_{m},$  where $\dim \mathcal{D}^{(s+1)}=\Lambda_{s+1},$ $ \mathcal{D}^{(s)}=\Lambda_s$ and $\mathcal{D}^{(s-\kappa_1)}=\Lambda_{s-\kappa_1}.$  
  Then the following hold.
 \begin{enumerate}\vspace{-2mm}\item[(a)]  There exists a  self-orthogonal code $\mathcal{C}_{3}$ of type $\{\Lambda_{s},\lambda_{s+1},\lambda_{s+2}\}$ and length $n$ over $\mathcal{R}_{3,m}$  satisfying  property $(\mathfrak{P})$ and  $Tor_1(\mathcal{C}_3)=\mathcal{D}^{(s)}$   and  $Tor_2(\mathcal{C}_3)=\mathcal{D}^{(s+1)},$ where such a code  exists: \vspace{-1mm}\begin{itemize}
     \item unconditionally, if $\textbf{1}\notin \mathcal{D}^{(s-\kappa_1)};$
 \item if and only if either  $n\equiv 0,4 \pmod 8$  or $n\equiv 2,6\pmod 8$  and $\eta_1=0,$ 
 when $\textbf{1}\in\mathcal{D}^{(s-\kappa_1)}.$ \end{itemize} 
\item[(b)] Furthermore,   each such triplet $(\mathcal{D}^{(s+1)},\mathcal{D}^{(s)}, \mathcal{D}^{(s-\kappa_1)} )$ of  codes over $\mathcal{T}_m$ gives rise to precisely 	\vspace{-1mm}\begin{equation*}\vspace{-1mm}
\displaystyle (2^m)^{\sum\limits_{i=3}^{s+2}\lambda_i\Lambda_{i-2}+\sum\limits_{j=4}^{s+2}\lambda_j\Lambda_{j-3}+(\Lambda_s+\Lambda_{s+1})(n-\Lambda_{s+2}-\Lambda_s)+\Lambda_s^2-\Lambda_{s-1}-\Lambda_{s-\kappa_1}+\omega_3}{\lambda_{s+2}+n-\Lambda_{s+2}-\Lambda_s \brack \lambda_{s+2}}_{2^m}
\end{equation*} distinct   self-orthogonal codes $\mathcal{C}_{3}$ of type $\{\Lambda_{s},\lambda_{s+1},\lambda_{s+2}\}$ and length $n$ over $\mathcal{R}_{3,m}$ satisfying  property $(\mathfrak{P})$ and $Tor_1(\mathcal{C}_3)=\mathcal{D}^{(s)}$   and $Tor_2(\mathcal{C}_3)=\mathcal{D}^{(s+1)},$    where $\omega_3=1$ if $\textbf{1}\in \mathcal{D}^{(s-\kappa_1)}$,  while $\omega_3=0$ otherwise.
 \end{enumerate} \end{proposition}
 \begin{proof} Working as in Proposition 3.2 of Yadav and Sharma \cite{YSub2} and Lemma  \ref{R=2},   the desired result follows immediately. 
    \vspace{-1mm} \end{proof}  
From this point on, we will  distinguish the following two cases: (i) $2 \kappa \leq e,$ and (ii) $2 \kappa > e.$  In the following proposition, we  assume that $\ell$ is a fixed integer satisfying $\ell \equiv e \pmod 2$ and $4\leq \ell \leq \kappa$ if $2\kappa \leq e,$ while  the integer $\ell$ satisfies $\ell \equiv e \pmod 2$ and $4 \leq \ell \leq  e-\kappa$ if $2\kappa >e.$ Here, we provide a method to  lift a self-orthogonal code  of type $\{\Lambda_{\gamma_{\ell}+2},\lambda_{\gamma_{\ell}+3}, \ldots, \lambda_{\gamma_{\ell}+\ell-1}\}$   and  length $n$ over $\mathcal{R}_{\ell-2,m}$ satisfying    property $(\mathfrak{P})$  to a self-orthogonal code  of type $\{\Lambda_{\gamma_{\ell}+1},\lambda_{\gamma_{\ell}+2}, \ldots, \lambda_{\gamma_{\ell}+\ell}\}$ and length $n$ over $\mathcal{R}_{\ell,m}$ satisfying   property $(\mathfrak{P}),$ where $\gamma_{\ell}=s-\mathrm{f}_{\ell}.$   We also compute the total number of distinct ways to carry out this lifting.
\begin{proposition}\label{p3.4Keven}   Let  $\ell$ be a fixed integer  satisfying  $\ell \equiv e \pmod 2$ and $4 \leq \ell \leq \kappa$  if $2\kappa \leq e,$ while  $\ell \equiv e \pmod 2$ and $4 \leq \ell \leq  e-\kappa$ if $2\kappa >e.$   Let us define $\gamma_{\ell}=s-\mathrm{f}_{\ell}.$  Let $\mathcal{C}_{\ell-2}$ be a self-orthogonal code of type $\{\Lambda_{\gamma_{\ell}+2},\lambda_{\gamma_{\ell}+3}, \ldots, \lambda_{\gamma_{\ell}+\ell-1}\}$   and  length $n$ over $\mathcal{R}_{\ell-2,m}$ satisfying   property $(\mathfrak{P}).$  
Let $\mathcal{D}^{(\gamma_{\ell}+1)}$ and $\mathcal{D}^{(\gamma_{\ell}+1-\kappa_1)}$ be   linear subcodes of  $Tor_{1}(\mathcal{C}_{\ell-2}),$ such that   \vspace{-1mm}\begin{equation*}\vspace{-1mm}
\mathcal{D}^{(\gamma_{\ell}+1-\kappa_1)}\subseteq \mathcal{D}^{(\gamma_{\ell}+1)},~~\dim \mathcal{D}^{(\gamma_{\ell}+1)}=\Lambda_{\gamma_{\ell}+1}~ \text{ and }~\dim \mathcal{D}^{(\gamma_{\ell}+1-\kappa_1)}=\Lambda_{\gamma_{\ell}+1-\kappa_1}.  
\end{equation*}
\begin{enumerate}\item[(a)] There exists a  self-orthogonal code $\mathcal{C}_{\ell}$ of type $\{\Lambda_{\gamma_{\ell}+1},\lambda_{\gamma_{\ell}+2}, \ldots, \lambda_{\gamma_{\ell}+\ell}\}$ and length $n$ over $\mathcal{R}_{\ell,m}$ satisfying   property $(\mathfrak{P})$ and     $Tor_{1}(\mathcal{C}_{\ell})=\mathcal{D}^{(\gamma_{\ell}+1)}$ and $Tor_{j+1}(\mathcal{C}_{\ell})=Tor_{j}(\mathcal{C}_{\ell-2})$ for $1 \leq j \leq  \ell-2,$ where such a code exists:
 \vspace{-1mm}\begin{itemize}
     \item unconditionally, if $\textbf{1}\notin \mathcal{D}^{(\gamma_{\ell}+1-\kappa_1)};$ \item  if and only if one of the following holds  when $\textbf{1}\in \mathcal{D}^{(\gamma_{\ell}+1-\kappa_1)}:$ \begin{itemize}\item[(i)]  $n\equiv 0,4 \pmod 8;$  \item[(ii)]   $n\equiv 2,6\pmod 8$, $2\kappa \leq e,$ $ 4 \leq \ell \leq \kappa_1+\theta_e$   and $\eta_{\ell-1-\theta_e}=0$; \item[(iii)]  $n\equiv 2,6 \pmod 8,$ $2\kappa \leq e,$  $\ell= \kappa_1+1+\theta_e$    and  $(\eta_0)^{\frac{3}{2}}=\eta_{\kappa_1};$  
     \item[(iv)]   $n\equiv 2,6\pmod 8$, $2\kappa > e,$ $ 4 \leq \ell \leq \min\{e-\kappa,\kappa_1+\theta_e\}$   and $\eta_{\ell-1-\theta_e}=0$; 
     
     \item[(v)]  $n\equiv 2,6 \pmod 8,$ $2\kappa > e \geq \frac{3}{2}\kappa+1+\theta_e,$  $\ell= \kappa_1+1+\theta_e$     and  $(\eta_0)^{\frac{3}{2}}=\eta_{\kappa_1}.$ \end{itemize}
 \end{itemize}
\item[(b)] Moreover,  each such triplet $(\mathcal{C}_{\ell-2},\mathcal{D}^{(\gamma_{\ell}+1)},\mathcal{D}^{(\gamma_{\ell}+1-\kappa_1)})$ of codes
gives rise to  precisely \vspace{-1.5mm}\begin{equation*}\vspace{-1mm}
\displaystyle (2^m)^{\sum\limits_{i=\ell}^{\gamma_{\ell}+\ell}\lambda_i\Lambda_{i-\ell+1}+\sum\limits_{j=\ell+1}^{\gamma_{\ell}+\ell}\lambda_j\Lambda_{j-\ell}+Y_{\ell}}  {\lambda_{\gamma_{\ell}+\ell}+n-\Lambda_{\gamma_{\ell}+\ell}-\Lambda_{\gamma_{\ell}+1} \brack \lambda_{\gamma_{\ell}+\ell}}_{2^m}
\end{equation*}  distinct  self-orthogonal codes $\mathcal{C}_{\ell}$ of type $\{\Lambda_{\gamma_{\ell}+1},\lambda_{\gamma_{\ell}+2}, \ldots, \lambda_{\gamma_{\ell}+\ell}\}$ and length $n$ over $\mathcal{R}_{\ell,m}$ satisfying  property $(\mathfrak{P})$ and $Tor_{1}(\mathcal{C}_{\ell})=\mathcal{D}^{(\gamma_{\ell}+1)}$ and $Tor_{j+1}(\mathcal{C}_{\ell})=Tor_{j}(\mathcal{C}_{\ell-2})$ for $1 \leq j \leq  \ell-2,$ where  \begin{equation*} Y_{\ell}=(\Lambda_{\gamma_{\ell}+\ell-1}+\Lambda_{\gamma_{\ell}+1})(n-\Lambda_{\gamma_{\ell}+\ell}-\Lambda_{\gamma_{\ell}+1}) +\Lambda_{\gamma_{\ell}+1}^2 +\Lambda_{\gamma_{\ell}+1}-\Lambda_{s-2\mathrm{f}_\ell+2}-\Lambda_{s-2\mathrm{f}_\ell+1}-\Lambda_{\gamma_{\ell}+1-\kappa_1 }+\omega_{\ell}\end{equation*} with 
$\omega_{\ell}=1$ if $\textbf{1}\in \mathcal{D}^{(\gamma_{\ell}+1-\kappa_1)}$,     while $\omega_{\ell}=0$ otherwise.
 \end{enumerate} \end{proposition}
 \begin{proof} To prove the result, 
 we assume, without any loss of generality, that the code  
$\mathcal{C}_{\ell-2}$ has a generator matrix 
  \begin{equation*}\label{e00.1Kodd}
\mathtt{G}_{\ell-2}=\begin{bmatrix}
T_1^{(\ell-2)}\\T_2^{(\ell-2)} \\ \vdots\\ T_{\gamma_{\ell}+2}^{(\ell-2)}\vspace{0.5mm}\\ uT_{\gamma_{\ell}+3}^{(\ell-2)}\vspace{0.5mm}\\  \vdots\\u^{\ell-3}T_{\gamma_{\ell}+\ell-1}^{(\ell-2)}
\end{bmatrix}=\begin{bmatrix}
T_1^{(0)}+u\mathtt{U}_1^{(1)}+u^2\mathtt{U}_1^{(2)}+\cdots +u^{\ell-3}\mathtt{U}_1^{(\ell-3)}\\T_2^{(0)}+u\mathtt{U}_2^{(1)}+u^2\mathtt{U}_2^{(2)}+\cdots +u^{\ell-3}\mathtt{U}_2^{(\ell-3)} \\ \vdots\\ T_{\gamma_{\ell}+2}^{(0)}+u\mathtt{U}_{\gamma_{\ell}+2}^{(1)}+u^2\mathtt{U}_{\gamma_{\ell}+2}^{(2)}+\cdots +u^{\ell-3}\mathtt{U}_{\gamma_{\ell}+2}^{(\ell-3)}\vspace{0.5mm}\\ uT_{\gamma_{\ell}+3}^{(\ell-2)}\vspace{0.5mm}\\ \vdots\\u^{\ell-3}T_{\gamma_{\ell}+\ell -1}^{(\ell-2)}
\end{bmatrix},\end{equation*} 
 where  
\begin{equation*}\label{e00.2Kodd}
[T^{(0)}]_{\gamma_{\ell}+2}=\begin{bmatrix}
T_1^{(0)}\\T_2^{(0)}\\ \vdots\\ T_{\gamma_{\ell}+2}^{(0)}
\end{bmatrix}  =\begin{bmatrix}
{\mathtt{I}}_{\lambda_1}&{\mathtt{A}}_{1,1}^{(0)}&\cdots&{\mathtt{A}}_{1,\gamma_{\ell}+1}^{(0)}&\cdots & {\mathtt{A}}_{1,e-1}^{(0)}& {\mathtt{A}}_{1,e}^{(0)}\\
\mathbf{0} &{\mathtt{I}}_{\lambda_2}  &\cdots &{\mathtt{A}}_{2,\gamma_{\ell}+1}^{(0)}&\cdots & {\mathtt{A}}_{2,e-1}^{(0)}&{\mathtt{A}}_{2,e}^{(0)}\\
	\vdots& \vdots   &\vdots&\vdots &\vdots& \vdots&\vdots\\
\mathbf{0}&\mathbf{0}&\cdots& {\mathtt{I}}_{\lambda_{\gamma_{\ell}+2}}& \cdots & {\mathtt{A}}_{\gamma_{\ell}+2,e-1}^{(0)}&{\mathtt{A}}_{\gamma_{\ell}+2,e}^{(0)}
\end{bmatrix}\end{equation*} 
with  ${\mathtt{I}}_{\lambda_i}$ as the $\lambda_i\times \lambda_i$ identity matrix over ${\mathcal{T}}_{m},$    $ {\mathtt{A}}_{i,j}^{(0)} \in \mathcal{M}_{\lambda_i\times \lambda_{j+1}}({\mathcal{T}}_{m})$  for $1 \leq i \leq \gamma_{\ell}+2$ and $i \leq j \leq e,$
 $[\mathtt{U}^{(b)}]_{\gamma_{\ell}+2}\in \mathcal{M}_{\Lambda_{\gamma_{\ell}+2}\times n}(\mathcal{T}_{m})$ for $1 \leq b \leq \ell-3,$ and  the matrix $T_{\gamma_{\ell}+g}^{(\ell-2)}\in \mathcal{M}_{\lambda_{\gamma_{\ell}+g}\times n}(\mathcal{R}_{\ell-2,m})$  to be considered modulo $u^{\ell-g}$   for  $3 \leq g \leq \ell-1.$ Here,  we note that the torsion code $Tor_1(\mathcal{C}_{\ell-2})$ is a $\Lambda_{\gamma_{\ell}+2}$-dimensional code over $\mathcal{T}_{m}$ with a generator matrix $[T^{(0)}]_{\gamma_{\ell}+2}.$   
  We also assume, without any loss of generality, that the  code  $\mathcal{D}^{(\gamma_{\ell}+1)}$ has a generator matrix $[T^{(0)}]_{\gamma_{\ell}+1},$
and  that the code $\mathcal{D}^{(\gamma_{\ell}+1-\kappa_1)}$  has a generator matrix $[T^{(0)}]_{\gamma_{\ell}+1-\kappa_1} .$ 
Further, by Remark \ref{r3.1}, we  assume that the matrix $(\mathtt{A}^{(0)})_ { \gamma_{\ell}+1,\gamma_{\ell}+\ell} $ has full row-rank. 
 
Now, since $\mathcal{C}_{\ell-2}$ is a self-orthogonal code over $\mathcal{R}_{\ell-2,m}$ satisfying   property $(\mathfrak{P}),$  we have
\begin{eqnarray*}
[T^{(\ell-2)}]_{\gamma_{\ell}+1}[T^{(\ell-2)}]_{\gamma_{\ell}+1}^t &\equiv &\textbf{0}\pmod{u^{\ell-2}},\label{}\\
  \displaystyle \mathcal{D}iag\left([T^{(\ell-2)}]_{\gamma_{\ell}+2-\mathrm{f}_i}[T^{(\ell-2})]_{\gamma_{\ell}+2-\mathrm{f}_i}^t\right) &\equiv & \textbf{0} \pmod{u^{\ell-2+i}} \text{ ~~for ~} i\in \{2,4,6,\ldots,\ell-2-\theta_e\},~\\ 
  \displaystyle  \pi_{\ell-2+j}\left( \mathcal{D}iag\big([T^{(\ell-2)}]_{\gamma_{\ell}+2-\mathrm{f}_{j+2}}[T^{(\ell-2)}]_{\gamma_{\ell}+2-\mathrm{f}_{j+2}}^t\big) \right)&=&\textbf{0} \text{ ~~for~~ } j \in \{ \ell-3,\ell-1,\ldots, \kappa -1-\theta_e  \},\\
  ~[T^{(\ell-2)}]_{ \gamma_{\ell}+1} T_{\gamma_{\ell}+g}^{(\ell-2)t}&\equiv &\textbf{0}  \pmod{u^{\ell-g}} \text{ ~for } 2 \leq g \leq \ell-1 , \text{ ~and}\\
T_{\gamma_{\ell}+i}^{(\ell-2)}
 T_{\gamma_{\ell}+j}^{(\ell-2)t}&\equiv &\textbf{0}  \pmod{u^{\ell+2-i-j}}~ \\ &&~~~~\text{ for } 2 \leq i,j\leq \ell-1 \text{ and } i+j\leq \ell+1.\end{eqnarray*} 
   This implies that \begin{eqnarray*}
    [T^{(\ell-2)}]_{\gamma_{\ell}+1} [T^{(\ell-2)}]_{\gamma_{\ell}+1}^t&\equiv&  u^{\ell-2}\mathtt{H}^{(\ell-2)}+u^{\ell-1}\mathtt{H}^{(\ell-1)}+\cdots+u^{\ell+\kappa-1}\mathtt{H}^{(\ell+\kappa-1)} \pmod{u^{\ell+\kappa}},\\ ~[T^{(\ell-2)}]_{\gamma_{\ell}+1} T_{\gamma_{\ell}+g}^{(\ell-2)t}&\equiv &u^{\ell-g}S_{g}  \pmod{u^{\ell+1-g}} \text{ ~for } 2 \leq g \leq \ell-1 ,\text{ ~and}\\
T_{\gamma_{\ell}+i}^{(\ell-2)} T_{\gamma_{\ell}+j}^{(\ell-2)t}&\equiv &\textbf{0} \pmod{u^{\ell+2-i-j}} ~\text{ for } 2 \leq i,j\leq \ell-1 \text{ and } i+j\leq \ell+1,
\end{eqnarray*} where $\mathtt{H}^{(\ell-2)},\mathtt{H}^{(\ell-1)}, \ldots, \mathtt{H}^{(\ell+\kappa-1)} \in Sym_{\Lambda_{\gamma_{\ell}+1}}(\mathcal{T}_{m})$ with $\mathtt{H}_{z,z}^{(\ell-4+i)}=\mathtt{H}^{(\ell-3+i)}_{z,z}=0 $ for $1 \leq z \leq \Lambda_{\gamma_{\ell}+2-\mathrm{f}_i }$ and $ i \in \{ 2,4,\ldots,\ell-2-\theta_e\}, $  $\mathtt{H}_{b,b}^{(\ell-2+j)}=0 $ for $ 1 \leq b \leq \Lambda_{\gamma_{\ell}+2-\mathrm{f}_{j+2}} $ and $ j\in \{ \ell-3,\ell-1, \ell+1,\ldots, \kappa-1-\theta_e\},$ and $S_{g} \in \mathcal{M}_{\Lambda_{\gamma_{\ell}+1} \times \lambda_{\gamma_{\ell}+g}}(\mathcal{T}_{m})  $ for $ 2 \leq g \leq \ell-1.$   
   
To establish the result,  we define a matrix  $\mathtt{G}_{\ell}$ over $\mathcal{R}_{\ell,m}$ as follows:
\vspace{-2mm}\begin{equation*} \vspace{-1mm}\mathtt{G}_{\ell}=\begin{bmatrix}
T_1^{(\ell)}\\T_2^{(\ell)}\\ \vdots\\ T_{\gamma_{\ell}+1}^{(\ell)}\\ \vspace{0.5mm}uT_{\gamma_{\ell}+2}^{(\ell)}\\ \vdots\\u^{\ell-1}T_{\gamma_{\ell}+\ell}^{(\ell)}
\end{bmatrix}=\begin{bmatrix} T_1^{(\ell-2)}+u^{\ell-2}\mathtt{U}_1^{(\ell-2)}+u^{\ell-1}\mathtt{U}_1^{(\ell-1)}\\ T_2^{(\ell-2)}+u^{\ell-2}\mathtt{U}_2^{(\ell-2)}+u^{\ell-1}\mathtt{U}_2^{(\ell-1)}\\ \vdots\\ T_{\gamma_{\ell}+1}^{(\ell-2)}+u^{\ell-2}\mathtt{U}_{\gamma_{\ell}+1}^{(\ell-2)}+u^{\ell-1}\mathtt{U}_{\gamma_{\ell}+1}^{(\ell-1)}\\ \vspace{0.5mm}uT_{\gamma_{\ell}+2}^{(\ell)}\\ \vdots\\u^{\ell-1}T_{\gamma_{\ell}+\ell}^{(\ell)}   \end{bmatrix},\end{equation*} 
where the matrices  $[\mathtt{U}^{(\mu)}]_{\gamma_{\ell}+1}\in \mathcal{M}_{\Lambda_{\gamma_{\ell}+1}\times n}(\mathcal{T}_{m})$ for $ \mu \in \{\ell-2,\ell-1\},$   $T_{\gamma_{\ell}+g}^{(\ell)}\in \mathcal{M}_{\lambda_{\gamma_{\ell}+g}\times n}(\mathcal{R}_{\ell,m})$   for  $2 \leq g \leq \ell-1$ and  $T_{\gamma_{\ell}+\ell}^{(\ell)}\in \mathcal{M}_{\lambda_{\gamma_{\ell}+\ell}\times n}(\mathcal{T}_{m})$ are of the forms 
\vspace{-1.5mm}\begin{equation*}\vspace{-2mm}\begin{bmatrix}
\mathtt{U}_1^{(\mu)}\\\mathtt{U}_2^{(\mu)}\\ \vdots\\ \mathtt{U}_{\gamma_{\ell}+1}^{(\mu)}
\end{bmatrix}= \begin{bmatrix}
\textbf{0}&\cdots&\textbf{0}&\mathtt{A}_{1,\mu+1}^{(\mu)}& \mathtt{A}_{1,\mu+2}^{(\mu)}&\cdots&\mathtt{A}_{1,\gamma_{\ell}+1+\mu}^{(\mu)}&\cdots &  \mathtt{A}_{1,e}^{(\mu)}\\
\textbf{0} & \cdots&\textbf{0}&\textbf{0}& \mathtt{A}_{2,\mu+2}^{(\mu)}&  \cdots &\mathtt{A}_{2,\gamma_{\ell}+1+\mu}^{(\mu)}&\cdots &\mathtt{A}_{2,e}^{(\mu)}\\
	\vdots& \cdots &\vdots &\vdots &\vdots&\vdots & \vdots&\vdots&\vdots\\
\textbf{0}&\cdots&\textbf{0}&\textbf{0}&\textbf{0}&\cdots&  \mathtt{A}_{\gamma_{\ell}+1,\gamma_{\ell}+1+\mu}^{(\mu)}&\cdots &  \mathtt{A}_{\gamma_{\ell}+1,e}^{(\mu)}\\
\end{bmatrix},\end{equation*} 
 \begin{equation*} T_{\gamma_{\ell}+g}^{(\ell)}= T_{\gamma_{\ell}+g}^{(\ell-2)}+u^{\ell-g}\begin{bmatrix}
\textbf{0}&\cdots&\textbf{0}& \mathtt{A}_{\gamma_{\ell}+g,\gamma_{\ell}+\ell}^{(\ell-g)}&\cdots & \mathtt{A}_{\gamma_{\ell}+g,e}^{(\ell-g)}
\end{bmatrix} \text{ and } \end{equation*} 
\begin{equation*} T_{\gamma_{\ell}+\ell}^{(\ell)}= \begin{bmatrix}
\textbf{0}&\cdots&\textbf{0}&\mathtt{I}_{\lambda_{\gamma_{\ell}+\ell}}& \mathtt{A}_{\gamma_{\ell}+\ell,\gamma_{\ell}+\ell}^{(0)}& \cdots & \mathtt{A}_{\gamma_{\ell}+\ell,e}^{(0)}
\end{bmatrix} \end{equation*}
with   $\mathtt{A}_{i,j}^{(\mu)} \in \mathcal{M}_{\lambda_i \times \lambda_{j+1}}(\mathcal{T}_{m})$ for $1 \leq i \leq \gamma_{\ell}+1$ and $ i+\mu \leq j \leq e,$   $\mathtt{A}_{\gamma_{\ell}+g,h}^{(\ell-g)} \in \mathcal{M}_{\lambda_{\gamma_{\ell}+g} \times  \lambda_{h+1}}(\mathcal{T}_{m}) \text{ for } \gamma_{\ell}+\ell \leq h \leq e$  and $\mathtt{A}_{\gamma_{\ell}+\ell,a}^{(0)} \in \mathcal{M}_{\lambda_{\gamma_{\ell}+\ell} \times  \lambda_{a+1}}(\mathcal{T}_{m}) \text{ for } \gamma_{\ell}+\ell \leq a \leq e. $
Let  $\mathcal{C}_{\ell}$  be a  linear code of length $n$ over $\mathcal{R}_{\ell,m}$ with a generator matrix $\mathtt{G}_{\ell}.$   Note that the code $\mathcal{C}_{\ell}$ is  of  type $\{\Lambda_{\gamma_{\ell}+1},\lambda_{\gamma_{\ell}+2}, \ldots, \lambda_{\gamma_{\ell}+\ell}\}$ over $\mathcal{R}_{\ell,m}$ and satisfies $Tor_1(\mathcal{C}_{\ell})=\mathcal{D}^{(\gamma_{\ell}+1)}$ and $Tor_{j+1}(\mathcal{C}_{\ell})=Tor_{j}(\mathcal{C}_{\ell-2})$ for $1 \leq j \leq \ell-2.$  Further, by  Lemma \ref{l2.2},  we see that the code  $\mathcal{C}_{\ell}$ is a self-orthogonal code over $\mathcal{R}_{\ell,m}$ satisfying  property $(\mathfrak{P})$   if and only if there exist matrices  $[\mathtt{U}^{(\ell-2)}]_{\gamma_{\ell}+1},$ $[\mathtt{U}^{(\ell-1)}]_{\gamma_{\ell}+1},$ $\left[0\cdots0~\mathtt{A}_{\gamma_{\ell}+g,\gamma_{\ell}+\ell}^{(\ell-g)}~\cdots ~ \mathtt{A}_{\gamma_{\ell}+g,e}^{(\ell-g)}\right]$ for $2 \leq g \leq \ell-1$ and  $T_{\gamma_{\ell}+\ell}^{(\ell)}$
satisfying the following system of matrix equations:
\small{\begin{eqnarray}
 \mathtt{H}^{(\ell-2)}+[T^{(0)}]_{\gamma_{\ell}+1}[\mathtt{U}^{(\ell-2)}]_{\gamma_{\ell}+1}^t+[\mathtt{U}^{(\ell-2)}]_{\gamma_{\ell}+1}[T^{(0)}]_{\gamma_{\ell}1+1}^t+
 u\left([T^{(0)}]_{\gamma_{\ell}+1}[\mathtt{U}^{(\ell-1)}]_{\gamma_{\ell}+1}^t\right.~~~~&&\nonumber\\\left.+[\mathtt{U}^{(\ell-1)}]_{\gamma_{\ell}+1}[T^{(0)}]_{\gamma_{\ell}+1}^t +[\mathtt{U}^{(1)}]_{\gamma_{\ell}+1}[\mathtt{U}^{(\ell-2)}]_{\gamma_{\ell}+1}^t+[\mathtt{U}^{(\ell-2)}]_{\gamma_{\ell}+1}[\mathtt{U}^{(1)}]_{\gamma_{\ell}+1}^t+\mathtt{H}^{(\ell-1)}\right)&\equiv& \textbf{0}\pmod{u^2},\label{e3.30Keven}\\
\mathcal{D}iag\left(\mathtt{E}^{(2\ell-4)}+[\mathtt{U}^{(\ell-2)}]_{s-2\mathrm{f}_\ell+2}[\mathtt{U}^{(\ell-2)}]_{s-2\mathrm{f}_\ell+2}^t\right) &\equiv & \textbf{0} \pmod{u},\label{e3.31Keven}~~~~~~~\\ 
\mathcal{D}iag\left(\mathtt{E}^{(2\ell-2)}+[\mathtt{U}^{(\ell-1)}]_{s-2\mathrm{f}_\ell+1}[\mathtt{U}^{(\ell-1)}]_{s-2\mathrm{f}_\ell+1}^t\right) &\equiv &\textbf{0} \pmod{u},\label{e3.32Keven}\\
\mathcal{D}iag\left(\mathtt{E}^{(\ell+\kappa-1)}+\eta_1[T^{(0)}]_{\gamma_{\ell}-\kappa_1+1}[\mathtt{U}^{(\ell-2)}]_{\gamma_{\ell}-\kappa_1+1}^t\right.~~~~~~~&&\nonumber\\\left.+\eta_0\left([T^{(0)}]_{\gamma_{\ell}-\kappa_1+1}[\mathtt{U}^{(\ell-1)}]_{\gamma_{\ell}-\kappa_1+1}^t+[\mathtt{U}^{(1)}]_{\gamma_{\ell}-\kappa_1+1}[\mathtt{U}^{(\ell-2)}]_{\gamma_{\ell}-\kappa_1+1}^t\right)\right) &\equiv & \textbf{0} \pmod{u} \text{ if } \theta_e=0,\label{e3.33bKeven}~~~~~~~\\
\mathcal{D}iag\left(\mathtt{E}^{(\ell+\kappa-2)}+\eta_0[T^{(0)}]_{\gamma_{\ell}-\kappa_1+1}[\mathtt{U}^{(\ell-2)}]_{\gamma_{\ell}-\kappa_1+1}^t\right) &\equiv &\textbf{0} \pmod{u}\text{ if } \theta_e=1,\label{e3.33Keven}~~~~~~~\\ 
S_{g}+[T^{(0)}]_{ \gamma_{\ell}+1} \begin{bmatrix} 0~ \cdots ~0 ~\mathtt{A}_{ \gamma_{\ell}+g, \gamma_{\ell}+\ell}^{(\ell-g)} ~\cdots ~ \mathtt{A}_{ \gamma_{\ell}+g,e}^{(\ell-g)}\end{bmatrix} ^t&\equiv &\textbf{0}  \pmod{u}, \text{ ~and}~~\label{e3.34Keven}\\
~[T^{(0)}]_{ \gamma_{\ell}+1} T_{ \gamma_{\ell}+\ell}^{(\ell)t} &\equiv& \textbf{0} \pmod{u},\label{e3.35Keven}\end{eqnarray}}\normalsize
 where   $\mathtt{E}^{(2\ell-4)},$ $\mathtt{E}^{(2\ell-2)}$ and $\mathtt{E}^{(\ell+\kappa-1-\theta_e)}$ are square matrices  of orders $\Lambda_{s-2\mathrm{f}_\ell+2} ,$  $ \Lambda_{s-2\mathrm{f}_\ell+1}$ and $\Lambda_{\gamma_{\ell}+1-\kappa_1} $  over $\mathcal{T}_{m}$ whose rows are the  first $\Lambda_{s-2\mathrm{f}_\ell+2},$ $\Lambda_{s-2\mathrm{f}_\ell+1}$ and $\Lambda_{\gamma_{\ell}+1-\kappa_1}$  rows of the matrices $\mathtt{H}^{(2\ell-4)}, \mathtt{H}^{(2\ell-2)}$ and $ \mathtt{H}^{(\ell+\kappa-1-\theta_e)}, $   respectively.  We will next establish the existence of the   matrices $[\mathtt{U}^{(\ell-2)}]_{\gamma_{\ell}+1},$ $[\mathtt{U}^{(\ell-1)}]_{\gamma_{\ell}+1},$ $\big[\mathbf{0}~\cdots ~\mathbf{0}~ \mathtt{A}_{\gamma_{\ell}+g,\gamma_{\ell}+\ell}^{(\ell-g)}~\cdots ~ \\ \mathtt{A}_{\gamma_{\ell}+g,e}^{(\ell-g)}\Big]$ for $2 \leq g \leq \ell-1$ and  $T_{\gamma_{\ell}+\ell}^{(\ell)}$  satisfying the system of matrix equations \eqref{e3.30Keven}-\eqref{e3.35Keven}.
To do this,  we  will distinguish the following two cases: (I) $\ell$ is odd, and (II)  $\ell$ is even. 
\begin{itemize}
\item[(I)] Let  $\ell$  be odd.  In this case, we will  first show that there exists a matrix  $[V^{(\ell-2)}]_{\gamma_{\ell}+1} \in \mathcal{M}_{\Lambda_{\gamma_{\ell}+1}\times n}(\mathcal{T}_{m}) $ satisfying the  following matrix equations:  
\vspace{-2mm}\begin{eqnarray}
  \mathtt{H}^{(\ell-2)}+[T^{(0)}]_{\gamma_{\ell}+1}[\mathtt{U}^{(\ell-2)}]_{\gamma_{\ell}+1}^t+[\mathtt{U}^{(\ell-2)}]_{\gamma_{\ell}+1}[T^{(0)}]_{\gamma_{\ell}+1}^t &\equiv & \textbf{0} \pmod{u},\label{e3.40p3.5Keven}\\
  \mathcal{D}iag\left(\mathtt{E}^{(2\ell-4)}+[\mathtt{U}^{(\ell-2)}]_{s-2\mathrm{f}_\ell+2}[\mathtt{U}^{(\ell-2)}]_{s-2\mathrm{f}_\ell+2}^t\right) &\equiv & \textbf{0} \pmod{u}\label{e3.41p3.5Keven},\\
\mathcal{D}iag\left(\mathtt{E}^{(\ell+\kappa-2)}+\eta_0[T^{(0)}]_{\gamma_{\ell}+1-\kappa_1} [\mathtt{U}^{(\ell-2)}]_{\gamma_{\ell}+1-\kappa_1}^t\right) &\equiv & \textbf{0} \pmod{u}.\label{e3.42p3.5Keven}
\vspace{-2mm}\end{eqnarray} 
To do this, let  $[T^{(0)}]_{\gamma_{\ell}+1}=(\textbf{a}_i)$ and $[\mathtt{U}^{(\ell-2)}]_{\gamma_{\ell}+1}=(\textbf{z}_j),$ where $\textbf{a}_i$\rq{}s and $\textbf{z}_j$\rq{}s denote the rows of the matrices $[T^{(0)}]_{\gamma_{\ell}+1}$ and $[\mathtt{U}^{(\ell-2)}]_{\gamma_{\ell}+1},$ respectively. 
Let $\mathtt{H}^{(\ell-2)}_{i,j},\mathtt{E}^{(2\ell-4)}_{i,j},\mathtt{E}^{(\ell+\kappa-2)}_{i,j}\in \mathcal{T}_{m}$ denote the $(i,j)$-th entries of the matrices  $\mathtt{H}^{(\ell-2)},$ $ \mathtt{E}^{(2\ell-4)},$ $\mathtt{E}^{(\ell+\kappa-2)},$ respectively. Now, the system  of  matrix equations \eqref{e3.40p3.5Keven}-\eqref{e3.42p3.5Keven} is  equivalent to the following system of equations in unknowns $ \textbf{z}_1, \textbf{z}_2,\ldots, \textbf{z}_{\Lambda_{\gamma_{\ell}+1}}$ over $\mathcal{T}_{m}$:  \vspace{-2mm}\begin{eqnarray}
\textbf{a}_i\cdot \textbf{z}_j+\textbf{a}_j \cdot \textbf{z}_i  &\equiv& \mathtt{H}^{(\ell-2)}_{i,j} \pmod{u} \text{ ~~for } 1 \leq i <j \leq \Lambda_{\gamma_{\ell}+1},\label{e3.6p3.5Keven}\\ 
\textbf{1}\cdot \textbf{z}_i & \equiv & (\mathtt{E}^{(2\ell-4)}_{i,i})^{2^{m-1}} \pmod{u} \text{ ~ for } 1 \leq i \leq \Lambda_{s-2\mathrm{f}_\ell+2}, \text{ ~and} \label{e3.7p3.5Keven}\\
\textbf{a}_i \cdot \textbf{z}_i & \equiv & \eta_0^{-1}\mathtt{E}_{i,i}^{(\ell+\kappa-2)} \pmod{u} \text{ ~~for } 1 \leq i \leq \Lambda_{\gamma_{\ell}+1-\kappa_1}. \label{e3.8p3.5Keven}
\vspace{-2mm}\end{eqnarray} 
For each integer $j$ satisfying $1 \leq j \leq \Lambda_{\gamma_{\ell}+1},$ it is easy to see that there exists a unique integer $r_{j}$ satisfying $1 \leq r_{j}\leq \gamma_{\ell}+1$ and $\Lambda_{r_{j}-1}+1 \leq j \leq \Lambda_{r_{j}}.$ The corresponding unknown vector $\textbf{z}_{j}$ can be written as $\textbf{z}_{j}=(\textbf{0}~ \textbf{z}_{j}^{n-\Lambda_{r_{j}+\ell-2}}),$ where $\textbf{0}$ denotes the zero vector of 
length $\Lambda_{r_{j}+\ell-2}$ and $\textbf{z}_{j}^{n-\Lambda_{r_{j}+\ell-2}}$ denotes the vector of length $n-\Lambda_{r_{j}+\ell-2}$ obtained   by omitting the first $\Lambda_{r_{j}+\ell-2}$ coordinates of $\textbf{z}_{j}$. This implies,  for $\Lambda_{r_{j}-1}+1 \leq j \leq \Lambda_{r_{j}}$, that the first $\Lambda_{r_{j}+\ell-2}$ coordinates of $\textbf{z}_{j}$ are zero,  leaving  $n-\Lambda_{r_{j}+\ell-2}$   variables
 in $\textbf{z}_{j}.$  For $1 \leq j \leq \Lambda_{\gamma_{\ell}+1},$ let $\widetilde{\textbf{z}}_{j}=\textbf{z}_{j}^{n-\Lambda_{r_{j}+\ell-2}}$ 
and  $\widetilde{\textbf{a}}_{j}=\textbf{a}_{j}^{n-\Lambda_{r_{j}+\ell-2}}$ denote  the  vectors of length $n-\Lambda_{r_{j}+\ell-2}$ obtained from $\textbf{z}_{j}$ and  $\textbf{a}_{j}$ by omitting the first $\Lambda_{r_{j}+\ell-2}$ coordinates, respectively. 
Consequently, the system   of  equations \eqref{e3.6p3.5Keven}-\eqref{e3.8p3.5Keven} is equivalent to the following system of equations in unknowns $\widetilde{\textbf{z}}_1, \widetilde{\textbf{z}}_2,\ldots, \widetilde{\textbf{z}}_{\Lambda_{\gamma_{\ell}+1}}$ over $\mathcal{T}_{m}$:   
\vspace{-2mm}\begin{eqnarray}
\widetilde{\textbf{a}}_i\cdot \widetilde{\textbf{z}}_j+\widetilde{\textbf{a}}_j \cdot \widetilde{\textbf{z}}_i  &\equiv& \mathtt{H}^{(\ell-2)}_{i,j} \pmod{u} \text{ ~~for } 1 \leq i <j \leq \Lambda_{\gamma_{\ell}+1},\label{e3.9p3.5Keven}\\ 
\widetilde{\textbf{1}}\cdot \widetilde{\textbf{z}}_i & \equiv & (\mathtt{E}^{(2\ell-4)}_{i,i})^{2^{m-1}} \pmod{u} \text{~~ for } 1 \leq i \leq \Lambda_{s-2\mathrm{f}_\ell+2},\text{ ~and}\label{e3.10p3.5Keven}\\
\widetilde{\textbf{a}}_i \cdot \widetilde{\textbf{z}}_i   & \equiv & \eta_0^{-1}\mathtt{E}_{i,i}^{(\ell+\kappa-2)} \pmod{u} \text{ ~~for } 1 \leq i \leq \Lambda_{\gamma_{\ell}+1-\kappa_1}, \label{e3.11p3.5Keven}
\vspace{-2mm}\end{eqnarray} 
 where $\widetilde{\textbf{1}}$ denotes the all-one vector having the same length as that of  $\widetilde{\textbf{z}}_i$ for each $i.$ Further, the above system of equations \eqref{e3.9p3.5Keven}-\eqref{e3.11p3.5Keven} can  be expressed in the matrix form
\begin{equation}\label{e3.12p3.5Keven}M\begin{bmatrix}
	\widetilde{\textbf{z}}_1^{t}\\ \widetilde{\textbf{z}}_2^{t}\\ \widetilde{\textbf{z}}_3^{t}\\ \vdots\\ \vdots\\ \widetilde{\textbf{z}}_{\Lambda_{\gamma_{\ell}+1-\kappa_1}}^{t}\vspace{0.50mm}\\ \vdots\\ \widetilde{\textbf{z}}_{\Lambda_{s-2\mathrm{f}_\ell+2}+1}^{t} \\ \vdots\\ \widetilde{\textbf{z}}_{\Lambda_{\gamma_{\ell}+1}-2}^{t}\\ \widetilde{\textbf{z}}_{\Lambda_{\gamma_{\ell}+1}-1}^{t} \\ \vspace{0.75mm}\widetilde{\textbf{z}}_{\Lambda_{\gamma_{\ell}+1}}^{t}
	\end{bmatrix}\equiv\begin{bmatrix}
	(\mathtt{E}_{1,1}^{(2\ell-4)})^{2^{m-1}}\\ \vdots\\ (\mathtt{E}_{\Lambda_{s-2\mathrm{f}_\ell+2},\Lambda_{s-2\mathrm{f}_\ell+2}}^{(2\ell-4)})^{2^{m-1}}\vspace{0.5mm}\\ \eta_0^{-1}\mathtt{E}^{(\ell+\kappa-2)}_{1,1}\vspace{-0.25mm}\\ \vdots\vspace{-0.25mm}\\ \eta_0^{-1}\mathtt{E}^{(\ell+\kappa-2)}_{\Lambda_{\gamma_{\ell}+1-\kappa_1},\Lambda_{\gamma_{\ell}+1-\kappa_1}}\vspace{0.5mm}\\\mathtt{H}^{(\ell-2)}_{1,2}\\\vdots\\\mathtt{H}^{(\ell-2)}_{1,\Lambda_{\gamma_{\ell}+1}}\vspace{-0.25mm}\\ \vdots\vspace{-0.25mm} \\\mathtt{H}^{(\ell-2)}_{\Lambda_{\gamma_{\ell}+1}-1,\Lambda_{\gamma_{\ell}+1}}
	\end{bmatrix}\pmod u,\end{equation}
where the matrix $M$ has order $\Big( \Lambda_{s-2\mathrm{f}_\ell+2}+\Lambda_{\gamma_{\ell}+1-\kappa_1}+\frac{\Lambda_{\gamma_{\ell}+1}(\Lambda_{\gamma_{\ell}+1}-1)}{2}\Big) \times \Big(\sum\limits_{i=\ell}^{\gamma_{\ell}+\ell}\lambda_i\Lambda_{i-\ell+1}+\Lambda_{\gamma_{\ell}+1}(n-\Lambda_{\gamma_{\ell}+\ell})\Big)$ and is given by \vspace{-2mm}\begin{equation*}
    M=\begin{bmatrix}
		\widetilde{\textbf{1}} &&&\\
    	&\widetilde{\textbf{1}} &&&\\
	& & \ddots  &  \\
	& & & \widetilde{\textbf{1}}\\
    	& & &&\ddots  &  \\
    & & & && \widetilde{\textbf{1}}\\
	\widetilde{\textbf{a}}_1 &&\\
    &\widetilde{\textbf{a}}_2 &&\\
	& & \ddots  &  \\
	& & &  \widetilde{\textbf{a}}_{\Lambda_{\gamma_{\ell}+1-\kappa_1}}&\\
\widetilde{\textbf{a}}_{2}&\widetilde{\textbf{a}}_{1}&&& \\
	\vdots&\vdots  & \vdots \\
    \widetilde{\textbf{a}}_{\Lambda_{s-2\mathrm{f}_\ell+2}}&&&&&\widetilde{\textbf{a}}_{1}&\\ 
    \vdots&\vdots  & \vdots \\
\widetilde{\textbf{a}}_{\Lambda_{\gamma_{\ell}+1}}&&&&&\cdots&&\widetilde{\textbf{a}}_{1} \\
	&&&\vdots&\vdots&\cdots&   &\vdots\\
	& && &&\cdots&\widetilde{\textbf{a}}_{\Lambda_{\gamma_{\ell}+1}}&\widetilde{\textbf{a}}_{\Lambda_{\gamma_{\ell}+1}-1}\\
	\end{bmatrix}.\vspace{-1mm}
\end{equation*}    Now,  working as in Lemma 4.1 of Yadav and Sharma \cite{Galois}, we observe that the rows of the matrix $M$  are linearly independent over $\mathcal{T}_{m}$  if and only if $\textbf{1}\notin \mathcal{D}^{(\gamma_{\ell}+1-\kappa_1)}.$ Accordingly, we will consider the following two cases separately:  (1) $\textbf{1}\notin \mathcal{D}^{(\gamma_{\ell}+1-\kappa_1)},$ and (2)  $\textbf{1}\in \mathcal{D}^{(\gamma_{\ell}+1-\kappa_1)}.$

\begin{itemize}\item[(1)] Let  $\textbf{1}\notin \mathcal{D}^{(\gamma_{\ell}+1-\kappa_1)}.$ In this case,  the  rows of the matrix $M$ are linearly independent over $\mathcal{T}_{m},$   and hence  the row-rank of the matrix $M$ is  $\Lambda_{s-2\mathrm{f}_\ell+2}+\Lambda_{\gamma_{\ell}+1-\kappa_1}+\frac{\Lambda_{\gamma_{\ell}+1}(\Lambda_{\gamma_{\ell}+1}-1)}{2}.$ Thus, the matrix equation \eqref{e3.12p3.5Keven} always has a solution. Moreover, 
   the number of solutions  of  the system \eqref{e3.12p3.5Keven}, and hence the number of choices for the matrix  $[\mathtt{U}^{(\ell-2)}]_{\gamma_{\ell}+1} \in \mathcal{M}_{\Lambda_{\gamma_{\ell}+1}\times n}(\mathcal{T}_{m})$ is  given by 
\vspace{-1mm}\begin{equation*}
   \displaystyle (2^m)^{\sum\limits_{i=\ell}^{\gamma_{\ell}+\ell }\lambda_i\Lambda_{i-\ell+1}+\Lambda_{\gamma_{\ell}+1}(n-\Lambda_{\gamma_{\ell}+\ell})-\Lambda_{s-2\mathrm{f}_\ell+2}-\Lambda_{\gamma_{\ell}+1-\kappa_1}-\frac{\Lambda_{\gamma_{\ell}+1}(\Lambda_{\gamma_{\ell}+1}-1)}{2} }. 
\vspace{-1mm}\end{equation*} 
\item[(2)] Let $\textbf{1}\in \mathcal{D}^{(\gamma_{\ell}+1-\kappa_1)}.$ In this case, $\textbf{1}$ belongs to the $\mathcal{T}_m$-span of the vectors $\textbf{a}_1,\textbf{a}_2,\ldots, \textbf{a}_{\Lambda_{\gamma_{\ell}+1-\kappa_1}}.$  Since the  vectors $\textbf{a}_1,\textbf{a}_2,\ldots, \textbf{a}_{\Lambda_{\gamma_{\ell}+1-\kappa_1}}$  are linearly independent over $\mathcal{T}_{m},$ there exist unique scalars $\beta_{1},\beta_{2},\ldots,\beta_{\Lambda_{\gamma_{\ell}+1-\kappa_1}} \in \mathcal{T}_{m}$ such that $\beta_{1} \textbf{a}_{1}+\beta_{2} \textbf{a}_{2}+\cdots+\beta_{\Lambda_{\gamma_{\ell}+1-\kappa_1}} \textbf{a}_{\Lambda_{\gamma_{\ell}+1-\kappa_1}} \equiv  \textbf{1}\pmod u.$ 
Moreover, all  rows of the matrix $M$ except the last row are linearly independent over $\mathcal{T}_{m}.$ This implies that the row-rank of the matrix $M$ is  $\Lambda_{\gamma_{\ell}+1-\kappa_1}+\Lambda_{s-2\mathrm{f}_\ell+2}+\frac{\Lambda_{\gamma_{\ell}+1}(\Lambda_{\gamma_{\ell}+1}-1)}{2}-1.$ Further, it is easy to see that   the matrix equation \eqref{e3.12p3.5Keven} has a solution if and only if 
 \vspace{-1mm}\begin{equation}\label{e3.13p3.5Keven}
   \displaystyle   \sum\limits_{b=1}^{\Lambda_{\gamma_{\ell}+1-\kappa_1}}\beta_b(\mathtt{E}^{(2\ell-4)}_{b,b})^{2^{m-1}} +\eta_0^{-1}\sum\limits_{g=1}^{\Lambda_{\gamma_{\ell}+1-\kappa_1}}\beta_g^2 \mathtt{E}^{(\ell+\kappa-2)}_{g,g} + \hspace{-4mm}\sum\limits_{~~~~~1\leq i< j \leq \Lambda_{\gamma_{\ell}+1-\kappa_1}}\hspace{-4mm}\beta_i \beta_j \mathtt{H}_{i,j}^{(\ell-2)}~ \equiv ~0\pmod{u}.
\vspace{-1mm} \end{equation}
  We next assert that the equation  \eqref{e3.13p3.5Keven} holds (\textit{i.e.,} the matrix equation \eqref{e3.12p3.5Keven}  has a solution) if and only if one of the following five conditions is satisfied: (i)  $n\equiv 0,4 \pmod 8 ,$  (ii)  $n\equiv 2,6\pmod 8,$ $4 \leq \ell \leq \kappa_1+1$ and $\eta_{\ell-2}=0$,   (iii) $n\equiv 2,6 \pmod 8 ,$ $\ell= \kappa_1+2$   and  $(\eta_0)^{\frac{3}{2}}=\eta_{\kappa_1},$      (iv)   $n\equiv 2,6\pmod 8$, $2\kappa > e,$ $ 4 \leq \ell \leq \min\{e-\kappa,\kappa_1+1\}$   and $\eta_{\ell-2}=0$, and (v)  $n\equiv 2,6 \pmod 8,$ $2\kappa > e \geq \frac{3}{2}\kappa+2,$  $\ell= \kappa_1+2$     and  $(\eta_0)^{\frac{3}{2}}=\eta_{\kappa_1}.$ To prove this assertion, let us  assume that $[\mathtt{U}^{(b)}]_{\gamma_{\ell}+1} = (\textbf{z}^{(b)}_j),$  where  $\textbf{z}^{(b)}_j$\rq{}s denote the rows of the matrix $[\mathtt{U}^{(b)}]_{\gamma_{\ell}+1}$  for $1 \leq b \leq \ell-3.$ Further, let us write \vspace{-1mm}\begin{equation*}\textbf{1} \equiv \sum\limits_{j=1}^{\Lambda_{\gamma_{\ell}+1-\kappa_1}} \hspace{-2mm}\beta_j \left(\textbf{a}_j+u\textbf{z}^{(1)}_j+ \cdots+u^{\ell-3}\textbf{z}^{(\ell-3)}_j\right)+ u^{\ell-2}\textbf{d}_{\ell-2}+\cdots+u^{\ell+\kappa-2} \textbf{d}_{\ell+\kappa-2} \pmod{u^{\ell+\kappa-1}}\vspace{-1mm}\end{equation*} for some $\textbf{d}_{\ell-2},\textbf{d}_{\ell-1},\ldots,\textbf{d}_{\ell+\kappa-2} \in \mathcal{T}_{m}^n.$ Note that  $\textbf{1}\cdot \textbf{1} \equiv n \pmod{ u^{\ell+\kappa-1}},$ which 
 implies that $\textbf{1}\cdot \textbf{1} \equiv 0 \pmod{ u^{\ell+\kappa-1}} $ when $ n\equiv 0,4\pmod 8,$ while 
$\textbf{1}\cdot \textbf{1} \equiv u^{\kappa}(\eta_0+u\eta_1+u^2\eta_2+\cdots+u^{\ell-2}\eta_{\ell-2}) \pmod{ u^{\ell+\kappa-1}} $ when $ n\equiv 2,6\pmod  8.$  

When either $2\kappa \leq e$ and $\kappa_1+3 \leq \ell \leq \kappa,$ or $2\kappa >e$ and $\kappa_1+3 \leq \ell \leq e-\kappa,$  we obtain $\eta_0=0,$ which is a contradiction. This implies that $\textbf{1}\notin \mathcal{D}^{(\gamma_{\ell}+1-\kappa_1)}$ when either $2\kappa \leq e$ and $\kappa_1+3 \leq \ell \leq \kappa,$ or $2\kappa >e$ and $\kappa_1+3 \leq \ell \leq e-\kappa$ holds.

Further, when either $ 4 \leq \ell \leq \kappa_1+2$ and $2\kappa \leq e,$ or  $ 4 \leq \ell \leq \min\{\kappa_1+2,e-\kappa\}$ and $2\kappa > e$ holds, 
 we get 
\vspace{-1mm}\begin{equation}\label{eta_0p3.5Keven}
 \sum\limits_{b=1}^{\Lambda_{\gamma_{\ell}+1-\kappa_1}}\beta_b^2 \mathtt{E}_{b,b}^{(2\ell-4)}+\textbf{d}_{\ell-2}\cdot \textbf{d}_{\ell-2}  \equiv \epsilon_1 \eta_0 \pmod{u}   \vspace{-1mm}\end{equation}       and \vspace{-1mm}\begin{equation}\label{eta_1p3.5Keven}
   \eta_0^{-1} \sum\limits_{g=1}^{\Lambda_{\gamma_{\ell}+1-\kappa_1}}\beta_g^2 \mathtt{E}^{(\ell+\kappa-2)}_{g,g}+ \hspace{-5mm}\sum\limits_{~~~~~1\leq i< j \leq \Lambda_{\gamma_{\ell}+1-\kappa_1}}\hspace{-5mm}\beta_i \beta_j \mathtt{H}_{i,j}^{(\ell-2)}+\textbf{1}\cdot \textbf{d}_{\ell-2} \equiv \epsilon_2 \eta_0^{-1}\eta_{\ell-2}\pmod{u}, \vspace{-1mm} 
\vspace{-1mm} \end{equation} where  \begin{equation*}\vspace{-1mm}\epsilon_1=\left\{\begin{array}{ll}0 & \text{if either }  n\equiv 0,4\pmod{8} \text{ or }n\equiv 2,6 \pmod 8 \text{~ with either ~}\ell \leq \kappa_1+1 \text{ and } \\ &2\kappa \leq e, \text{ or } \ell \leq \min\{e-\kappa,\kappa_1+1\} \text{ and }2\kappa > e ;\\ 1 & \text{if }n\equiv 2,6 \pmod 8 \text{ and }\ell=\kappa_1+2\text{ with either }  2\kappa \leq e \text{ or }  2 \kappa > e\geq \frac{3}{2}\kappa+2  \end{array}\right.\end{equation*}
 and \begin{equation*}
  \vspace{-1mm}   
 \epsilon_2=\left\{\begin{array}{ll}0 & \text{if }  n\equiv 0,4\pmod{8} ;\\ 1 & \text{if }n\equiv 2,6 \pmod 8 .\end{array}\right.\end{equation*}
 By \eqref{eta_0p3.5Keven} and \eqref{eta_1p3.5Keven},  one can easily see that the above assertion holds. This shows that the matrix equation \eqref{e3.12p3.5Keven}  has a solution if and only if exactly one of the following holds: (i)  $n\equiv 0,4 \pmod 8,$ (ii) $n\equiv 2,6\pmod 8,$ $4 \leq \ell \leq \kappa_1+1$  and $\eta_{\ell-2}=0$,    (iii)  $n\equiv 2,6 \pmod 8,$ $\ell= \kappa_1+2$ and  $(\eta_0)^{\frac{3}{2}}=\eta_{\kappa_1},$  (iv)   $n\equiv 2,6\pmod 8$, $2\kappa > e,$ $ 4 \leq \ell \leq \min\{e-\kappa,\kappa_1+1\}$   and $\eta_{\ell-2}=0$, and (v) $n\equiv 2,6 \pmod 8,$ $2\kappa > e\geq \frac{3}{2}\kappa+2,$  $\ell= \kappa_1+2$    and  $(\eta_0)^{\frac{3}{2}}=\eta_{\kappa_1}.$
  Furthermore, when exactly  one of the above conditions (i) - (v) holds,
  the number of solutions of the matrix equation \eqref{e3.12p3.5Keven}, and hence the number of choices for the matrix $[\mathtt{U}^{(\ell-2)}]_{\gamma_{\ell}+1} \in \mathcal{M}_{\Lambda_{\gamma_{\ell}+1} \times n } (\mathcal{T}_{m})$ is given by 
 \vspace{-1mm} \begin{equation*}\vspace{-1mm}
   \displaystyle (2^m)^{\sum\limits_{i=\ell}^{\gamma_{\ell}+\ell}\lambda_i\Lambda_{i-\ell+1}+\Lambda_{\gamma_{\ell}+1}(n-\Lambda_{\gamma_{\ell}+\ell})-\Lambda_{s-2\mathrm{f}_\ell+2}-\Lambda_{\gamma_{\ell}+1-\kappa_1}-\frac{\Lambda_{\gamma_{\ell}+1}(\Lambda_{\gamma_{\ell}+1}-1)}{2}+1}. 
\end{equation*}
Next, for a given choice of the matrix $[\mathtt{U}^{(\ell-2)}]_{\gamma_{\ell}+1}$  satisfying \eqref{e3.40p3.5Keven}-\eqref{e3.42p3.5Keven},  we get 
\begin{equation}\label{e3.43}
 \mathtt{H}^{(\ell-2)}+[T^{(0)}]_{\gamma_{\ell}+1}[\mathtt{U}^{(\ell-2)}]_{\gamma_{\ell}+1}^t+[\mathtt{U}^{(\ell-2)}]_{\gamma_{\ell}+1}[T^{(0)}]_{\gamma_{\ell}+1}^t ~\equiv ~ uR_2 \pmod{u^2}
\end{equation}
for some  $R_2\in Sym_{\Lambda_{\gamma_{\ell}+1}}(\mathcal{T}_{m}) .$ 
This,  by \eqref{e3.30Keven}, gives
\begin{eqnarray}
    R_2+\mathtt{H}^{(\ell-1)}+[T^{(0)}]_{\gamma_{\ell}+1}[\mathtt{U}^{(\ell-1)}]_{\gamma_{\ell}+1}^t+[\mathtt{U}^{(\ell-1)}]_{\gamma_{\ell}+1}[T^{(0)}]_{\gamma_{\ell}+1}^t&&\nonumber\\+[\mathtt{U}^{(1)}]_{\gamma_{\ell}+1}[\mathtt{U}^{(\ell-2)}]_{\gamma_{\ell}+1}^t+[\mathtt{U}^{(\ell-2)}]_{\gamma_{\ell}+1}[\mathtt{U}^{(1)}]_{\gamma_{\ell}+1}^t&\equiv& \textbf{0}\pmod{u}.\label{e3.44Keven}
    \end{eqnarray}
Now, working as in Lemma 4.1  of Yadav and Sharma   \cite{quasi}, we observe that there exists a matrix $[\mathtt{U}^{(\ell-1)}]_{\gamma_{\ell}+1}$ satisfying  \eqref{e3.32Keven} and \eqref{e3.44Keven} and that such a matrix $[\mathtt{U}^{(\ell-1)}]_{\gamma_{\ell}+1}$ has precisely \begin{equation*}
    \displaystyle (2^m)^{\sum\limits_{i=\ell+1}^{\gamma_{\ell}+\ell }\lambda_i\Lambda_{i-\ell}+\Lambda_{\gamma_{\ell}+1}(n-\Lambda_{\gamma_{\ell}+\ell})-\Lambda_{s-2\mathrm{f}_\ell+1}-\frac{\Lambda_{\gamma_{\ell}+1}(\Lambda_{\gamma_{\ell}+1}-1)}{2}}
\end{equation*} distinct choices. 
 \end{itemize} Finally, working as in Proposition 4.3 of Yadav and Sharma \cite{Galois}, we see that the  matrices 
$\Big[
  ~\textbf{0}~ \cdots ~\textbf{0} ~ \\ \mathtt{A}_{\gamma_{\ell}+g,\gamma_{\ell}+\ell }^{(\ell-g)} ~\cdots ~ \mathtt{A}_{\gamma_{\ell}+g,e}^{(\ell-g)}  
\Big] $ for $2 \leq g \leq \ell-1$ and $T_{\gamma_{\ell}+\ell }^{(\ell)}$
satisfying \eqref{e3.34Keven} and \eqref{e3.35Keven} have  precisely \vspace{-1mm}\begin{equation*} (2^m)^{(\Lambda_{\gamma_{\ell}+\ell -1}-\Lambda_{\gamma_{\ell}+1})(n-\Lambda_{\gamma_{\ell}+\ell }-\Lambda_{\gamma_{\ell}+1})}{\lambda_{\gamma_{\ell}+\ell }+n-\Lambda_{\gamma_{\ell}+\ell }-\Lambda_{\gamma_{\ell}+1} \brack \lambda_{\gamma_{\ell}+\ell }}_{2^m}\vspace{-1mm}\end{equation*} distinct choices. 
From this,  the desired result follows in the case when $\ell$ is odd. 
  \item[(II)] For even   $\ell,$ the desired result follows working as in case (I).  \vspace{-2mm}\end{itemize}\vspace{-2mm}\vspace{-2mm}
 \end{proof}

In the following proposition, we assume that $e$ is an odd integer and $\kappa \geq 4$ is an even integer satisfying $2\kappa\leq e.$ Here, we provide a method to  lift a self-orthogonal code   of type $\{\Lambda_{s-\kappa_1+2},\lambda_{s-\kappa_1+3}, \ldots, \lambda_{s+\kappa_1}\}$   and  length $n$ over $\mathcal{R}_{\kappa-1,m}$ satisfying  property $(\mathfrak{P})$ to a  self-orthogonal code  of type $\{\Lambda_{s-\kappa_1+1},\lambda_{s-\kappa_1+2}, \ldots, \lambda_{s+\kappa_1+1}\}$ and length $n$ over $\mathcal{R}_{\kappa+1,m}$ satisfying  property $(\mathfrak{P}).$    We also count the total number of distinct ways in which this lifting can be performed.
\begin{proposition}\label{p3.5Keven} Let $e$ be an odd integer and $\kappa\geq 4$ be an even integer satisfying $2\kappa\leq e.$   Let $\mathcal{C}_{\kappa-1}$ be a self-orthogonal code of type $\{\Lambda_{s-\kappa_1+2},\lambda_{s-\kappa_1+3}, \ldots, \lambda_{s+\kappa_1}\}$   and  length $n$ over $\mathcal{R}_{\kappa-1,m}$ satisfying  property $(\mathfrak{P}).$ 
Let $\mathcal{D}^{(s-\kappa+1)} \subseteq \mathcal{D}^{(s-\kappa_1+1)}$ be two   linear subcodes of  $Tor_{1}(\mathcal{C}_{\kappa-1})$ with $\dim \mathcal{D}^{(s-\kappa_1+1)}=\Lambda_{s-\kappa_1+1}$ and $\dim \mathcal{D}^{(s-\kappa+1)}=\Lambda_{s-\kappa+1} .$  
The following hold. \begin{enumerate}\item[(a)] There exists a  self-orthogonal code $\mathcal{C}_{\kappa+1}$ of type $\{\Lambda_{s-\kappa_1+1},\lambda_{s-\kappa_1+2}, \ldots, \lambda_{s+\kappa_1+1}\}$ and length $n$ over $\mathcal{R}_{\kappa+1,m}$ satisfying  property $(\mathfrak{P})$ and $Tor_1(\mathcal{C}_{\kappa+1})=\mathcal{D}^{(s-\kappa_1+1)}$ and    $Tor_{j+1}(\mathcal{C}_{\kappa+1})=Tor_{j}(\mathcal{C}_{\kappa-1})$ for $1 \leq j \leq  \kappa-1,$ where such a code exists: \vspace{-1mm}\begin{itemize} \item unconditionally, if $\textbf{1}\notin \mathcal{D}^{(s-\kappa+1)};$
\item if and only if either $n \equiv 0 \pmod 8$ or $n\equiv 4 \pmod 8 $ and $m$ even when $\textbf{1}\in \mathcal{D}^{(s-\kappa+1)}.$     
 \end{itemize}
\item[(b)] Moreover,  each such triplet  $(\mathcal{C}_{\kappa-1}, \mathcal{D}^{(s-\kappa_1+1)},\mathcal{D}^{(s-\kappa+1)})$ of codes gives rise to precisely	 	
\vspace{-1mm}\begin{equation*}\vspace{-1mm}
2^{\epsilon}(2^m)^{\sum\limits_{i=\kappa+1}^{s+\kappa_1+1}\lambda_i\Lambda_{i-\kappa}+\sum\limits_{j=\kappa+2}^{s+\kappa_1+1}\lambda_j\Lambda_{j-\kappa-1}+Y_{\kappa+1}}  {\lambda_{s+\kappa_1+1}+n-\Lambda_{s+\kappa_1+1}-\Lambda_{s-\kappa_1+1} \brack \lambda_{s+\kappa_1+1}}_{2^m}
\end{equation*}  distinct  self-orthogonal codes of type $\{\Lambda_{s-\kappa_1+1},\lambda_{s-\kappa_1+2}, \ldots, \lambda_{s+\kappa_1+1}\}$ and length $n$ over $\mathcal{R}_{\kappa+1,m}$ satisfying  property $(\mathfrak{P})$ and $Tor_1(\mathcal{C}_{\kappa+1})=\mathcal{D}^{(s-\kappa_1+1)}$ and     $Tor_{j+1}(\mathcal{C}_{\kappa+1})=Tor_{j}(\mathcal{C}_{\kappa-1})$ for $1 \leq j \leq  \kappa-1,$ where  \vspace{-1mm}\begin{equation*}  \vspace{-1mm} Y_{\kappa+1}=(\Lambda_{s+\kappa_1}+\Lambda_{s-\kappa_1+1})(n-\Lambda_{s+\kappa_1+1}-\Lambda_{s-\kappa_1+1})+\Lambda_{s-\kappa_1+1}^2+ \Lambda_{s-\kappa_1+1}-2\Lambda_{s-\kappa+1}-\Lambda_{s-\kappa+2}  +\epsilon\end{equation*} with  $\epsilon=1$ if $\textbf{1} \in \mathcal{D}^{(s-\kappa+1)},$  while $\epsilon=0$ otherwise.
 \end{enumerate} \end{proposition}
 \begin{proof} Working  as in Proposition 3.4 of Yadav and Sharma \cite{YSub2}, the desired result follows.\end{proof}

 In the following proposition, we assume that $e$ is an even integer satisfying $2\kappa\leq e.$  Here, we  provide a method to  lift a self-orthogonal code  of type  $\{\Lambda_{s-\kappa_1+1},\lambda_{s-\kappa_1+2}, \ldots, \lambda_{s+\kappa_1}\}$   and  length $n$ over $\mathcal{R}_{\kappa,m}$ satisfying   property $(\mathfrak{P}) $ to a  self-orthogonal code  of type $\{\Lambda_{s-\kappa_1},\lambda_{s-\kappa_1+1}, \ldots, \lambda_{s+\kappa_1+1}\}$ and length $n$ over $\mathcal{R}_{\kappa+2,m}$ satisfying  property $(\mathfrak{P}).$    We also compute  the total number of  distinct ways to carry out  this lifting. \begin{proposition}\label{p3.6Keven} Let $e$ be an even integer satisfying $2\kappa\ \leq e.$   Let $\mathcal{C}_{\kappa}$ be a self-orthogonal code of type $\{\Lambda_{s-\kappa_1+1},\\\lambda_{s-\kappa_1+2}, \ldots, \lambda_{s+\kappa_1}\}$   and  length $n$ over $\mathcal{R}_{\kappa,m}$ satisfying  property $(\mathfrak{P}).$  Let $\mathcal{D}^{(s-\kappa_1)}\supseteq \mathcal{D}^{(s-\kappa)}$ be two  linear subcodes of $Tor_{1}(\mathcal{C}_{\kappa})$ with $\dim \mathcal{D}^{(s-\kappa_1)}=\Lambda_{s-\kappa_1}$ and $\dim \mathcal{D}^{(s-\kappa)}=\Lambda_{s-\kappa}.$ 
The following hold.
\begin{enumerate}\item[(a)] There exists a  self-orthogonal code $\mathcal{C}_{\kappa+2}$ of type $\{\Lambda_{s-\kappa_1},\lambda_{s-\kappa_1+1}, \ldots, \lambda_{s+\kappa_1+1}\}$ and length $n$ over $\mathcal{R}_{\kappa+2,m}$ satisfying  property $(\mathfrak{P})$ and $Tor_{1}(\mathcal{C}_{\kappa+2})=\mathcal{D}^{(s-\kappa_1)}$ and     $Tor_{j+1}(\mathcal{C}_{\kappa+2})=Tor_{j}(\mathcal{C}_{\kappa})$ for $1 \leq j \leq  \kappa,$ where such a code exists:
 \vspace{-1mm}\begin{itemize}
     \item unconditionally, if $\textbf{1}\notin \mathcal{D}^{(s-\kappa)};$\item if and only if either $n \equiv 0 \pmod 8$ or $n \equiv 4 \pmod{8}$ and $m$  even when $\textbf{1}\in \mathcal{D}^{(s-\kappa)}.$
 \end{itemize}
\item[(b)] Moreover,  each such triplet  $(\mathcal{C}_{\kappa},\mathcal{D}^{(s-\kappa_1)},\mathcal{D}^{(s-\kappa)})$  of codes  gives rise to precisely 
\begin{eqnarray*}
2^{\epsilon}(2^m)^{\sum\limits_{i=\kappa+2}^{s+\kappa_1+1}\lambda_i\Lambda_{i-\kappa-1}+\sum\limits_{j=\kappa+3}^{s+\kappa_1+1}\lambda_j\Lambda_{j-\kappa-2}+Y_{\kappa+2}}  {\lambda_{s+\kappa_1+1}+n-\Lambda_{s+\kappa_1+1}-\Lambda_{s-\kappa_1} \brack \lambda_{s+\kappa_1+1}}_{2^m}
\end{eqnarray*}  distinct  self-orthogonal codes $\mathcal{C}_{\kappa+2}$ of type $\{\Lambda_{s-\kappa_1},\lambda_{s-\kappa_1+1}, \ldots, \lambda_{s+\kappa_1+1}\}$ and length $n$ over $\mathcal{R}_{\kappa+2,m}$ satisfying  property $(\mathfrak{P})$ and $Tor_{1}(\mathcal{C}_{\kappa+2})=\mathcal{D}^{(s-\kappa_1)}$ and   $Tor_{j+1}(\mathcal{C}_{\kappa+2})=Tor_{j}(\mathcal{C}_{\kappa})$ for $1 \leq j \leq  \kappa,$ where $$Y_{\kappa+2}=(\Lambda_{s+\kappa_1}+\Lambda_{s-\kappa_1})(n-\Lambda_{s+\kappa_1+1}-\Lambda_{s-\kappa_1})+\Lambda_{s-\kappa_1}^2+\Lambda_{s-\kappa_1}-2\Lambda_{s-\kappa},$$ and $\epsilon=1$ if $\textbf{1} \in \mathcal{D}^{(s-\kappa)},$ 
while $\epsilon=0$ otherwise. 
 \end{enumerate} \end{proposition}
  \begin{proof} Working  as in Proposition 3.4 of Yadav and Sharma \cite{YSub2}, the desired result follows.\vspace{-2mm}
\vspace{-2mm} \end{proof}
In the following proposition, we  assume that $2\kappa \leq e$ and $\ell$ is a fixed integer satisfying $\ell \equiv e \pmod 2$ and $\kappa+3 \leq \ell \leq e-\kappa+1 .$ Here, we  provide a method to  lift a self-orthogonal code  of type $\{\Lambda_{\gamma_{\ell}+2},\lambda_{\gamma_{\ell}+3}, \ldots, \lambda_{\gamma_{\ell}+\ell-1}\}$   and  length $n$ over $\mathcal{R}_{\ell-2,m}$ satisfying    property $(\mathfrak{P})$  to a self-orthogonal code  of type $\{\Lambda_{\gamma_{\ell}+1},\lambda_{\gamma_{\ell}+2}, \ldots, \lambda_{\gamma_{\ell}+\ell}\}$ and length $n$ over $\mathcal{R}_{\ell,m}$ satisfying   property $(\mathfrak{P}),$ where $\gamma_{\ell}=s-\mathrm{f}_{\ell}.$     We also determine the total number of distinct ways to preform this lifting.

\begin{proposition}\label{p3.7Keven} Let $2\kappa\leq e,$  and let  $\ell$ be a fixed integer  satisfying  $\ell \equiv e \pmod 2$ and $\kappa+3 \leq \ell \leq e-\kappa+1 .$ Define $\gamma_{\ell}=s-\mathrm{f}_{\ell}.$   Let $\mathcal{C}_{\ell-2}$ be a self-orthogonal code of type $\{\Lambda_{\gamma_{\ell}+2},\lambda_{\gamma_{\ell}+3}, \ldots, \lambda_{\gamma_{\ell}+\ell-1}\}$   and  length $n$ over $\mathcal{R}_{\ell-2,m}$ satisfying   property $(\mathfrak{P}).$  Let $\mathcal{D}^{(\gamma_{\ell}+1)}$ be a   $\Lambda_{\gamma_{\ell}+1}$-dimensional linear subcode of $Tor_1(\mathcal{C}_{\ell-2}).$ 
The following hold.
 \begin{enumerate}\item[(a)] There exists a  self-orthogonal code $\mathcal{C}_{\ell}$ of type $\{\Lambda_{\gamma_{\ell}+1},\lambda_{\gamma_{\ell}+2}, \ldots, \lambda_{\gamma_{\ell}+\ell}\}$ and length $n$ over $\mathcal{R}_{\ell,m}$ satisfying  property $(\mathfrak{P})$ and $Tor_1(\mathcal{C}_{\ell})=\mathcal{D}^{(\gamma_{\ell}+1)}$ and    $Tor_{i+1}(\mathcal{C}_{\ell})=Tor_{i}(\mathcal{C}_{\ell-2})$ for $1 \leq i \leq  \ell-2.$
\item[(b)] Moreover,  each such pair  $(\mathcal{C}_{\ell-2},\mathcal{D}^{(\gamma_{\ell}+1)})$ of codes  gives rise to precisely 	 \vspace{-2mm}\begin{equation*}
(2^m)^{\sum\limits_{i=\ell}^{\gamma_{\ell}+\ell}\lambda_i\Lambda_{i-\ell+1}+\sum\limits_{j=\ell+1}^{\gamma_{\ell}+\ell}\lambda_j\Lambda_{j-\ell}+Y_{\ell}} {\lambda_{\gamma_{\ell}+\ell}+n-\Lambda_{\gamma_{\ell}+\ell}-\Lambda_{\gamma_{\ell}+1} \brack \lambda_{\gamma_{\ell}+\ell}}_{2^m}
\vspace{-2mm}\end{equation*}  distinct  self-orthogonal codes $\mathcal{C}_{\ell}$ of type $\{\Lambda_{\gamma_{\ell}+1},\lambda_{\gamma_{\ell}+2}, \ldots, \lambda_{\gamma_{\ell}+\ell}\}$ and length $n$ over $\mathcal{R}_{\ell,m}$ satisfying  property $(\mathfrak{P})$ and $Tor_1(\mathcal{C}_{\ell})=\mathcal{D}^{(\gamma_{\ell}+1)}$ and  $Tor_{i+1}(\mathcal{C}_{\ell})=Tor_{i}(\mathcal{C}_{\ell-2})$ for $1 \leq i \leq  \ell-2,$ where \vspace{-1mm}\begin{equation*}\vspace{-1mm}
Y_{\ell}=(\Lambda_{\gamma_{\ell}+\ell-1}+\Lambda_{\gamma_{\ell}+1})(n-\Lambda_{\gamma_{\ell}+\ell}-\Lambda_{\gamma_{\ell}+1})+\Lambda_{\gamma_{\ell}+1}^2+\Lambda_{\gamma_{\ell}+1}-2\Lambda_{\gamma_{\ell}+1-\kappa_1}.    
\end{equation*}
 \vspace{-1mm}\end{enumerate} \end{proposition}
 \begin{proof}  Working as in Proposition 4.3 of Yadav and Sharma \cite{Galois}, the desired result follows immediately. 
 \vspace{-1mm}\end{proof}
In the following proposition, we  assume that $2\kappa > e$ and $\ell$ is a fixed integer satisfying $\ell \equiv e \pmod 2$ and $e-\kappa+1 \leq \ell \leq \kappa-\mathrm{f}_{2\kappa-e}+1 .$ Here, we  provide a method to  lift a self-orthogonal code  of type $\{\Lambda_{\gamma_{\ell}+2},\lambda_{\gamma_{\ell}+3}, \ldots, \lambda_{\gamma_{\ell}+\ell-1}\}$   and  length $n$ over $\mathcal{R}_{\ell-2,m}$ satisfying    property $(\mathfrak{P})$  to a self-orthogonal code  of type $\{\Lambda_{\gamma_{\ell}+1},\lambda_{\gamma_{\ell}+2}, \ldots, \lambda_{\gamma_{\ell}+\ell}\}$ and length $n$ over $\mathcal{R}_{\ell,m}$ satisfying   property $(\mathfrak{P}),$ where $\gamma_{\ell}=s-\mathrm{f}_{\ell}.$     We also count the total number of distinct ways in which this lifting can be carried out.

\begin{proposition}\label{p3.8Keven}  Let $2\kappa>e,$  and let $\ell$ be a fixed integer  satisfying $\ell \equiv e\pmod 2$ and $e-\kappa+1 \leq \ell \leq \kappa-\mathrm{f}_{2\kappa-e}+1 .$ Let us define $\gamma_{\ell}=s-\mathrm{f}_{\ell}.$   Let $\mathcal{C}_{\ell-2}$ be a self-orthogonal code of type $\{\Lambda_{\gamma_{\ell}+2},\lambda_{\gamma_{\ell}+3}, \ldots, \lambda_{\gamma_{\ell}+\ell-1}\}$   and  length $n$ over $\mathcal{R}_{\ell-2,m}$ satisfying   property $(\mathfrak{P}).$ Let $\mathcal{D}^{(\gamma_{\ell}+1)}$ be a  $\Lambda_{\gamma_{\ell}+1}$-dimensional linear subcode of $Tor_1(\mathcal{C}_{\ell-2}).$ 
The following hold.
\begin{enumerate}\item[(a)] There exists a  self-orthogonal code $\mathcal{C}_{\ell}$ of type $\{\Lambda_{\gamma_{\ell}+1},\lambda_{\gamma_{\ell}+2}, \ldots, \lambda_{\gamma_{\ell}+\ell}\}$ and length $n$ over $\mathcal{R}_{\ell,m}$ satisfying  property $(\mathfrak{P})$ and   $Tor_1(\mathcal{C}_{\ell})=\mathcal{D}^{(\gamma_{\ell}+1)}$ and   $Tor_{i+1}(\mathcal{C}_{\ell})=Tor_{i}(\mathcal{C}_{\ell-2})$ for $1 \leq i \leq  \ell-2.$
\item[(b)] Moreover,  each such pair $(\mathcal{C}_{\ell-2},\mathcal{D}^{(\gamma_{\ell}+1)})$ of codes gives rise to precisely  
\vspace{-2mm}\begin{equation*}
\displaystyle(2^m)^{\sum\limits_{i=\ell}^{\gamma_{\ell}+\ell}\lambda_i\Lambda_{i-\ell+1}+\sum\limits_{j=\ell+1}^{\gamma_{\ell}+\ell}\lambda_j\Lambda_{j-\ell}+Y_{\ell}} {\lambda_{\gamma_{\ell}+\ell}+n-\Lambda_{\gamma_{\ell}+\ell}-\Lambda_{\gamma_{\ell}+1} \brack \lambda_{\gamma_{\ell}+\ell}}_{2^m}
\vspace{-2mm}\end{equation*}  distinct  self-orthogonal codes $\mathcal{C}_{\ell}$ of type $\{\Lambda_{\gamma_{\ell}+1},\lambda_{\gamma_{\ell}+2}, \ldots, \lambda_{\gamma_{\ell}+\ell}\}$ and length $n$ over $\mathcal{R}_{\ell,m}$ satisfying  property $(\mathfrak{P})$ and $Tor_1(\mathcal{C}_{\ell})=\mathcal{D}^{(\gamma_{\ell}+1)}$ and  $Tor_{i+1}(\mathcal{C}_{\ell})=Tor_{i}(\mathcal{C}_{\ell-2})$ for $1 \leq i \leq  \ell-2,$ where \vspace{-1mm}\begin{equation*} \vspace{-1mm}Y_{\ell}=(\Lambda_{\gamma_{\ell}+\ell-1}+\Lambda_{\gamma_{\ell}+1})(n-\Lambda_{\gamma_{\ell}+\ell}-\Lambda_{\gamma_{\ell}+1})+\Lambda_{\gamma_{\ell}+1}^2+\Lambda_{\gamma_{\ell}+1}-\Lambda_{s-2\mathrm{f}_{\ell}+2}-\Lambda_{s-2\mathrm{f}_{\ell}+1}.\vspace{-2mm}\end{equation*}
 \vspace{-2mm}\end{enumerate} \vspace{-3mm}\end{proposition}
 \begin{proof}
   By applying Lemma 4.1 of Yadav and Sharma \cite{quasi} and working as in Proposition \ref{p3.4Keven},  the desired result follows immediately.  
 \vspace{-1mm}\end{proof} 
In the following proposition, we  assume that  $\ell$ is a fixed integer satisfying $\ell \equiv e \pmod 2$ and $e-\kappa+2 \leq \ell \leq e $ when $2\kappa \leq e$, whereas $\ell \equiv e \pmod 2$ and    $\kappa-\mathrm{f}_{2\kappa-e}+1 < \ell \leq e$ when $2\kappa >e.$  Here, we  provide a method to  lift a self-orthogonal code  of type $\{\Lambda_{\gamma_{\ell}+2},\lambda_{\gamma_{\ell}+3}, \ldots, \lambda_{\gamma_{\ell}+\ell-1}\}$   and  length $n$ over $\mathcal{R}_{\ell-2,m}$ satisfying    property $(\mathfrak{P})$  to a self-orthogonal code  of type $\{\Lambda_{\gamma_{\ell}+1},\lambda_{\gamma_{\ell}+2}, \ldots, \lambda_{\gamma_{\ell}+\ell}\}$ and length $n$ over $\mathcal{R}_{\ell,m}$ satisfying   property $(\mathfrak{P}),$ where $\gamma_{\ell}=s-\mathrm{f}_{\ell}.$     We also determine the total number of distinct ways to carry out this lifting.

\begin{proposition}\label{p3.9Keven} Suppose that $\ell$ is a fixed integer satisfying $\ell\equiv e\pmod 2$ and $e-\kappa+2 \leq \ell \leq e $ when $2\kappa \leq e$, while  $\ell\equiv e\pmod 2$ and  $\kappa-\mathrm{f}_{2\kappa-e}+1 < \ell \leq e$ when $2\kappa >e.$ Let us define 
$\gamma_{\ell}=s-\mathrm{f}_{\ell}.$
 Let $\mathcal{C}_{\ell-2}$ be a self-orthogonal code of type $\{\Lambda_{\gamma_{\ell}+2},\lambda_{\gamma_{\ell}+3}, \ldots, \lambda_{\gamma_{\ell}+\ell-1}\}$   and  length $n$ over $\mathcal{R}_{\ell-2,m}$ satisfying   property $(\mathfrak{P}).$ Let $\mathcal{D}^{(\gamma_{\ell}+1)}$ be a  $\Lambda_{\gamma_{\ell}+1}$-dimensional  linear subcode of  $Tor_1(\mathcal{C}_{\ell-2}).$ 
The following hold. \begin{enumerate}\item[(a)] There exists a  self-orthogonal code $\mathcal{C}_{\ell}$ of type $\{\Lambda_{\gamma_{\ell}+1},\lambda_{\gamma_{\ell}+2}, \ldots, \lambda_{\gamma_{\ell}+\ell}\}$ and length $n$ over $\mathcal{R}_{\ell,m}$ satisfying  property $(\mathfrak{P})$ and $Tor_1(\mathcal{C}_{\ell})=\mathcal{D}^{(\gamma_{\ell}+1)}$ and  $Tor_{i+1}(\mathcal{C}_{\ell})=Tor_{i}(\mathcal{C}_{\ell-2})$ for $1 \leq i \leq  \ell-2.$
\item[(b)] Moreover,  each such pair $(\mathcal{C}_{\ell-2},\mathcal{D}^{(\gamma_{\ell}+1)})$ of codes gives rise to precisely \vspace{-2mm}\begin{equation*}
\displaystyle(2^m)^{\sum\limits_{i=\ell}^{\gamma_{\ell}+\ell}\lambda_i\Lambda_{i-\ell+1}+\sum\limits_{j=\ell+1}^{\gamma_{\ell}+\ell}\lambda_j\Lambda_{j-\ell}+Y_{\ell}} {\lambda_{\gamma_{\ell}+\ell}+n-\Lambda_{\gamma_{\ell}+\ell}-\Lambda_{\gamma_{\ell}+1} \brack \lambda_{\gamma_{\ell}+\ell}}_{2^m}
\vspace{-2mm}\end{equation*}  distinct  self-orthogonal codes of type $\{\Lambda_{\gamma_{\ell}+1},\lambda_{\gamma_{\ell}+2}, \ldots, \lambda_{\gamma_{\ell}+\ell}\}$ and length $n$ over $\mathcal{R}_{\ell,m}$ satisfying  property $(\mathfrak{P})$ and  $Tor_1(\mathcal{C}_{\ell})=\mathcal{D}^{(\gamma_{\ell}+1)}$ and $Tor_{i+1}(\mathcal{C}_{\ell})=Tor_{i}(\mathcal{C}_{\ell-2})$ for $1 \leq i \leq  \ell-2,$ where \vspace{-1mm}\begin{equation*}
Y_{\ell}=(\Lambda_{\gamma_{\ell}+\ell-1}+\Lambda_{\gamma_{\ell}+1})(n-\Lambda_{\gamma_{\ell}+\ell}-\Lambda_{\gamma_{\ell}+1})+\Lambda_{\gamma_{\ell}+1}^2+\Lambda_{\gamma_{\ell}+1}.\vspace{-1mm}\end{equation*}
 \end{enumerate} \end{proposition}
 \begin{proof}
   Working as in Proposition \ref{p3.4Keven}, the desired result follows. 
 \end{proof}


In the following lemma, we consider the case  $e \geq 3$ and $2\kappa \leq e,$ and show that for every self-orthogonal code  of type $\{\lambda_1,\lambda_2,\ldots,\lambda_e\}$ and length $n$  over $\mathcal{R}_{e,m},$  there exists a chain $\mathcal{D}^{(1)}\subseteq \mathcal{D}^{(2)} \subseteq \cdots \subseteq \mathcal{D}^{(s+\theta_e)} $ of self-orthogonal codes of length $n$ over $\mathcal{T}_m,$  with  $\dim \mathcal{D}^{(i)}=\Lambda_i$ for $1 \leq i \leq s+\theta_e,$ satisfying  certain additional conditions.


 \begin{lemma}\label{l3.4} 
    Let $e\geq 3$ be an integer satisfying $2\kappa \leq e.$ For every self-orthogonal code  of type $\{\lambda_1,\lambda_2,\ldots,\lambda_e\}$ and length $n$  over $\mathcal{R}_{e,m},$ there exists a chain $\mathcal{D}^{(1)}\subseteq \mathcal{D}^{(2)} \subseteq \cdots \subseteq \mathcal{D}^{(s+\theta_e)} $ of self-orthogonal codes of length $n$ over $\mathcal{T}_m$ satisfying  the following  conditions:
    \begin{enumerate}\vspace{-1mm}\item[\textbf{A1)}]  $\dim \mathcal{D}^{(i)}=\Lambda_i$ for $1 \leq i \leq s+\theta_e.$ 
\vspace{-1mm}\item[\textbf{A2)}]  $\textbf{1}\notin \mathcal{D}^{(s-\mathrm{f}_\ell+1-\kappa_1)},$ provided that $n\equiv 2,6\pmod 8$  and $\eta_{\ell-\theta_e-1}\neq 0,$ where $2\leq \ell \leq \kappa_1+\theta_e.$  
\vspace{-1mm}\item[\textbf{A3)}] $\textbf{1}\notin \mathcal{D}^{(s-\mathrm{f}_{\kappa_1+1}+1-\kappa_1)},$ provided that $n\equiv 2,6 \pmod 8,$  
$\eta_{\kappa_1}\neq (\eta_0)^{\frac{3}{2}}$ and either $\kappa \geq 4$ is a singly even integer or $\kappa=2$ and  $e$ is an even integer. 
\vspace{-1mm}\item[\textbf{A4)}] $\textbf{1}\notin \mathcal{D}^{(s-\mathrm{f}_\ell+1-\kappa_1)}$  for all $\kappa_1+2+\theta_e\leq \ell \leq \kappa+2-\theta_e ,$ provided that $n\equiv 2,6 \pmod 8.$ 
\vspace{-1mm}\item[\textbf{A5)}] $\textbf{1} \notin \mathcal{D}^{(s-\kappa+\theta_e)},$ provided that   either $n\equiv 4\pmod 8$ and $m$ is odd, or $n\equiv 2,6\pmod 8,$ $\kappa=2$ and $e$ is an odd integer. \end{enumerate}
 \end{lemma}
 \begin{proof} 
To prove the result, let $\mathcal{C}_e$ be a self-orthogonal code of type $\{\lambda_1,\lambda_2,\ldots,\lambda_e\}$ and length $n$  over $\mathcal{R}_{e,m}$ with a generator matrix $\mathcal{G}_e$ as defined in \eqref{e3.0} for the case $\ell=e.$

For $1 \leq i \leq s+\theta_e,$  we note, by Lemma \ref{l1.2}, that the $i$-th torsion code $Tor_i(\mathcal{C}_e)$ is a self-orthogonal code of length $n$ and dimension $\Lambda_i$ over $\mathcal{T}_m.$ It is easy to see that $Tor_1(\mathcal{C}_e)\subseteq Tor_2(\mathcal{C}_e) \subseteq \cdots \subseteq Tor_{s+\theta_e}(\mathcal{C}_e).$ 
Further,  for each integer $\ell$ satisfying  $2 \leq \ell \leq e-2$ and $\ell \equiv e \pmod 2,$ corresponding to the code $\mathcal{C}_e,$ we  consider a linear code $\mathcal{C}_{\ell}$ of type $\{\Lambda_{s-\mathrm{f}_\ell+1},\lambda_{s-\mathrm{f}_\ell+2}, \ldots, \lambda_{s-\mathrm{f}_\ell+\ell}\}$ and length $n$ over $\mathcal{R}_{\ell,m}$  with a generator matrix $\mathcal{G}_{\ell}$ as defined in \eqref{e3.0}.   We will now distinguish the following two cases: (I) $e$ is even, and (II) $e$ is odd.  

 \begin{enumerate}
     \item[(I)] Let $e$ be even.  For $\ell=2,$ we see, by Lemma \ref{l3.2},  that $\mathcal{C}_2$ is a self-orthogonal code of type $\{\Lambda_s,\lambda_{s+1}\}$ and length $n$ over $\mathcal{R}_{2,m}$ satisfying  property $(\mathfrak{P}).$  Now, by closely looking at the proof of Lemma \ref{R=2}, we see that $\textbf{1}\notin Tor_{s-\kappa_1}(\mathcal{C}_e),$ provided that   $n\equiv 2,6\pmod 8$ and  either $\kappa=2$ and $\eta_1\neq (\eta_0)^{\frac{3}{2}},$ or $\kappa \geq 4$   and $\eta_1\neq 0.$ 

Next, for $4\leq \ell  \leq \kappa+2$ and $\ell \equiv e\pmod 2,$ we note,   by Lemma \ref{l3.2} again,  that $\mathcal{C}_{\ell-2}$ is a self-orthogonal code of type $\{\Lambda_{s-\mathrm{f}_{\ell}+2},\lambda_{s-\mathrm{f}_{\ell}+3}, \ldots, \lambda_{s-\mathrm{f}_{\ell}+\ell-1}\}$ and length $n$ over $\mathcal{R}_{\ell-2,m}$ and $\mathcal{C}_{\ell}$ is a self-orthogonal code of type $\{\Lambda_{s-\mathrm{f}_\ell+1},\lambda_{s-\mathrm{f}_\ell+2}, \ldots, \lambda_{s-\mathrm{f}_\ell+\ell}\}$ and length $n$ over $\mathcal{R}_{\ell,m},$ both satisfying   property $(\mathfrak{P}).$ For a fixed choice of the  integer $\ell$ satisfying $4\leq \ell \leq \kappa,$  working as in Proposition \ref{p3.4Keven}, we see that  \begin{itemize}
    \item  $\textbf{1}\notin Tor_{s-\mathrm{f}_\ell+1-\kappa_1}(\mathcal{C}_e)$ when $n\equiv 2,6\pmod 8$  and $\eta_{\ell-1}\neq 0$ if $4\leq \ell \leq \kappa_1.$ \item $\textbf{1}\notin Tor_{s-\mathrm{f}_{\kappa_1+1}+1-\kappa_1}(\mathcal{C}_e)$ when $n\equiv 2,6 \pmod 8,$  $\kappa_1+1$ is even and 
$\eta_{\kappa_1}\neq (\eta_0)^{\frac{3}{2}}.$  \item   $\textbf{1}\notin Tor_{s-\mathrm{f}_\ell+1-\kappa_1}(\mathcal{C}_e)$   when $n\equiv 2,6 \pmod 8$ and $\kappa_1+2\leq \ell \leq \kappa .$
\end{itemize}
Finally,  for $\ell=\kappa+2,$ working as in Proposition 3.4 of Yadav and Sharma \cite{YSub2}, we see that $\textbf{1}\notin Tor_{s-\kappa}(\mathcal{C}_e)$ when either $n\equiv 2,6\pmod8$ or $n\equiv 4\pmod 8$ and $m$ is odd. Now, on taking $\mathcal{D}^{(i)}=Tor_i(\mathcal{C}_e)$ for $1 \leq i \leq s,$ the desired result follows.
\item[(II)] For odd $e,$ the desired result follows by working as in case (I).
\vspace{-4mm}\end{enumerate}
\vspace{-4mm}\end{proof}

In the following lemma, we consider  the case $e \geq 3$ and $2\kappa > e,$ and show that for every self-orthogonal code  of type $\{\lambda_1,\lambda_2,\ldots,\lambda_e\}$ and length $n$  over $\mathcal{R}_{e,m},$  there exists a chain $\mathcal{D}^{(1)}\subseteq \mathcal{D}^{(2)} \subseteq \cdots \subseteq \mathcal{D}^{(s+\theta_e)} $ of self-orthogonal codes of length $n$ over $\mathcal{T}_m,$  with  $\dim \mathcal{D}^{(i)}=\Lambda_i$ for $1 \leq i \leq s+\theta_e,$   satisfying certain additional
conditions.
 \begin{lemma}\label{l3.4a} 
    Let $e\geq 3$ be an integer satisfying $2\kappa > e.$ For every self-orthogonal code  of type $\{\lambda_1,\lambda_2,\ldots,\lambda_e\}$ and length $n$  over $\mathcal{R}_{e,m},$ there exists a chain $\mathcal{D}^{(1)}\subseteq \mathcal{D}^{(2)} \subseteq \cdots \subseteq \mathcal{D}^{(s+\theta_e)} $ of self-orthogonal codes of length $n$ over $\mathcal{T}_m$ satisfying  the following  conditions:
    \begin{enumerate}\vspace{-1mm}\item[\textbf{B1)}]  $\dim \mathcal{D}^{(i)}=\Lambda_i$ for $1 \leq i \leq s+\theta_e.$
    \vspace{-1mm}\item[\textbf{B2)}]  $\textbf{1}\notin \mathcal{D}^{(s-\mathrm{f}_\ell+1-\kappa_1)},$ provided that $n\equiv 2,6\pmod 8$  and $\eta_{\ell-\theta_e-1}\neq 0,$ where  $2\leq \ell \leq \min\{\kappa_1+\theta_e,e-\kappa\}.$ 
   \vspace{-1mm} \item[\textbf{B3)}] $\textbf{1}\notin \mathcal{D}^{(s-\mathrm{f}_{\kappa_1+1}+1-\kappa_1)},$ provided that $n\equiv 2,6 \pmod 8,$  $\kappa$ is a singly even integer, $e \geq \frac{3\kappa}{2}+1+\theta_e$  and 
$\eta_{\kappa_1}\neq (\eta_0)^{\frac{3}{2}}.$ 
\vspace{-1mm}\item[\textbf{B4)}] $\textbf{1}\notin \mathcal{D}^{(s-\mathrm{f}_\ell+1-\kappa_1)}$ for all $\kappa_1+2+\theta_e\leq \ell \leq e-\kappa,$  provided that $n\equiv 2,6 \pmod 8.$ 
\end{enumerate}
\end{lemma}
\begin{proof}
    Working as in Lemma \ref{l3.4}, 
    the desired result follows. 
\end{proof}
  Given a chain $\mathcal{D}^{(1)}\subseteq \mathcal{D}^{(2)} \subseteq \cdots \subseteq \mathcal{D}^{(s+\theta_e)} $  of self-orthogonal codes of length $n$ over $\mathcal{T}_m$ satisfying  conditions \textbf{A1)} - \textbf{A5)} as stated in Lemma \ref{l3.4} when $2\kappa \leq e$ and   conditions \textbf{B1)} - \textbf{B4)} as stated in Lemma \ref{l3.4a} when $2\kappa > e,$ one can easily observe  that  the proofs of Propositions \ref{pe=k=2} - \ref{p3.9Keven} provide a recursive method to construct a 
self-orthogonal code $\mathcal{C}_e$ of type $\{\lambda_1,\lambda_2,\ldots,\lambda_e\}$ and length $n$  over $\mathcal{R}_{e,m}$ satisfying $Tor_i(\mathcal{C}_e)=\mathcal{D}^{(i)}$ for $1 \leq i \leq s+\theta_e.$ We will now outline this recursive construction method by  distinguishing the following two cases: (\textbf{X}) $2 \kappa \leq e,$ and (\textbf{Y}) $2 \kappa > e.$ 
\begin{enumerate}
    \item[(\textbf{X})] Method to construct a self-orthogonal code of type $\{\lambda_1,\lambda_2,\ldots,\lambda_e\}$ and length $n$ over $\mathcal{R}_{e,m}$  when $2\kappa \leq e$ 
\begin{itemize}
\item[\textbf{Input}:] Start with a chain  $\mathcal{D}^{(1)}\subseteq \mathcal{D}^{(2)}\subseteq \cdots\subseteq \mathcal{D}^{(s+\theta_e)}$ of  self-orthogonal codes of length $n$ over $\mathcal{T}_m$ satisfying conditions \textbf{A1)} - \textbf{A5)} as stated in Lemma \ref{l3.4}. 
\item[\textbf{Step 1}:] If $e=\kappa=2,$  apply Proposition \ref{pe=k=2} to construct a self-orthogonal code $\mathcal{C}_{2}$ of type $\{\lambda_{1},\lambda_{2}\}$ and length $n$ over $\mathcal{R}_{2,m}$ satisfying $Tor_1(\mathcal{C}_2)=\mathcal{D}^{(1)}.$  On the other hand, if $e\geq 4$ is even,  apply  Proposition \ref{p3.1} to obtain  a self-orthogonal code $\mathcal{C}_2$ of type $\{\Lambda_s, \lambda_{s+1}\}$ and length $n$ over $\mathcal{R}_{2,m}$  satisfying  property $(\mathfrak{P})$  and $Tor_1(\mathcal{C}_2)=\mathcal{D}^{(s)}.$ 

On the other hand, for odd $e$,   apply Proposition \ref{p3.3b} when $\kappa=2$ and   Proposition \ref{p3.3c} when $\kappa\geq 4$ to construct a self-orthogonal code $\mathcal{C}_{3}$ of type $\{\Lambda_{s},\lambda_{s+1},\lambda_{s+2}\}$ and length $n$ over $\mathcal{R}_{3,m}$  satisfying  property $(\mathfrak{P}),$ and   $Tor_1(\mathcal{C}_3)=\mathcal{D}^{(s)}$   and  $Tor_2(\mathcal{C}_3)=\mathcal{D}^{(s+1)}.$
\item[\textbf{Step 2}:] For $4 \leq \ell \leq \kappa$ and $\ell \equiv e \pmod 2,$ apply  Proposition \ref{p3.4Keven}  to construct a self-orthogonal code $\mathcal{C}_{\ell}$ of type $\{\Lambda_{\gamma_{\ell}+1},\lambda_{\gamma_{\ell}+2}, \ldots, \lambda_{\gamma_{\ell}+\ell}\}$ and length $n$ over $\mathcal{R}_{\ell,m}$ satisfying   property $(\mathfrak{P})$ and     $Tor_{1}(\mathcal{C}_{\ell})=\mathcal{D}^{(\gamma_{\ell}+1)}$ and $Tor_{j+1}(\mathcal{C}_{\ell})=Tor_{j}(\mathcal{C}_{\ell-2})$ for $1 \leq j \leq  \ell-2$ from a self-orthogonal code  $\mathcal{C}_{\ell-2}$ of type $\{\Lambda_{\gamma_{\ell}+2},\lambda_{\gamma_{\ell}+3}, \ldots, \lambda_{\gamma_{\ell}+\ell-1}\}$   and  length $n$ over $\mathcal{R}_{\ell-2,m}$ satisfying   property $(\mathfrak{P}).$  
\item[\textbf{Step 3}:] When $e$ is odd, apply  Proposition \ref{p3.5Keven}  to construct a self-orthogonal code $\mathcal{C}_{\kappa+1}$ of type $\{\Lambda_{s-\kappa_1+1},\lambda_{s-\kappa_1+2}, \\ \ldots, \lambda_{s+\kappa_1+1}\}$ and length $n$ over $\mathcal{R}_{\kappa+1,m}$ satisfying  property $(\mathfrak{P})$ and $Tor_1(\mathcal{C}_{\kappa+1})=\mathcal{D}^{(s-\kappa_1+1)}$ and    $Tor_{j+1}(\mathcal{C}_{\kappa+1})=Tor_{j}(\mathcal{C}_{\kappa-1})$ for $1 \leq j \leq  \kappa-1$ from a self-orthogonal code $\mathcal{C}_{\kappa-1}$ of type $\{\Lambda_{s-\kappa_1+2},\lambda_{s-\kappa_1+3},\\ \ldots, \lambda_{s+\kappa_1}\}$   and  length $n$ over $\mathcal{R}_{\kappa-1,m}$ satisfying  property $(\mathfrak{P}).$ 

On the other hand, when $e$ is even, apply  Proposition 
  \ref{p3.6Keven} to construct a self-orthogonal code $\mathcal{C}_{\kappa+2}$ of type $\{\Lambda_{s-\kappa_1},\lambda_{s-\kappa_1+1}, \ldots, \lambda_{s+\kappa_1+1}\}$ and length $n$ over $\mathcal{R}_{\kappa+2,m}$ satisfying  property $(\mathfrak{P})$ and  $Tor_{1}(\mathcal{C}_{\kappa+2})=\mathcal{D}^{(s-\kappa_1)}$ and     $Tor_{j+1}(\mathcal{C}_{\kappa+2})=Tor_{j}(\mathcal{C}_{\kappa})$ for $1 \leq j \leq  \kappa$ from a self-orthogonal code  $\mathcal{C}_{\kappa}$ of type $\{\Lambda_{s-\kappa_1+1},\\ \lambda_{s-\kappa_1+2}, \ldots, \lambda_{s+\kappa_1}\}$   and  length $n$ over $\mathcal{R}_{\kappa,m}$ satisfying   property $(\mathfrak{P}).$  
\item[\textbf{Step 4}:] For $\kappa+3 \leq \ell \leq e-\kappa+1 $ and $\ell \equiv e\pmod2,$ apply
  Proposition \ref{p3.7Keven}  to construct a  self-orthogonal code $\mathcal{C}_{\ell}$ of type $\{\Lambda_{\gamma_{\ell}+1},\lambda_{\gamma_{\ell}+2}, \ldots, \lambda_{\gamma_{\ell}+\ell}\}$ and length $n$ over $\mathcal{R}_{\ell,m}$ satisfying  property $(\mathfrak{P})$ and $Tor_1(\mathcal{C}_{\ell})=\mathcal{D}^{(\gamma_{\ell}+1)}$ and    $Tor_{i+1}(\mathcal{C}_{\ell})=Tor_{i}(\mathcal{C}_{\ell-2})$ for $1 \leq i \leq  \ell-2$ from a self-orthogonal code $\mathcal{C}_{\ell-2}$ of type $\{\Lambda_{\gamma_{\ell}+2},\lambda_{\gamma_{\ell}+3}, \ldots, \lambda_{\gamma_{\ell}+\ell-1}\}$   and  length $n$ over $\mathcal{R}_{\ell-2,m}$ satisfying   property $(\mathfrak{P}).$ 
\item[\textbf{Step 5}:]  For $e-\kappa+2 \leq \ell \leq e $ and  $\ell \equiv e\pmod2, $ apply Proposition \ref{p3.9Keven}  to construct a   self-orthogonal code $\mathcal{C}_{\ell}$ of type $\{\Lambda_{\gamma_{\ell}+1},\lambda_{\gamma_{\ell}+2}, \ldots, \lambda_{\gamma_{\ell}+\ell}\}$ and length $n$ over $\mathcal{R}_{\ell,m}$ satisfying  property $(\mathfrak{P})$ and  $Tor_1(\mathcal{C}_{\ell})=\mathcal{D}^{(\gamma_{\ell}+1)}$ and  $Tor_{i+1}(\mathcal{C}_{\ell})=Tor_{i}(\mathcal{C}_{\ell-2})$ for $1 \leq i \leq  \ell-2$ from a self-orthogonal code $\mathcal{C}_{\ell-2}$  of type $\{\Lambda_{\gamma_{\ell}+2},\lambda_{\gamma_{\ell}+3}, \ldots, \lambda_{\gamma_{\ell}+\ell-1}\}$   and  length $n$ over $\mathcal{R}_{\ell-2,m}$ satisfying   property $(\mathfrak{P}).$ 
\item[\textbf{Output}:] A 
self-orthogonal code $\mathcal{C}_e$ of type $\{\lambda_1,\lambda_2,\ldots,\lambda_e\}$ and length $n$  over $\mathcal{R}_{e,m}$ satisfying $Tor_i(\mathcal{C}_e)=\mathcal{D}^{(i)}$ for $1 \leq i \leq s+\theta_e.$
\end{itemize}
\item[(\textbf{Y})]  Method to construct a self-orthogonal code of type $\{\lambda_1,\lambda_2,\ldots, \lambda_e\}$ and length $n$ over $\mathcal{R}_{e,m}$  when $2\kappa>e.$ 
\begin{itemize}
\item[\textbf{Input}:] Start with a chain  $\mathcal{D}^{(1)}\subseteq \mathcal{D}^{(2)}\subseteq \cdots\subseteq \mathcal{D}^{(s+\theta_e)}$ of self-orthogonal codes of length $n$ over $\mathcal{T}_m$ satisfying conditions \textbf{B1)} - \textbf{B4)} as stated in Lemma \ref{l3.4a}.
\item[\textbf{Step 1}:] If $e$ is even,  apply Proposition \ref{p3.1} to construct a self-orthogonal code $\mathcal{C}_2$ of type $\{\Lambda_s, \lambda_{s+1}\}$ and length $n$ over $\mathcal{R}_{2,m}$  satisfying  property $(\mathfrak{P})$  and $Tor_1(\mathcal{C}_2)=\mathcal{D}^{(s)}.$ 

If $e=3$  and $\kappa=2$, apply Proposition \ref{p3.3be=3} to construct a self-orthogonal code $\mathcal{C}_{3}$ of type $\{\lambda_{1},\lambda_{2},\lambda_{3}\}$ and length $n$ over $\mathcal{R}_{3,m}$ satisfying $Tor_1(\mathcal{C}_3)=\mathcal{D}^{(1)}$   and  $Tor_2(\mathcal{C}_3)=\mathcal{D}^{(2)}.$ On the other hand, if $e$ is odd and $\kappa\geq 4,$ apply Proposition \ref{p3.3c} to construct a self-orthogonal code $\mathcal{C}_{3}$ of type $\{\Lambda_{s},\lambda_{s+1},\lambda_{s+2}\}$ and length $n$ over $\mathcal{R}_{3,m}$  satisfying  property $(\mathfrak{P})$ and  $Tor_1(\mathcal{C}_3)=\mathcal{D}^{(s)}$   and  $Tor_2(\mathcal{C}_3)=\mathcal{D}^{(s+1)}.$
\item[\textbf{Step 2}:] For $4 \leq \ell \leq e-\kappa $ and $\ell \equiv e \pmod 2,$ apply Proposition \ref{p3.4Keven}  to construct a self-orthogonal code $\mathcal{C}_{\ell}$ of type $\{\Lambda_{\gamma_{\ell}+1},\lambda_{\gamma_{\ell}+2}, \ldots, \lambda_{\gamma_{\ell}+\ell}\}$ and length $n$ over $\mathcal{R}_{\ell,m}$ satisfying   property $(\mathfrak{P})$ and     $Tor_{1}(\mathcal{C}_{\ell})=\mathcal{D}^{(\gamma_{\ell}+1)}$ and $Tor_{j+1}(\mathcal{C}_{\ell})=Tor_{j}(\mathcal{C}_{\ell-2})$ for $1 \leq j \leq  \ell-2$ from a self-orthogonal code  $\mathcal{C}_{\ell-2}$ of type $\{\Lambda_{\gamma_{\ell}+2},\lambda_{\gamma_{\ell}+3}, \ldots, \lambda_{\gamma_{\ell}+\ell-1}\}$   and  length $n$ over $\mathcal{R}_{\ell-2,m}$ satisfying   property $(\mathfrak{P}).$  
\item[\textbf{Step 3}:]  
For $e-\kappa+1 \leq \ell \leq \kappa-\mathrm{f}_{2\kappa-e}+1 $ and $\ell \equiv e\pmod2,$ apply  
  Proposition \ref{p3.8Keven}  to construct a  self-orthogonal code $\mathcal{C}_{\ell}$ of type $\{\Lambda_{\gamma_{\ell}+1},\lambda_{\gamma_{\ell}+2}, \ldots, \lambda_{\gamma_{\ell}+\ell}\}$ and length $n$ over $\mathcal{R}_{\ell,m}$ satisfying  property $(\mathfrak{P})$ and $Tor_1(\mathcal{C}_{\ell})=\mathcal{D}^{(\gamma_{\ell}+1)}$ and    $Tor_{i+1}(\mathcal{C}_{\ell})=Tor_{i}(\mathcal{C}_{\ell-2})$ for $1 \leq i \leq  \ell-2$ from a self-orthogonal code $\mathcal{C}_{\ell-2}$ of type $\{\Lambda_{\gamma_{\ell}+2},\lambda_{\gamma_{\ell}+3}, \ldots, \lambda_{\gamma_{\ell}+\ell-1}\}$   and  length $n$ over $\mathcal{R}_{\ell-2,m}$ satisfying   property $(\mathfrak{P}).$ 
\item[\textbf{Step 4}:]   For $\kappa-\mathrm{f}_{2\kappa-e}+1 < \ell \leq e$ and  $\ell \equiv e\pmod2, $  apply Proposition \ref{p3.9Keven}  to construct a   self-orthogonal code $\mathcal{C}_{\ell}$ of type $\{\Lambda_{\gamma_{\ell}+1},\lambda_{\gamma_{\ell}+2}, \ldots, \lambda_{\gamma_{\ell}+\ell}\}$ and length $n$ over $\mathcal{R}_{\ell,m}$ satisfying  property $(\mathfrak{P})$ and  $Tor_1(\mathcal{C}_{\ell})=\mathcal{D}^{(\gamma_{\ell}+1)}$ and  $Tor_{i+1}(\mathcal{C}_{\ell})=Tor_{i}(\mathcal{C}_{\ell-2})$ for $1 \leq i \leq  \ell-2$ from a self-orthogonal code $\mathcal{C}_{\ell-2}$  of type $\{\Lambda_{\gamma_{\ell}+2},\lambda_{\gamma_{\ell}+3}, \ldots, \lambda_{\gamma_{\ell}+\ell-1}\}$   and  length $n$ over $\mathcal{R}_{\ell-2,m}$ satisfying   property $(\mathfrak{P}).$ 
\item[\textbf{Output}:] A 
self-orthogonal code $\mathcal{C}_e$ of type $\{\lambda_1,\lambda_2,\ldots,\lambda_e\}$ and length $n$  over $\mathcal{R}_{e,m}$ satisfying $Tor_i(\mathcal{C}_e)=\mathcal{D}^{(i)}$ for $1 \leq i \leq s+\theta_e.$
\end{itemize}
\end{enumerate}

In particular, when $\lambda_j=\lambda_{e-j+2}$ for $1 \leq j \leq e,$
 Lemma \ref{l2.2} implies that the above construction methods yield a self-dual code $\mathcal{C}_e$
 of type $\{\lambda_1,\lambda_2,\ldots,\lambda_{e}\}$ and length $n$ over $\mathcal{R}_{e,m}$ satisfying $Tor_i(\mathcal{C}_e)=\mathcal{D}^{(i)}$ for $1 \leq i \leq s+\theta_e.$   Moreover, when $\mathfrak{s}=1,$
 the recursive construction method described above coincides with the method presented in  \cite{quasi}.

The following example illustrates that the self-orthogonal code $\mathcal{D}_0,$ 
 considered in Example  \ref{EX1}, can be lifted to a self-orthogonal code over $\mathcal{R}_{5,2}$ by applying the recursive construction method (\textbf{X}).

\begin{example}
Let  $\mathcal{R}_{5,2}$  be the finite commutative chain ring
as defined in Example \ref{EX1}.   
 Let $n=3,$ $\lambda_1=1$  and $\lambda_2=\lambda_3=\lambda_4=\lambda_5=0.$ Let $\mathcal{D}_0$ be a self-orthogonal  code of length $3$ and dimension $1$ over $\mathcal{T}_2$ with a generator matrix $\begin{bmatrix}1& 1&0\end{bmatrix},$ as considered in Example \ref{EX1}. We recall, from Example \ref{EX1},  that $\pi_{1}(\textbf{v}\cdot \textbf{v})=0$ for $\textbf{v}\in \mathcal{D}_0,$ where each $\textbf{v}\cdot \textbf{v}$ is viewed as an element of $\mathcal{R}_{5,2}.$ 
 We now demonstrate that  the code $\mathcal{D}_0$ can be lifted to a self-orthogonal code of type $\{1,0,0,0,0\}$ and length $3$ over $\mathcal{R}_{5,2},$  using the recursive construction method (\textbf{X}).
To proceed, consider a linear code $\mathcal{C}_3$ of type $\{1,0,0\} $ and length $3$ over $\mathcal{R}_{3,2}$ with a generator matrix $$G_3=\begin{bmatrix}
    1&1&0
\end{bmatrix}+u\begin{bmatrix}
    0&a_1&b_1
\end{bmatrix}+u^2\begin{bmatrix}
    0&a_2&b_2
\end{bmatrix},  $$  where $a_1,b_1,a_2,b_2\in \mathcal{T}_2.$ By the recursive construction method (\textbf{X}),  the code $\mathcal{D}_0$ can be lifted to a self-orthogonal code of type $\{1,0,0,0,0\}$ over $\mathcal{R}_{5,2}$ if and only if $\mathcal{C}_3$ is a   self-orthogonal code   over $\mathcal{R}_{3,2}$ satisfying  property $(\mathfrak{P}),$  which holds if and only if  $1+a_1^2+b_1^2+ua_1+u^2(1+a_2+a_2^2+b_2^2)\equiv 0\pmod{u^3}.$ This holds if and only if  $a_1=0,$ $b_1=1$   and $a_2+a_2^2+b_2^2\equiv 0\pmod u.$ It is easy to see that $a_2,b_2\in \mathcal{T}_2$ satisfying $a_2+a_2^2+b_2^2\equiv 0\pmod u$ have precisely $4$ distinct choices.    Therefore,  for $a_1=0,$  $b_1=1$ and $a_2,b_2\in \mathcal{T}_2$ satisfying $a_2+a_2^2+b_2^2\equiv 0\pmod u,$ we see that  the code $\mathcal{C}_3$  with a generator matrix $G_3$ is a self-orthogonal code satisfying  property $(\mathfrak{P}).$ We also note that $G_3G_3^t \equiv 0\pmod{u^5}.$ 

Now, corresponding to each such choice of 
the code $\mathcal{C}_3$,  let us consider a linear code $\mathcal{C}_5$ of type $\{1,0,0,0,0\} $ and length $3$ over $\mathcal{R}_{5,2}$ with a generator matrix $$G_5=G_3+u^3\begin{bmatrix}
0&a_3&b_3\end{bmatrix}+u^4\begin{bmatrix}
    0&a_4&b_4
\end{bmatrix},  $$  where $a_3,b_3,a_4,b_4\in \mathcal{T}_2.$ Note that for all choices of $a_3,b_3,a_4,b_4\in \mathcal{T}_2,$ the code $\mathcal{C}_5$ is  self-orthogonal.  This shows that the code $\mathcal{D}_0$ can be lifted to a self-orthogonal code of type $\{1,0,0,0,0\}$ and length $3$ over $\mathcal{R}_{5,2}$ by applying the  recursive construction  method (\textbf{X}).
\end{example}

 \section{Enumeration formulae for self-orthogonal and self-dual codes of
length $n$ over $\mathcal{R}_{e,m}$}\label{counting}

In this section, we will  derive explicit enumeration formulae for self-orthogonal and self-dual codes of an arbitrary length over $\mathcal{R}_{e,m}.$ Throughout this section, we assume that 
 $e \geq 3$ and  $n\geq 1$ are integers and that $\lambda_1,\lambda_2,\ldots,\lambda_{e+1}$ are non-negative integers satisfying $n=\lambda_1+\lambda_2+\cdots+\lambda_{e+1}.$
 We  denote by $\mathscr{M}_e(n;\lambda_1,\lambda_2,\ldots,\lambda_e)$ and $\mathscr{B}_e(n;\lambda_1,\lambda_2,\ldots,\lambda_e),$  the number of distinct  self-orthogonal and self-dual codes of type $\{\lambda_1,\lambda_2,\ldots,\lambda_e\}$ and length $n$ over $\mathcal{R}_{e,m},$  respectively.  By Remark \ref{R2.1}, we have   $\mathscr{M}_e(n;\lambda_1,\lambda_2,\ldots,\lambda_e)=0$ if  $2\lambda_1+2\lambda_2+\cdots+2\lambda_{e-i+1}+\lambda_{e-i+2}+\cdots+\lambda_i > n$  for some integer $i$ satisfying $s+1\leq i\leq e.$ On the other hand,  by Lemma \ref{l1.2}, we have  $\mathscr{B}_e(n;\lambda_1,\lambda_2,\ldots,\lambda_e)=0$ if $\lambda_j\neq \lambda_{e-j+2}$ for some integer $j$ satisfying $1 \leq j \leq e.$
\begin{remark}\label{R4.1}
    Let  $\mathscr{M}_e(n)$ and $\mathscr{B}_{e}(n)$ denote the total number of distinct self-orthogonal and self-dual codes of length $n$ over $\mathcal{R}_{e,m},$ respectively. By Lemma \ref{l2.2} and Remark \ref{R2.1}, we have  \vspace{-1mm}\begin{equation*}\label{summation}\vspace{-1mm}\mathscr{M}_e(n)= \mathlarger{\Sigma}_{1}~\mathscr{M}_e(n;\lambda_1,\lambda_2,\ldots,\lambda_e) \text{ ~and~ } \mathscr{B}_e(n)= \mathlarger{\Sigma}_{2}~ \mathscr{B}_e(n;\lambda_1,\lambda_2,\ldots,\lambda_e),\end{equation*}
where the summation $\mathlarger{\Sigma}_{1}$ runs over all $e$-tuples $(\lambda_1, \lambda_2, \ldots,  \lambda_{e})$ of non-negative integers satisfying $ \lambda_1+\lambda_2+\cdots+\lambda_{e}\leq n$ and  $2\lambda_1+2\lambda_2+\cdots+2\lambda_{e-j+1}+\lambda_{e-j+2}+\cdots+\lambda_j \leq n$ for $s+1\leq j\leq e,$ whereas the summation $\mathlarger{\Sigma}_{2}$ runs over all $e$-tuples $(\lambda_1, \lambda_2, \ldots,  \lambda_{e})$ of non-negative integers  satisfying $\lambda_j=\lambda_{e-j+2}$ for $2 \leq j \leq e$ and $ 2(\lambda_1+\lambda_2+\cdots+\lambda_{s})+(1+\theta_e)\lambda_{s+1}=n.$

\end{remark} 
From the above remark, it follows that to determine the numbers  $\mathscr{M}_e(n)$ and $\mathscr{B}_{e}(n),$ it suffices to obtain   enumeration formulae for the numbers $\mathscr{M}_e(n;\lambda_1,\lambda_2,\ldots,\lambda_e)$ and $\mathscr{B}_e(n;\lambda_1,\lambda_2,\ldots,\lambda_e).$ Furthermore, in view of Remark \ref{R2.1},  we assume,  throughout this section,  that the non-negative integers    $\lambda_1, \lambda_2, \ldots,  \lambda_{e+1}$  satisfy $n=\lambda_1+\lambda_2+\cdots+\lambda_{e+1}$ and  $2\lambda_1+2\lambda_2+\cdots+2\lambda_{e-j+1}+\lambda_{e-j+2}+\cdots+\lambda_j \leq n$ for $s+1\leq j\leq e.$ We also  define  $\Lambda_0=0$ and $\Lambda_i=\lambda_1+\lambda_2+\cdots+\lambda_{i}$ for $1 \leq i \leq e+1.$ Now, to derive   enumeration formulae for the numbers $\mathscr{M}_e(n;\lambda_1,\lambda_2,\ldots,\lambda_e)$ and $\mathscr{B}_e(n;\lambda_1,\lambda_2,\ldots,\lambda_e),$ we will first establish several key  lemmas.

In the following lemma, we assume that    $2\kappa \leq e,$ and we consider a chain $\mathcal{D}^{(1)}\subseteq \mathcal{D}^{(2)} \subseteq \cdots \subseteq \mathcal{D}^{(s+\theta_e)} $ of self-orthogonal codes of length $n$ over $\mathcal{T}_m$ satisfying  
conditions  \textbf{A1)} - \textbf{A5)} as stated in Lemma \ref{l3.4}.
Here, we  enumerate self-orthogonal codes $\mathcal{C}_e$ of type $\{\lambda_{1},\lambda_2,\ldots,\lambda_{e}\}$  and length $n$ over $\mathcal{R}_{e,m},$  satisfying $Tor_i(\mathcal{C}_e)=\mathcal{D}^{(i)}$ for $1 \leq i \leq s+\theta_e.$

\begin{lemma}\label{l3.1Keven}
Let $e\geq 3$ be an integer satisfying $2\kappa \leq e.$ 
Let  $\mathcal{D}^{(1)}\subseteq \mathcal{D}^{(2)} \subseteq \cdots \subseteq \mathcal{D}^{(s+\theta_e)} $ be a chain of self-orthogonal codes of length $n$ over $\mathcal{T}_m$ 
satisfying conditions \textbf{A1)} - \textbf{A5)} as stated in Lemma \ref{l3.4}.
The number of self-orthogonal codes $\mathcal{C}_e$ of type $\{\lambda_1,\lambda_2,\ldots,\lambda_e\}$ and length $n$ over $\mathcal{R}_{e,m}$ satisfying $Tor_i(\mathcal{C}_e)=\mathcal{D}^{(i)}$   for $1 \leq i \leq s+\theta_e,$ is given  by   \vspace{-2mm}\begin{equation*} \vspace{-1mm}  2^{\epsilon}(2^m)^{\sum\limits_{i=1}^{s}\Lambda_{i}(n-\Lambda_{i+1})+\sum\limits_{j=1}^{s-1+\theta_e}\Lambda_{s+j}(n-\Lambda_{s+j+1}-\Lambda_{s+\theta_e-j})-\sum\limits_{a=1}^{s-\kappa_1}\Lambda_a-(1-\theta_e)\frac{\Lambda_s(\Lambda_s-1)}{2}+\mu} \hspace{-2mm}
\prod\limits_{\ell=s+1+\theta_e}^{e}{\lambda_{\ell}+n-\Lambda_{\ell}-\Lambda_{e+1-\ell}\brack \lambda_{\ell}}_{2^m}, \vspace{-1mm}\end{equation*}   where 
\vspace{-1mm}\begin{equation*}(\epsilon,\mu)=\left\{\begin{array}{ll} (0,\omega) &\text{if~ } n\equiv 0,4\pmod 8 \text{ ~and there exists an integer ~}\omega \text{ ~satisfying  ~ }1\leq \omega \leq \kappa_1-\theta_e \text{ ~and} \\& \textbf{1} \in \mathcal{D}^{(s-\kappa_1-\omega+1)}\setminus  \mathcal{D}^{(s-\kappa_1-\omega)};\\ (1,\kappa_1) & \text{if } \textbf{1} \in \mathcal{D}^{(s-\kappa+\theta_e)}\text{ with either }n\equiv 0\pmod 8\text{ or }n\equiv 4 \pmod 8\text{ ~ and }~m\text{~ being even};\\
(0,\delta)& \text{if ~} n\equiv 2,6  \pmod8 \text{ ~and there exists an integer ~} \delta \text{~  satisfying ~} 1 \leq \delta \leq \mathrm{f}_{\kappa_1}, ~\eta_{2i-1}=0 \\& \text{for~ } 1 \leq i \leq \delta      \text{ ~and~ }  \textbf{1} \in \mathcal{D}^{(s-\kappa_1-\delta+1)}\setminus \mathcal{D}^{(s-\kappa_1-\delta)};\\
(0,\mathrm{f}_{\kappa_1+1}) &\text{if } n\equiv 2,6 \pmod 8, ~ \kappa \geq 4 \text{ ~is singly even,~} ~ \eta_{\kappa_1}=(\eta_0)^{\frac{3}{2}},~    \eta_{2j-1}=0  \text{ ~for~ } 1\leq j \leq \mathrm{f}_{\kappa_1-1}\\& \text{and ~} \textbf{1} \in \mathcal{D}^{(s-\mathrm{f}_{\kappa_1+1}+1-\kappa_1)}\setminus \mathcal{D}^{(s-\mathrm{f}_{\kappa_1+1}-\kappa_1)}; \\
(0,1)&\text{if } n\equiv 2,6\pmod 8, \kappa=2,~ e \text{ is even,~ }\eta_{1}=(\eta_0)^{\frac{3}{2}} \text{ ~and ~}\textbf{1} \in \mathcal{D}^{(s-\kappa_1)}\setminus \mathcal{D}^{(s-1-\kappa_1)};  
\\(0,0) & \text{otherwise.} \end{array} \right.\vspace{-1mm}\end{equation*}
\end{lemma}
\begin{proof} The desired result follows by applying Propositions \ref{p3.1}, \ref{p3.4Keven}, \ref{p3.6Keven},  \ref{p3.7Keven} and \ref{p3.9Keven}  in the case when $e$ is even, and Propositions \ref{p3.3b} - \ref{p3.5Keven}, \ref{p3.7Keven} and \ref{p3.9Keven} when $e$ is odd.
  \vspace{-2mm} \end{proof}
  In the following lemma, we assume that $e \geq 4$ and   $2\kappa > e,$ and we consider a chain $\mathcal{D}^{(1)}\subseteq \mathcal{D}^{(2)} \subseteq \cdots \subseteq \mathcal{D}^{(s+\theta_e)} $ of self-orthogonal codes of length $n$ over $\mathcal{T}_m$ 
satisfying  
conditions  \textbf{B1)} - \textbf{B4)} as stated in Lemma \ref{l3.4a}. 
Here, we  count  self-orthogonal codes $\mathcal{C}_e$ of type $\{\lambda_{1},\lambda_2,\ldots,\lambda_{e}\}$  and length $n$ over $\mathcal{R}_{e,m},$ satisfying $Tor_i(\mathcal{C}_e)=\mathcal{D}^{(i)}$ for $1 \leq i \leq s+\theta_e.$ 
\begin{lemma}\label{l3.2Keven}
Let $e\geq 4$ be an  integer satisfying $2\kappa > e.$ 
Let  $\mathcal{D}^{(1)}\subseteq \mathcal{D}^{(2)} \subseteq \cdots \subseteq \mathcal{D}^{(s+\theta_e)} $ be a chain of self-orthogonal codes of length $n$ over $\mathcal{T}_m$ 
satisfying  
conditions  \textbf{B1)} - \textbf{B4)} as stated in Lemma \ref{l3.4a}.  
The number of self-orthogonal codes $\mathcal{C}_e$ of type $\{\lambda_1,\lambda_2,\ldots,\lambda_e\}$ and length $n$ over $\mathcal{R}_{e,m}$ satisfying $Tor_i(\mathcal{C}_e)=\mathcal{D}^{(i)}$   for $1 \leq i \leq s+\theta_e,$ is given  by   \vspace{-2mm}\begin{equation*} \vspace{-1mm} (2^m)^{\sum\limits_{i=1}^{s}\Lambda_{i}(n-\Lambda_{i+1})+\sum\limits_{j=1}^{s-1+\theta_e}\Lambda_{s+j}(n-\Lambda_{s+j+1}-\Lambda_{s+\theta_e-j})-\sum\limits_{a=1}^{s-\kappa_1}\Lambda_a-(1-\theta_e)\frac{\Lambda_s(\Lambda_s-1)}{2}+\mu} 
\hspace{-1mm}\prod\limits_{\ell=s+1+\theta_e}^{e}{\lambda_{\ell}+n-\Lambda_{\ell}-\Lambda_{e+1-\ell}\brack \lambda_{\ell}}_{2^m},  \vspace{-1mm}\end{equation*}  
 where
\vspace{-2mm}\begin{equation*}\mu=\left\{ \begin{array}{cl}
     \omega & \text{if ~} n\equiv 0,4\pmod 8 \text{~and there exists an integer }\omega\text{~satisfying~}1\leq \omega \leq s-\kappa_1-1 \text{~and}\\&  \textbf{1} \in \mathcal{D}^{(s-\kappa_1+1-\omega)}\setminus \mathcal{D}^{(s-\kappa_1-\omega)};  \\
     s-\kappa_1 & \text{if } n\equiv 0,4\pmod 8 \text{~and~}\textbf{1}\in \mathcal{D}^{(1)};\\
     \delta &\text{if there exists an integer~} \delta  \text{  satisfying~} 1\leq \delta \leq \min\{e-\kappa,\mathrm{f}_{\kappa_1}\},  ~  \textbf{1} \in \mathcal{D}^{(s-\kappa_1-\delta+1)}\setminus \mathcal{D}^{(s-\kappa_1-\delta)}\\&  \text{and~}\eta_{2i-1}=0 \text{  for~} 1 \leq i \leq \delta ;  \\
     \mathrm{f}_{\kappa_1+1} & \text{if~}\kappa \text{~is singly even,~} e \geq \frac{3}{2}\kappa+1+\theta_e,~   \eta_{2j-1}=0 \text{  for~} 1\leq j \leq \mathrm{f}_{\kappa_1-1}, ~    \eta_{\kappa_1}=(\eta_0)^{\frac{3}{2}}\text{~and~}\\& \textbf{1} \in \mathcal{D}^{(s-\mathrm{f}_{\kappa_1+1}+1-\kappa_1)}\setminus \mathcal{D}^{(s-\mathrm{f}_{\kappa_1+1}-\kappa_1)}; \\
        0 & \text{otherwise.}
 \end{array}\right.\vspace{-1mm}\end{equation*}
\end{lemma}
\begin{proof} The desired result follows by applying Propositions \ref{p3.1}, \ref{p3.4Keven}, \ref{p3.8Keven} and \ref{p3.9Keven}  in the case when $e$ is even, and by Propositions \ref{p3.3b} - \ref{p3.4Keven}, \ref{p3.8Keven} and \ref{p3.9Keven} when $e$ is odd.
 \end{proof}

 We now proceed to count all possible choices for the chain $\mathcal{D}^{(1)}\subseteq \mathcal{D}^{(2)} \subseteq \cdots \subseteq \mathcal{D}^{(s+\theta_e)} $ of self-orthogonal codes of length $n$ over $\mathcal{T}_m,$  satisfying  conditions \textbf{A1)} - \textbf{A5)} as stated in Lemma \ref{l3.4} when  $2\kappa \leq e,$ and  conditions \textbf{B1)} - \textbf{B4)} as stated in Lemma \ref{l3.4a}  when $2\kappa >e.$  Towards this, we recall, from Section \ref{prelim}, that  $B_{m}(\textbf{a},\textbf{b})=\pi_0(\textbf{a}\cdot \textbf{b})=\pi_0(\sum\limits_{i=1}^{n}a_ib_i)$ for all $\textbf{a}=(a_1,a_2,\ldots,a_n),$ $\textbf{b}=(b_1,b_2,\ldots,b_n) \in \mathcal{T}_{m}^n.$ One can easily observe that every self-orthogonal code of length $n$ over $\mathcal{T}_{m}$ 
is contained in the set  $\mathcal{I}(\mathcal{T}_m^n)=\{\textbf{a} \in \mathcal{T}_{m}^n: B_{m}(\textbf{a},\textbf{a})=0 \}.$ Further, it is easy to see that the set $\mathcal{I}(\mathcal{T}_m^n)$ is an $(n-1)$-dimensional $\mathcal{T}_m$-linear subspace of $\mathcal{T}_m^n$  and 
that $\textbf{1} \in \mathcal{I}(\mathcal{T}_{m}^n)$ if and only if $n$ is even.  We will now distinguish the following two cases:     $\textbf{1} \notin \mathcal{D}^{(s-\kappa_1)},$ and $\textbf{1}\in\mathcal{D}^{(s-\kappa_1)}.$ 

 In the following lemma, we enumerate all possible choices for the chain $\mathcal{D}^{(1)}\subseteq \mathcal{D}^{(2)} \subseteq \cdots \subseteq \mathcal{D}^{(s+\theta_e)} $ of self-orthogonal codes of length $n$  over $\mathcal{T}_m,$ such that $\dim \mathcal{D}^{(i)}=\Lambda_i$ for $1 \leq i \leq s+\theta_e$ and  $\textbf{1} \notin \mathcal{D}^{(s-\kappa_1)}.$
\begin{lemma}\label{p5.1} For an integer $e \geq 3,$ let  $\mathfrak{B}(\lambda_1,\lambda_2,\ldots,\lambda_{s+\theta_e}) $ denote the number of distinct choices  for the chain $\mathcal{D}^{(1)}\subseteq \mathcal{D}^{(2)} \subseteq \cdots \subseteq \mathcal{D}^{(s+\theta_e)} $ of self-orthogonal codes of length $n$  over $\mathcal{T}_m,$  satisfying $\dim \mathcal{D}^{(i)}=\Lambda_i$ for $1 \leq i \leq s+\theta_e$ and  $\textbf{1} \notin \mathcal{D}^{(s-\kappa_1)}.$  We have
\vspace{-1mm}\small{\begin{equation*} \vspace{-1mm}\mathfrak{B}(\lambda_1,\lambda_2,\ldots,\lambda_{s+\theta_e})= \left\{\begin{array}{l}
		\displaystyle \prod\limits_{i=0}^{\Lambda_{s+\theta_e}-1} \left( \frac{2^{m(n-1-2i)}-1}{2^{m(i+1)}-1}\right)\prod\limits_{j=1}^{s+\theta_e}{\Lambda_j \brack \lambda_j}    \vspace{1mm}~~~~ \text{if  } n \text{ is odd};  \vspace{1mm}\\
 \displaystyle \prod\limits_{j=0}^{\Lambda_{s-\kappa_1}-2}\hspace{-1mm}\left(\frac{2^{m(n-2j-2)}-1}{2^{m(j+1)}-1}\right)\hspace{-1mm}\prod\limits_{\ell=\Lambda_{s-\kappa_1}}^{\Lambda_{s+\theta_e}-2}\hspace{-2mm}\left(\frac{2^{m(n-2\ell-2)}-1}{2^{m(\ell+1-\Lambda_{s-\kappa_1})}-1} \right)\prod\limits_{i=1}^{s-\kappa_1}{\Lambda_i \brack \lambda_i}_{2^m}\prod\limits_{a=s-\kappa_1+1}^{s+\theta_e}\hspace{-1mm}{\Lambda_a-\Lambda_{s-\kappa_1} \brack \lambda_a}_{2^m} \vspace{1mm} \\  \times \displaystyle 
  \left( \frac{ 2^{m(n-\Lambda_{s+\theta_e}-\Lambda_{s-\kappa_1})}-1+\big(2^{m(\Lambda_{s-\kappa_1})}-1\big)\big(2^{m(n-2\Lambda_{s+\theta_e})}+2^{m(\Lambda_{s+\theta_e}-\Lambda_{s-\kappa_1})}-2\big)}{\big(2^{m(\Lambda_{s+\theta_e}-\Lambda_{s-\kappa_1})}-1\big)\big(2^{m\Lambda_{s-\kappa_1}}-1\big)} \right) \vspace{1mm}
  \\
  \displaystyle \times \big(2^{m(n-2\Lambda_{s-\kappa_1})}-1\big) ~~~ \text{if~} n \text{~is even,~} \Lambda_{s-\kappa_1}\geq 1  \text{~and~}\Lambda_{s+\theta_e}\neq \Lambda_{s-\kappa_1};\vspace{1mm}\\
  \displaystyle \prod\limits_{i=0}^{\Lambda_{s+\theta_e}-2}\left(\frac{2^{m(n-2i-2)}-1}{2^{m(i+1)}-1}\right) \vspace{1mm} \left(\frac{2^{m(n-\Lambda_{s+\theta_e})}-1}{2^{m\Lambda_{s+\theta_e}}-1}\right)  \prod\limits_{j=s-\kappa_1+1}^{s+\theta_e}{\Lambda_j \brack \lambda_j}_{2^m} \vspace{1mm}\\ \text{if~} n \text{~is even,~} \Lambda_{s-\kappa_1}= 0 \text{~and~} \Lambda_{s+\theta_e}\geq 1; \vspace{1mm}\\
   \displaystyle \prod\limits_{i=0}^{\Lambda_{s-\kappa_1}-2}\left(\frac{2^{m(n-2i-2)}-1}{2^{m(i+1)}-1}\right) \vspace{1mm} \left(\frac{2^{m(n-\Lambda_{s-\kappa_1})}-2^{m\Lambda_{s-\kappa_1}}}{2^{m\Lambda_{s-\kappa_1}}-1}\right)  \prod\limits_{j=1}^{s-\kappa_1}{\Lambda_j \brack \lambda_j}_{2^m} \vspace{1mm}\\ \text{if~} n \text{~is even,~}\Lambda_{s-\kappa_1}\geq  1 \text{~and~} \Lambda_{s+\theta_e}= \Lambda_{s-\kappa_1};\vspace{1mm}\\
  1 ~~~~\text{if~} n \text{~is even and~} \Lambda_{s+\theta_e}=0.
\end{array}\right. \end{equation*}}\normalsize 
\end{lemma}
\begin{proof} To establish the result, we recall  that 
  $\textbf{1} \in \mathcal{I}(\mathcal{T}_{m}^n)$ if and only if $n$ is even.  Accordingly, we will distinguish the following two cases: (I) $n$ is odd, and (II) $n$ is even.  
\begin{enumerate}
\vspace{-1mm}\item[(I)] Let $n$ be odd. In this case, we see that $(\mathcal{I}(\mathcal{T}_m^n),B_m(\cdot,\cdot){\restriction_{\mathcal{I}(\mathcal{T}_m^n) \times (\mathcal{I}(\mathcal{T}_m^n)}})$ 
 is a symplectic space of dimension $n-1$ over $\mathcal{T}_m$ and  that $\textbf{1}\notin \mathcal{I}(\mathcal{T}_m^n).$   By  Exercise 8.1 (ii) of \cite{Taylor}, we  see that  the code $\mathcal{D}^{(s+\theta_e)}$ has precisely $\prod\limits_{i=0}^{\Lambda_{s+\theta_e}-1} \left( \frac{2^{m(n-1-2i)}-1}{2^{m(i+1)}-1}\right) $ distinct choices. Further, it is easy  to observe that the chain  $\mathcal{D}^{(1)}\subseteq \mathcal{D}^{(2)}\subseteq \cdots \subseteq \mathcal{D}^{(s+\theta_e)}$ of self-orthogonal codes of length $n$ over $\mathcal{T}_m$  has precisely  $\prod\limits_{i=0}^{\Lambda_{s+\theta_e}-1} \left( \frac{2^{m(n-1-2i)}-1}{2^{m(i+1)}-1}\right)\prod\limits_{j=1}^{s+\theta_e}{\Lambda_j \brack \lambda_j}$ distinct choices.   
\vspace{-1mm}\item[(II)] Next, let $n$ be even.  In this case, we have $\textbf{1}\in \mathcal{I}(\mathcal{T}_{m}^n).$ Let us choose  an $(n-2)$-dimensional $\mathcal{T}_{m}$-linear subspace  $\mathcal{V}_m$ of $\mathcal{I}(\mathcal{T}_{m}^n)$ such that  $\textbf{1} \notin \mathcal{V}_m.$ One can easily see that $\mathcal{I}(\mathcal{T}_{m}^n)=\mathcal{V}_m \perp \langle \textbf{1} \rangle$ and that $(\mathcal{V}_m, B_{m}(\cdot,\cdot){\restriction_{\mathcal{V}_m\times \mathcal{V}_m}})$ is an  $(n-2)$-dimensional symplectic space  over $\mathcal{T}_{m}.$ 

It is easy to see that $\mathfrak{B}(\lambda_1,\lambda_2,\ldots,\lambda_{s+\theta_e})=1$ if $\Lambda_{s+\theta_e}=0.$ So  we assume, throughout the proof,  that $\Lambda_{s+\theta_e}\geq 1.$ 

When $\Lambda_{s-\kappa_1}=0,$ working as in case (II) in the proof of 
Proposition 3.1 of Sharma and Kaur \cite{Sharma},  we get
\vspace{-2mm}\begin{equation*}\vspace{-2mm}
\mathfrak{B}(\lambda_1,\lambda_2,\ldots,\lambda_{s+\theta_e})=\prod\limits_{i=0}^{\Lambda_{s+\theta_e}-2}\left(\frac{2^{m(n-2i-2)}-1}{2^{m(i+1)}-1}\right) \vspace{1mm} \left(\frac{2^{m(n-\Lambda_{s+\theta_e})}-1}{2^{m\Lambda_{s+\theta_e}}-1}\right)  \prod\limits_{j=s-\kappa_1+1}^{s+\theta_e}{\Lambda_j \brack \lambda_j}_{2^m}.\end{equation*}  

On the other hand, when $\Lambda_{s-\kappa_1}\geq 1,$  working again as in case (II) in the proof of 
Proposition 3.1 of Sharma and Kaur \cite{Sharma}, we observe that the code $\mathcal{D}^{(s-\kappa_1)}$ assumes one of the following two forms: \begin{itemize}
    \item[(i)] $\langle \textbf{z}_1,\textbf{z}_2,\ldots,\textbf{z}_{\Lambda_{s-\kappa_1}} \rangle, $ where
$\textbf{z}_1,\textbf{z}_2,\ldots, \textbf{z}_{\Lambda_{s-\kappa_1}} \in \mathcal{V}_m$ are  mutually orthogonal, isotropic and linearly independent vectors 
 over  $\mathcal{T}_{m}.$   
\item[(ii)] $\langle \textbf{z}_1,\textbf{z}_2,\ldots,\textbf{z}_{\Lambda_{s-\kappa_1}-1},\textbf{1}+\textbf{z}_{\Lambda_{s-\kappa_1}} \rangle, $ where
$\textbf{z}_1,\textbf{z}_2,\ldots, \textbf{z}_{\Lambda_{s-\kappa_1}} \in \mathcal{V}_m\setminus\{\textbf{0}\}$ are such that   $\textbf{z}_1,\textbf{z}_2,\ldots,\textbf{z}_{\Lambda_{s-\kappa_1}-1},$ $\textbf{1}+\textbf{z}_{\Lambda_{s-\kappa_1}}$ are mutually orthogonal, isotropic  and   linearly independent vectors over $\mathcal{T}_{m}.$ 
\end{itemize}
Accordingly, we will examine the following two cases separately.
\begin{enumerate}
      \item[(i)] First of all, suppose that  the code $\mathcal{D}^{(s-\kappa_1)}$ is of the form $\langle \textbf{z}_1,\textbf{z}_2,\ldots,\textbf{z}_{\Lambda_{s-\kappa_1}} \rangle,$ where
$\textbf{z}_1,\textbf{z}_2,\ldots, \textbf{z}_{\Lambda_{s-\kappa_1}} \in \mathcal{V}_m$ are mutually orthogonal, isotropic and linearly independent  vectors
 over  $\mathcal{T}_{m}.$ Working as in case (II) in the proof of 
Proposition 3.1 of Sharma and Kaur \cite{Sharma} again, we see that such a code $\mathcal{D}^{(s-\kappa_1)}$ has precisely $\prod\limits_{j=0}^{\Lambda_{s-\kappa_1}-1}\left(\frac{2^{m(n-2j-2)}-1}{2^{m(j+1)}-1}\right)$ distinct choices. Now, 
it is easy to see that the  chain 
$\mathcal{D}^{(1)}\subseteq\mathcal{D}^{(2)}\subseteq \cdots \subseteq \mathcal{D}^{(s-\kappa_1)}$ of self-orthogonal codes of length $n$ over $\mathcal{T}_m$  has precisely $\prod\limits_{j=0}^{\Lambda_{s-\kappa_1}-1}\left(\frac{2^{m(n-2j-2)}-1}{2^{m(j+1)}-1}\right) \prod\limits_{i=1}^{s-\kappa_1}{\Lambda_i \brack \lambda_i}_{2^m}$
 distinct choices.
 
 Now, when $\Lambda_{s+\theta_e}=\Lambda_{s-\kappa_1},$ it is easy to see that the desired chain $\mathcal{D}^{(1)}\subseteq \mathcal{D}^{(2)}\subseteq \cdots \subseteq \mathcal{D}^{(s+\theta_e)}$ 
 has  precisely $\prod\limits_{j=0}^{\Lambda_{s-\kappa_1}-1}\left(\frac{2^{m(n-2j-2)}-1}{2^{m(j+1)}-1}\right) \prod\limits_{i=1}^{s-\kappa_1}{\Lambda_i \brack \lambda_i}_{2^m}$ distinct choices.
 
On the other hand, when $\Lambda_{s+\theta_e}\neq \Lambda_{s-\kappa_1},$   we will  first count the choices for a self-orthogonal code $\mathcal{D}^{(s+\theta_e)}$ satisfying  $\mathcal{D}^{(s-\kappa_1)}\subseteq \mathcal{D}^{(s+\theta_e)} \subseteq (\mathcal{D}^{(s-\kappa_1)})^{\perp_{B_m}}$  for a given  choice of the code $\mathcal{D}^{(s-\kappa_1)}.$  To do this, we note, by
\cite[pp. 69-70]{Taylor}, that the space $(\mathcal{I}(\mathcal{T}_m^n),B_m(\cdot,\cdot){\restriction_{\mathcal{I}(\mathcal{T}_m^n) \times (\mathcal{I}(\mathcal{T}_m^n)}})$  admits an orthogonal direct sum decomposition of the form: $\mathcal{I}(\mathcal{T}_m^n)=\mathcal{S}_0\perp \mathcal{S}_1\perp \langle \textbf{1}\rangle$ with  $\mathcal{S}_0=\langle \textbf{z}_1,\textbf{z}_1^{\prime}\rangle \perp \langle \textbf{z}_2,\textbf{z}_2^{\prime}\rangle \perp\cdots \perp \langle \textbf{z}_{\Lambda_{s-\kappa_1}},\textbf{z}_{\Lambda_{s-\kappa_1}}^{\prime}\rangle $ and $\mathcal{S}_1=\langle \textbf{z}_{\Lambda_{s-\kappa_1}+1},\textbf{z}_{\Lambda_{s-\kappa_1}+1}^{\prime}\rangle \perp \langle \textbf{z}_{\Lambda_{s-\kappa_1}+2},\textbf{z}_{\Lambda_{s-\kappa_1}+2}^{\prime}\rangle \perp  \cdots \perp \langle \mathbf{z}_{\frac{n-2}{2}},\mathbf{z}_{\frac{n-2}{2}}^{\prime} \rangle ,$ where $(\textbf{z}_i,\textbf{z}_i^\prime)$ is a hyperbolic pair in $\mathcal{I}(\mathcal{T}_m^n)$ for $1 \leq  i \leq \frac{n-2}{2}.$
We next observe that each $\textbf{w}\in (\mathcal{D}^{(s-\kappa_1)})^{\perp_{B_m}}$ can be uniquely written as $\textbf{w}=\sum\limits_{i=1}^{\frac{n-2}{2}}(a_i\textbf{z}_i+b_i\textbf{z}_i^{\prime})+\varepsilon \textbf{1},$ where  $a_i,b_i,\varepsilon \in \mathcal{T}_m$ for $1 \leq i \leq \frac{n-2}{2}.$ Since $\textbf{w}\in (\mathcal{D}^{(s-\kappa_1)})^{\perp_{B_m}},$ we have $B_m(\textbf{w},\textbf{z}_j)=0$ for $1\leq  j \leq \Lambda_{s-\kappa_1},$ which implies that $b_j=0$ for $1 \leq j\leq \Lambda_{s-\kappa_1}.$ Thus, each $\textbf{w}\in (\mathcal{D}^{(s-\kappa_1)})^{\perp_{B_m}}$ is  of the form  $\textbf{w}=\sum\limits_{i=1}^{\Lambda_{s-\kappa_1}}a_i\textbf{z}_i +\sum\limits_{j=\Lambda_{s-\kappa_1}+1}^{\frac{n-2}{2}}(a_j\textbf{z}_j+b_j\textbf{z}_j^{\prime})+\varepsilon \textbf{1},$ which implies that  $\langle \textbf{z}_1,\textbf{z}_2,\ldots,\textbf{z}_{\Lambda_{s-\kappa_1}},\textbf{w}\rangle=\langle \textbf{z}_1,\textbf{z}_2,\ldots,\textbf{z}_{\Lambda_{s-\kappa_1}}, \sum\limits_{j=\Lambda_{s-\kappa_1}+1}^{\frac{n-2}{2}}(a_j\textbf{z}_j+b_j\textbf{z}_j^{\prime})+\varepsilon\textbf{1}\rangle. $ In view of this, we assume,  without any loss of generality,  that the code $\mathcal{D}^{(s+\theta_e)}$  is one of the following two forms: 
\begin{itemize}
    \item[($\dagger$)] 
$\langle \textbf{z}_1,\textbf{z}_2,\ldots,\textbf{z}_{\Lambda_{s-\kappa_1}}\rangle \perp \langle \textbf{w}_1,\textbf{w}_2,\ldots, \textbf{w}_{\Lambda_{s+\theta_e}-\Lambda_{s-\kappa_1}} \rangle ,$ where $\textbf{w}_1,\textbf{w}_2,\ldots, \textbf{w}_{\Lambda_{s+\theta_e}-\Lambda_{s-\kappa_1}}\in \mathcal{S}_1$ are mutually orthogonal, isotropic and linearly independent vectors   over $\mathcal{T}_m.$  
\item[($\ddagger$)] $\langle \textbf{z}_1,\textbf{z}_2,\ldots,\textbf{z}_{\Lambda_{s-\kappa_1}}\rangle \perp \langle \textbf{w}_1,\textbf{w}_2,\ldots, \textbf{w}_{\Lambda_{s+\theta_e}-\Lambda_{s-\kappa_1}-1},\textbf{1}+\textbf{w}_{\Lambda_{s+\theta_e}-\Lambda_{s-\kappa_1}} \rangle ,$ where $\textbf{w}_1,\textbf{w}_2,\textbf{w}_3,\ldots,$ \\$ \textbf{w}_{\Lambda_{s+\theta_e}-\Lambda_{s-\kappa_1}-1}\in \mathcal{S}_1$ are mutually orthogonal isotropic vectors,  $\textbf{w}_{\Lambda_{s+\theta_e}-\Lambda_{s-\kappa_1}}\in \mathcal{S}_1$ and the vectors $\textbf{w}_1,\textbf{w}_2,\ldots\ \textbf{w}_{\Lambda_{s+\theta_e}-\Lambda_{s-\kappa_1}-1},\textbf{1}+\textbf{w}_{\Lambda_{s+\theta_e}-\Lambda_{s-\kappa_1}}$ are linearly independent over $\mathcal{T}_m.$ \end{itemize}
Again, working as in case (II) in the proof of 
Proposition 3.1 of Sharma and Kaur \cite{Sharma} and using Exercise 8.1(ii) of \cite{Taylor},  we see that the subspaces of the forms ($\dagger$) and ($\ddagger$), and hence the code $\mathcal{D}^{(s+\theta_e)}$ has  precisely 
 \vspace{-2mm}
\begin{equation*}
 \vspace{-2mm} \prod\limits_{i=\Lambda_{s-\kappa_1}}^{\Lambda_{s+\theta_e}-2}\left(\frac{2^{m(n-2i-2)}-1}{2^{m(i+1-\Lambda_{s-\kappa_1})}-1}\right) \vspace{1mm} \left(\frac{2^{m(n-\Lambda_{s+\theta_e}-\Lambda_{s-\kappa_1})}-1}{2^{m(\Lambda_{s+\theta_e}-\Lambda_{s-\kappa_1})}-1}\right)  \end{equation*}  distinct choices. Further, for a given choice of $\mathcal{D}^{(s-\kappa_1)},$  it is easy to see that the desired chain  $\mathcal{D}^{(s-\kappa_1+1)}\subseteq \mathcal{D}^{(s-\kappa_1+2)}\subseteq  \cdots\subseteq  \ \mathcal{D}^{(s+\theta_e)}$ of self-orthogonal codes over $\mathcal{T}_m$ has precisely
\vspace{-2mm}\begin{equation*}\vspace{-2mm}
\prod\limits_{i=\Lambda_{s-\kappa_1}}^{\Lambda_{s+\theta_e}-2}\left(\frac{2^{m(n-2i-2)}-1}{2^{m(i+1-\Lambda_{s-\kappa_1})}-1}\right) \vspace{1mm} \left(\frac{2^{m(n-\Lambda_{s+\theta_e}-\Lambda_{s-\kappa_1})}-1}{2^{m(\Lambda_{s+\theta_e}-\Lambda_{s-\kappa_1})}-1}\right)  \prod\limits_{j=s-\kappa_1+1}^{s+\theta_e}{\Lambda_j-\Lambda_{s-\kappa_1} \brack \lambda_j}_{2^m}\end{equation*}  distinct choices.

\item[(ii)] Next, let us suppose that the code $\mathcal{D}^{(s-\kappa_1)}$ is of the form 
 $\langle \textbf{z}_1,\textbf{z}_2,\ldots,\textbf{z}_{\Lambda_{s-\kappa_1}-1},\textbf{1}+\textbf{z}_{\Lambda_{s-\kappa_1}} \rangle$ with
$\textbf{z}_1,\textbf{z}_2,\ldots,$ $ \textbf{z}_{\Lambda_{s-\kappa_1}} \in \mathcal{V}_m\setminus\{\textbf{0}\},$ where $\textbf{z}_1,\textbf{z}_2,\ldots,\textbf{z}_{\Lambda_{s-\kappa_1}-1},\textbf{1}+\textbf{z}_{\Lambda_{s-\kappa_1}}$ are  mutually orthogonal, isotropic and linearly independent vectors over $\mathcal{T}_{m}.$  Again, working as in case (II) in the proof of 
Proposition 3.1 of Sharma and Kaur \cite{Sharma}, we see that the code $\mathcal{D}^{(s-\kappa_1)}$ has precisely $(2^{m(n-2\Lambda_{s-\kappa_1})}-1)\prod\limits_{j=0}^{\Lambda_{s-\kappa_1}-2}\left(\frac{2^{m(n-2j-2)}-1}{2^{m(j+1)}-1}\right)  $ distinct choices. Now, 
it is easy to see that the  chain 
$\mathcal{D}^{(1)}\subseteq\mathcal{D}^{(2)}\subseteq \cdots \subseteq \mathcal{D}^{(s-\kappa_1)}$ of self-orthogonal codes of length $n$ over $\mathcal{T}_m$  has precisely $$(2^{m(n-2\Lambda_{s-\kappa_1})}-1)\prod\limits_{j=0}^{\Lambda_{s-\kappa_1}-2}\left(\frac{2^{m(n-2j-2)}-1}{2^{m(j+1)}-1}\right) \prod\limits_{i=1}^{s-\kappa_1}{\Lambda_i \brack \lambda_i}_{2^m}$$
 distinct choices.  
 
 When $\Lambda_{s+\theta_e}=\Lambda_{s-\kappa_1},$ it is easy to see that the desired chain $\mathcal{D}^{(1)}\subseteq \mathcal{D}^{(2)}\subseteq \cdots \subseteq \mathcal{D}^{(s+\theta_e)}$ 
 has  precisely $(2^{m(n-2\Lambda_{s-\kappa_1})}-1)\prod\limits_{j=0}^{\Lambda_{s-\kappa_1}-2}\left(\frac{2^{m(n-2j-2)}-1}{2^{m(j+1)}-1}\right) \prod\limits_{i=1}^{s-\kappa_1}{\Lambda_i \brack \lambda_i}_{2^m}$ distinct choices.
 
 On the other hand, when $\Lambda_{s+\theta_e}\neq \Lambda_{s-\kappa_1},$   we will first enumerate the choices for a self-orthogonal code $\mathcal{D}^{(s+\theta_e)}$ satisfying  $\mathcal{D}^{(s-\kappa_1)}\subseteq 
   \mathcal{D}^{(s+\theta_e)} \subseteq (\mathcal{D}^{(s-\kappa_1)})^{\perp_{B_m}}$ for a given choice of $ \mathcal{D}^{(s-\kappa_1)}.$  To do this, we observe, working as in case (i), that the code $\mathcal{D}^{(s+\theta_e)}$ is either of the form $\langle \textbf{z}_1,\textbf{z}_2,\ldots,\textbf{z}_{\Lambda_{s-\kappa_1}-1},\textbf{1}+\textbf{z}_{\Lambda_{s-\kappa_1}}\rangle \perp \langle \textbf{w}_1,\textbf{w}_2,\ldots, \textbf{w}_{\Lambda_{s+\theta_e}-\Lambda_{s-\kappa_1}} \rangle ,$ or of the form  $\langle \textbf{z}_1,\textbf{z}_2,\ldots,\textbf{z}_{\Lambda_{s-\kappa_1}-1},\textbf{1}+\textbf{z}_{\Lambda_{s-\kappa_1}}\rangle \perp \langle \textbf{w}_1,\textbf{w}_2,\ldots, \\\textbf{w}_{\Lambda_{s+\theta_e}-\Lambda_{s-\kappa_1}-1},\textbf{1} \rangle ,$ where $\textbf{w}_1,\textbf{w}_2,\ldots, \textbf{w}_{\Lambda_{s+\theta_e}-\Lambda_{s-\kappa_1}}\in \mathcal{S}_1$ are mutually orthogonal, isotropic and linearly independent vectors  over $\mathcal{T}_m.$  Again, working as in case (II) in the proof of
Proposition 3.1 of Sharma and Kaur \cite{Sharma} and by  Exercise 8.1 (ii) of \cite{Taylor},  we see that such a subspace, and hence the code $\mathcal{D}^{(s+\theta_e)}$ has  precisely \vspace{-2mm}
\begin{equation*}
 \vspace{-2mm} \left(\frac{2^{m(n-2\Lambda_{s+\theta_e})}+2^{m(\Lambda_{s+\theta_e}-\Lambda_{s-\kappa_1})}-2}{2^{m(\Lambda_{s+\theta_e}-\Lambda_{s-\kappa_1})}-1} \right)\prod\limits_{i=\Lambda_{s-\kappa_1}}^{\Lambda_{s+\theta_e}-2}\left(\frac{2^{m(n-2i-2)}-1}{2^{m(i+1-\Lambda_{s-\kappa_1})}-1}\right)  
\end{equation*}  distinct choices. 
Furthermore, for a given choice of $\mathcal{D}^{(s-\kappa_1)},$  it is easy to see that the desired chain  $\mathcal{D}^{(s-\kappa_1+1)}\subseteq \mathcal{D}^{(s-\kappa_1+2)}\subseteq  \cdots\subseteq  \ \mathcal{D}^{(s+\theta_e)}$ of codes has precisely
\vspace{-2mm}\begin{equation*}\vspace{-2mm}
\left(\frac{2^{m(n-2\Lambda_{s+\theta_e})}+2^{m(\Lambda_{s+\theta_e}-\Lambda_{s-\kappa_1})}-2}{2^{m(\Lambda_{s+\theta_e}-\Lambda_{s-\kappa_1})}-1} \right)\prod\limits_{i=\Lambda_{s-\kappa_1}}^{\Lambda_{s+\theta_e}-2}\left(\frac{2^{m(n-2i-2)}-1}{2^{m(i+1-\Lambda_{s-\kappa_1})}-1}\right)   \prod\limits_{j=s-\kappa_1+1}^{s+\theta_e}{\Lambda_j-\Lambda_{s-\kappa_1} \brack \lambda_j}_{2^m}\end{equation*}  distinct choices.
\vspace{-1mm}\end{enumerate}  When $n$ is even, the desired result follows   by combining the cases (i) and (ii).
\vspace{-2mm}\end{enumerate}
\vspace{-4mm}\end{proof}
We will now consider the case  $\textbf{1}\in \mathcal{D}^{(s-\kappa_1)}.$ As $\mathcal{D}^{(s-\kappa_1)}$ is self-orthogonal  and $\textbf{1}\in \mathcal{D}^{(s-\kappa_1)},$ the integer $n$ must be   even. We will now distinguish the following two cases: (I) $n\equiv 0,4 \pmod 8,$ and (II) $n\equiv 2,6\pmod8.$ 
 
 Notably, when $n\equiv 0,4 \pmod 8,$ we see,  in view of Propositions \ref{p3.1}, \ref{p3.3b} - \ref{p3.6Keven}, that
the following  three cases arise: \begin{itemize}
\vspace{-1mm}\item[(Z1)] There exists an integer $\omega$ satisfying $1\leq \omega \leq \kappa_1-\theta_e$ if $2\kappa \leq e,$ whereas $1 \leq \omega \leq s-\kappa_1-1$ if $2\kappa > e,$ and $\textbf{1}\in \mathcal{D}^{(s-\kappa_1-\omega+1)}\setminus \mathcal{D}^{(s-\kappa_1-\omega)}.$  
\vspace{-1mm}\item[(Z2)] $\textbf{1}\in  \mathcal{D}^{(s-\kappa+\theta_e)}$ and  $2\kappa \leq e$ with  either $n\equiv 0\pmod 8$  or $n\equiv 4 \pmod 8$ and $m$ being even. 
 \vspace{-1mm}   \item[(Z3)]  $\textbf{1}\in\mathcal{D}^{(1)}$ and $2\kappa>e.$
\end{itemize}
In the following lemma, we  assume that $n\equiv 0,4\pmod 8$ and consider the case (Z1). Here, we  count the choices for the chain $\mathcal{D}^{(1)}\subseteq \mathcal{D}^{(2)} \subseteq \cdots \subseteq \mathcal{D}^{(s+\theta_e)} $ of self-orthogonal codes of length $n$  over $\mathcal{T}_m,$  such that $\dim \mathcal{D}^{(i)}=\Lambda_i$ for $1 \leq i \leq s+\theta_e.$
\begin{lemma}\label{p5.2}
    Let $e\geq 3$ be an  integer, and let $n\equiv 0,4\pmod 8.$  
    Suppose that there exists an integer $\omega$ satisfying $1\leq \omega \leq \kappa_1-\theta_e$ if $2\kappa \leq e,$ whereas $1 \leq \omega \leq s-\kappa_1-1$ if $2\kappa > e.$   Let  $\mathfrak{D}_\omega(\lambda_1,\lambda_2,\ldots,\lambda_{s+\theta_e}) $ be the number  of distinct choices  for the chain $\mathcal{D}^{(1)}\subseteq \mathcal{D}^{(2)} \subseteq \cdots \subseteq \mathcal{D}^{(s+\theta_e)} $ of self-orthogonal codes of length $n$  over $\mathcal{T}_m,$  such that $\dim \mathcal{D}^{(i)}=\Lambda_i$ for $1 \leq i \leq s+\theta_e$ and  $\textbf{1}\in \mathcal{D}^{(s-\kappa_1-\omega+1)}\setminus \mathcal{D}^{(s-\kappa_1-\omega)}.$   We have $\mathfrak{D}_\omega(\lambda_1,\lambda_2,\ldots,\lambda_{s+\theta_e})=0$ if $\Lambda_{s-\kappa_1-\omega+1}=0.$ Further, when $ \Lambda_{s-\kappa_1-\omega+1}\geq 1,$ we have  
\vspace{-2mm}\begin{eqnarray*}
\mathfrak{D}_\omega(\lambda_1,\lambda_2,\ldots,\lambda_{s+\theta_e})&=&  (2^m)^{\Lambda_{s-\kappa_1-\omega}}  {\Lambda_{s-\kappa_1-\omega+1}-1 \brack \Lambda_{s-\kappa_1-\omega}}_{2^m}  \prod\limits_{i=1}^{s-\kappa_1-\omega}{\Lambda_i \brack \lambda_i}_{2^m}\prod\limits_{b=s-\kappa_1-\omega+2}^{s+\theta_e}\hspace{-1mm}{\Lambda_b-\Lambda_{s-\kappa_1-\omega+1} \brack \lambda_b}_{2^m} \\ &&\times \hspace{-1mm}\prod\limits_{j=0}^{\Lambda_{s-\kappa_1-\omega+1}-2}\left( \frac{2^{m(n-2-2j)}-1}{2^{m(j+1)}-1}\right)  \prod\limits_{g=\Lambda_{s-\kappa_1-\omega+1}}^{\Lambda_{s+\theta_e}-1}\left( \frac{2^{m(n-2g)}-1}{2^{m(g+1-\Lambda_{s-\kappa_1-\omega+1})}-1}\right).\end{eqnarray*} \vspace{-2mm}
\end{lemma}
\vspace{-2mm}\begin{proof}
    Working as in case (II) in the proof of
Proposition 3.1 of Sharma and Kaur \cite{Sharma} and using  Exercise 8.1(ii) of \cite{Taylor},  we see that  the chain  $\mathcal{D}^{(1)}\subseteq  \mathcal{D}^{(2)}\subseteq \cdots\subseteq \mathcal{D}^{(s-\kappa_1-\omega+1)}$ of self-orthogonal codes of length $n$ over $\mathcal{T}_m$ satisfying  $\textbf{1} \in \mathcal{D}^{(s-\kappa_1-\omega+1)}\setminus \mathcal{D}^{(s-\kappa_1-\omega)}$  has precisely \vspace{-1mm}\begin{equation*}
  \vspace{-1mm}  \prod\limits_{j=0}^{\Lambda_{s-\kappa_1-\omega+1}-2}\left( \frac{2^{m(n-2-2j)}-1}{2^{m(j+1)}-1}\right)  \left ( {\Lambda_{s-\kappa_1-\omega+1} \brack \Lambda_{s-\kappa_1-\omega}}_{2^m} -{\Lambda_{s-\kappa_1-\omega+1}-1 \brack \Lambda_{s-\kappa_1-\omega}-1}_{2^m}\right)\prod\limits_{i=1}^{s-\kappa_1-\omega} {\Lambda_i \brack \lambda_i}_{2^m}\end{equation*}
distinct choices. Further, by the Pascal's Identity for Gaussian binomial coefficients, we have  \vspace{-1mm}\begin{equation*}\vspace{-1mm}
     {\Lambda_{s-\kappa_1-\omega+1} \brack \Lambda_{s-\kappa_1-\omega}}_{2^m} -{\Lambda_{s-\kappa_1-\omega+1}-1 \brack \Lambda_{s-\kappa_1-\omega}-1}_{2^m}= (2^m)^{\Lambda_{s-\kappa_1-\omega}}{\Lambda_{s-\kappa_1-\omega+1}-1 \brack \Lambda_{s-\kappa_1-\omega}}_{2^m}.\end{equation*}  This implies that the  chain  $\mathcal{D}^{(1)}\subseteq  \mathcal{D}^{(2)}\subseteq \cdots\subseteq \mathcal{D}^{(s-\kappa_1-\omega+1)}$ of self-orthogonal codes over $\mathcal{T}_m$  satisfying  $\textbf{1} \in \mathcal{D}^{(s-\kappa_1-\omega+1)}\setminus \mathcal{D}^{(s-\kappa_1-\omega)}$  has precisely \vspace{-2mm}\begin{equation*}
  \vspace{-2mm}  (2^m)^{\Lambda_{s-\kappa_1-\omega}}
   \prod\limits_{j=0}^{\Lambda_{s-\kappa_1-\omega+1}-2}\left( \frac{2^{m(n-2-2j)}-1}{2^{m(j+1)}-1}\right) {\Lambda_{s-\kappa_1-\omega+1}-1 \brack \Lambda_{s-\kappa_1-\omega}}_{2^m}  \prod\limits_{i=1}^{s-\kappa_1-\omega} {\Lambda_i \brack \lambda_i}_{2^m}
  \end{equation*}
distinct choices.
Now, for a given choice of the code $ \mathcal{D}^{(s-\kappa_1-\omega+1)}$ containing $\textbf{1},$  we need to count the choices for 
 the desired chain  $\mathcal{D}^{(s-\kappa_1-\omega+2)}\subseteq \mathcal{D}^{(s-\kappa_1-\omega+3)}\subseteq \cdots\subseteq  \mathcal{D}^{(s+\theta_e)}$ of self-orthogonal codes over $\mathcal{T}_m$. For this, we see,  working as in case (II) in the proof of
Lemma \ref{p5.1}  and by  Exercise 8.1(ii) of \cite{Taylor},   that  the desired chain  $\mathcal{D}^{(s-\kappa_1-\omega+2)}\subseteq \mathcal{D}^{(s-\kappa_1-\omega+3)}\subseteq \cdots\subseteq  \mathcal{D}^{(s+\theta_e)}$  has precisely 
\vspace{-2mm}\begin{equation*}
  \prod\limits_{b=s-\kappa_1-\omega+2}^{s+\theta_e}\hspace{-1mm}{\Lambda_b-\Lambda_{s-\kappa_1-\omega+1} \brack \lambda_b}_{2^m}   \prod\limits_{g=\Lambda_{s-\kappa_1-\omega+1}}^{\Lambda_{s+\theta_e}-1}\left( \frac{2^{m(n-2g)}-1}{2^{m(g+1-\Lambda_{s-\kappa_1-\omega+1})}-1}\right) \vspace{-2mm}\end{equation*} distinct choices.  From this, the desired result follows immediately. 
\vspace{-1mm}
\end{proof}
 In the following lemma, we  consider the case (Z2). Here, we count the choices for the chain $\mathcal{D}^{(1)}\subseteq \mathcal{D}^{(2)} \subseteq \cdots \subseteq \mathcal{D}^{(s+\theta_e)} $ of self-orthogonal codes of length $n$  over $\mathcal{T}_m, $ such that $\dim \mathcal{D}^{(i)}=\Lambda_i$ for $1 \leq i \leq s+\theta_e.$

 \begin{lemma}\label{p5.3}
    Let $e\geq 3$ be an  integer satisfying $2\kappa \leq e.$ Let  us suppose that  either $n\equiv 0\pmod 8$  or $n\equiv 4 \pmod 8$ and $m$ is even.   Let  $\mathcal{X}(\lambda_1,\lambda_2,\ldots,\lambda_{s+\theta_e}) $ denote the number of  all distinct choices  for the chain $\mathcal{D}^{(1)}\subseteq \mathcal{D}^{(2)} \subseteq \cdots \subseteq \mathcal{D}^{(s+\theta_e)} $ of self-orthogonal codes of length $n$  over $\mathcal{T}_m,$  such that $\dim \mathcal{D}^{(i)}=\Lambda_i$ for $1 \leq i \leq s+\theta_e$ and  $\textbf{1}\in  \mathcal{D}^{(s-\kappa+\theta_e)}.$ We have $\mathcal{X}(\lambda_1,\lambda_2,\ldots,\lambda_{s+\theta_e})=0$ if $\Lambda_{s-\kappa+\theta_e}=0.$ Further, when $ \Lambda_{s-\kappa+\theta_e}\geq 1,$ we have  
\vspace{-2mm}\begin{eqnarray*}
\mathcal{X}(\lambda_1,\lambda_2,\ldots,\lambda_{s+\theta_e})&=& \prod\limits_{\ell=1}^{s-\kappa+\theta_e}{\Lambda_{\ell} \brack \lambda_{\ell}}_{2^m} \prod\limits_{b=s-\kappa+1+\theta_e}^{s+\theta_e}{\Lambda_b-\Lambda_{s-\kappa+\theta_e} \brack \lambda_b}_{2^m}        \prod\limits_{j=0}^{\Lambda_{s-\kappa+\theta_e}-2}\left( \frac{2^{m(n-2-2j)}-1}{2^{m(j+1)}-1}\right)   \\&&\times  \prod\limits_{g=\Lambda_{s-\kappa+\theta_e}}^{\Lambda_{s+\theta_e}-1}\left( \frac{2^{m(n-2g)}-1}{2^{m(g+1-\Lambda_{s-\kappa+\theta_e})}-1}\right)  . \end{eqnarray*} 
\vspace{-3mm}\end{lemma}
\vspace{-2mm} \begin{proof}
Working as in Lemma \ref{p5.2}, the desired result follows immediately.
\end{proof}

 In the following lemma, we assume that $n\equiv 0,4\pmod8$ and $e \geq 4$ is an integer satisfying $2 \kappa > e.$ Here, we consider the case (Z3) and count the choices for the chain $\mathcal{D}^{(1)}\subseteq \mathcal{D}^{(2)} \subseteq \cdots \subseteq \mathcal{D}^{(s+\theta_e)} $ of self-orthogonal codes of length $n$  over $\mathcal{T}_m, $   such that $\dim \mathcal{D}^{(i)}=\Lambda_i$ for $1 \leq i \leq s+\theta_e.$
 \begin{lemma}\label{p5.4}
    Let $e\geq 4$ be an  integer satisfying $2\kappa > e,$ and let $n\equiv 0,4\pmod8.$  Let $\mathcal{Z}(\lambda_1,\lambda_2,\ldots,\lambda_{s+\theta}) $  denote the number of distinct choices  for the chain $\mathcal{D}^{(1)}\subseteq \mathcal{D}^{(2)} \subseteq \cdots \subseteq \mathcal{D}^{(s+\theta_e)} $ of self-orthogonal codes of length $n$  over $\mathcal{T}_m,$ such that $\dim \mathcal{D}^{(i)}=\Lambda_i$ for $1 \leq i \leq s+\theta_e$ and  $\textbf{1}\in  \mathcal{D}^{(1)}.$  We have $ \mathcal{Z}(\lambda_1,\lambda_2,\ldots,\lambda_{s+\theta_e})=0$ if $\Lambda_{1}=0.$ Further, when $ \Lambda_{1}\geq 1,$ we have  
\vspace{-2mm}\begin{equation*}\vspace{-2mm}
   \mathcal{Z}(\lambda_1,\lambda_2,\ldots,\lambda_{s+\theta_e})= \prod\limits_{i=0}^{\Lambda_1-2} \left( \frac{2^{m(n-2-2i)}-1}{2^{m(i+1)}-1}\right)\prod\limits_{g=\Lambda_{1}}^{\Lambda_{s+\theta_e}-1}\hspace{-1mm}\left( \frac{2^{m(n-2g)}-1}{2^{m(g+1-\Lambda_{1})}-1}\right)  \prod\limits_{d=2}^{s+\theta_e}{\Lambda_d-\Lambda_{1} \brack \lambda_d}_{2^m} . \end{equation*} 
\end{lemma}
 \begin{proof}
Working as in Lemma \ref{p5.2}, the desired result follows immediately.
\end{proof}

We will next consider the case  $n\equiv 2,6\pmod 8.$  Here,  in view of Propositions \ref{p3.1}, \ref{p3.3b} - \ref{p3.4Keven}, the following two  cases arise: 
\begin{itemize}
    \item[(L1)] There exists an integer $\delta$  satisfying $ 1 \leq \delta \leq \mathrm{f}_{\kappa_1}$  if $2\kappa \leq e,$ while $1\leq \delta \leq \min\{e-\kappa,\mathrm{f}_{\kappa_1}\}$ if $2\kappa>e,$  such that  $\eta_{2i-1}=0 $ for $ 1 \leq i \leq \delta $ and $\textbf{1} \in \mathcal{D}^{(s-\kappa_1-\delta+1)}\setminus \mathcal{D}^{(s-\kappa_1-\delta)}.$  
\item[(L2)] When $2\kappa \leq  e, $ either $\kappa \geq 4$   is a singly even integer or $\kappa=2$ and $e$ is even. On the other hand, when     $2\kappa > e, $     $\kappa$ is a singly even integer  and $e \geq \frac{3}{2}\kappa+1+\theta_e $. In both cases, we have $ \eta_{\kappa_1}=(\eta_0)^{\frac{3}{2}} $ and  $ \eta_{2j-1}=0$  for $ 1\leq j \leq \mathrm{f}_{\kappa_1-1}.$  In addition,   we have $\textbf{1} \in \mathcal{D}^{(s-\mathrm{f}_{\kappa_1+1}+1-\kappa_1)}\setminus  \mathcal{D}^{(s-\mathrm{f}_{\kappa_1+1}-\kappa_1)}.$     

\end{itemize}

In the following lemma, we  consider the former case and  count the choices for the chain $\mathcal{D}^{(1)}\subseteq \mathcal{D}^{(2)} \subseteq \cdots \subseteq \mathcal{D}^{(s+\theta_e)} $ of self-orthogonal codes of length $n$  over $\mathcal{T}_m,$ such that $\dim \mathcal{D}^{(i)}=\Lambda_i$ for $1 \leq i \leq s+\theta_e.$
\begin{lemma}\label{p5.7}
    Let $e\geq 3$ be an  integer, and let $n\equiv 2,6\pmod 8.$  
  Let us suppose that there exists an integer $\delta$  satisfying $ 1 \leq \delta \leq \mathrm{f}_{\kappa_1}$  if $2\kappa \leq e,$ while $1\leq \delta \leq \min\{e-\kappa,\mathrm{f}_{\kappa_1}\}$ if $2\kappa>e,$  and  $\eta_{2i-1}=0 $ for $ 1 \leq i \leq \delta .$ 
Let  $\mathcal{W}_\delta(\lambda_1,\lambda_2,\ldots,\lambda_{s+\theta_e}) $ denote the number  of distinct choices  for the chain $\mathcal{D}^{(1)}\subseteq \mathcal{D}^{(2)} \subseteq \cdots \subseteq \mathcal{D}^{(s+\theta_e)} $ of self-orthogonal codes of length $n$  over $\mathcal{T}_m,$  such that $\dim \mathcal{D}^{(i)}=\Lambda_i$ for $1 \leq i \leq s+\theta_e$ and $\textbf{1} \in \mathcal{D}^{(s-\kappa_1-\delta+1)}\setminus \mathcal{D}^{(s-\kappa_1-\delta)}.$ We have $ \mathcal{W}_\delta(\lambda_1,\lambda_2,\ldots,\lambda_{s+\theta_e})=0$ if $\Lambda_{s-\kappa_1-\delta+1}=0.$ Further, when  $ \Lambda_{s-\kappa_1-\delta+1}\geq 1,$ we have   
\vspace{-2mm}\begin{eqnarray*}
\mathcal{W}_\delta(\lambda_1,\lambda_2,\ldots,\lambda_{s+\theta_e})&=&  (2^m)^{\Lambda_{s-\kappa_1-\delta}}  {\Lambda_{s-\kappa_1-\delta+1}-1 \brack \Lambda_{s-\kappa_1-\delta}}_{2^m}  \prod\limits_{i=1}^{s-\kappa_1-\delta}{\Lambda_i \brack \lambda_i}_{2^m}\prod\limits_{b=s-\kappa_1-\delta+2}^{s+\theta_e}\hspace{-1mm}{\Lambda_b-\Lambda_{s-\kappa_1-\delta+1} \brack \lambda_b}_{2^m} \\ &&\times \hspace{-1mm}\prod\limits_{j=0}^{\Lambda_{s-\kappa_1-\delta+1}-2}\left( \frac{2^{m(n-2-2j)}-1}{2^{m(j+1)}-1}\right)  \prod\limits_{g=\Lambda_{s-\kappa_1-\delta+1}}^{\Lambda_{s+\theta_e}-1}\left( \frac{2^{m(n-2g)}-1}{2^{m(g+1-\Lambda_{s-\kappa_1-\delta+1})}-1}\right).\end{eqnarray*} 
\end{lemma}
\vspace{-1mm}\begin{proof}
    Working as in Lemma \ref{p5.2}, the desired result follows immediately.
\end{proof}

In the following lemma, we  consider the latter case and  count the choices for the chain $\mathcal{D}^{(1)}\subseteq \mathcal{D}^{(2)} \subseteq \cdots \subseteq \mathcal{D}^{(s+\theta_e)} $ of self-orthogonal codes of length $n$  over $\mathcal{T}_m,$ such that $\dim \mathcal{D}^{(i)}=\Lambda_i$ for $1 \leq i \leq s+\theta_e.$

\begin{lemma}\label{p5.8}
    Let $e\geq 3$ be an  integer,  and let $n\equiv 2,6\pmod 8.$ 
    Let us suppose that either $\kappa \geq 4$   is a singly even integer or $\kappa=2$ and $e$ is even when  $2\kappa \leq  e, $ while       $\kappa$ is a singly even integer  and $e \geq \frac{3}{2}\kappa+1+\theta_e $ when  $2\kappa > e. $  In both cases, let us suppose that  $ \eta_{\kappa_1}=(\eta_0)^{\frac{3}{2}} $ and  $ \eta_{2j-1}=0$  for $ 1\leq j \leq \mathrm{f}_{\kappa_1-1}.$    
 Let    $\mathcal{Y}(\lambda_1,\lambda_2,\ldots,\lambda_{s+\theta_e})$ denote the number of distinct choices  for the chain $\mathcal{D}^{(1)}\subseteq \mathcal{D}^{(2)} \subseteq \cdots \subseteq \mathcal{D}^{(s+\theta_e)} $ of self-orthogonal codes of length $n$  over $\mathcal{T}_m,$  such that $\dim \mathcal{D}^{(i)}=\Lambda_i$ for $1 \leq i \leq s+\theta_e$ and  $\textbf{1} \in \mathcal{D}^{(s-\mathrm{f}_{\kappa_1+1}+1-\kappa_1)}\setminus  \mathcal{D}^{(s-\mathrm{f}_{\kappa_1+1}-\kappa_1)}.$ We have $ \mathcal{Y}(\lambda_1,\lambda_2,\ldots,\lambda_{s+\theta_e})=0$ if $\Lambda_{s-\mathrm{f}_{\kappa_1+1}+1-\kappa_1}=0.$ Further, when  $ \Lambda_{s-\mathrm{f}_{\kappa_1+1}+1-\kappa_1}\geq 1,$ we have    
    \vspace{-2mm} \begin{eqnarray*}
\mathcal{Y}(\lambda_1,\lambda_2,\ldots,\lambda_{s+\theta_e})\hspace{-2mm}&=&  \hspace{-2mm}(2^m)^{\Lambda_{s-\kappa_1-\mathrm{f}_{\kappa_1+1}}}\hspace{-3.mm}\prod\limits_{j=0}^{\Lambda_{s-\kappa_1-\mathrm{f}_{\kappa_1+1}+1}-2}\hspace{-3mm}\left( \frac{2^{m(n-2-2j)}-1}{2^{m(j+1)}-1}\right)  \hspace{-1.5mm} \prod\limits_{g=\Lambda_{s-\kappa_1-\mathrm{f}_{\kappa_1+1}+1}}^{\Lambda_{s+\theta_e}-1}\hspace{-1.5mm} \left( \frac{2^{m(n-2g)}-1}{2^{m(g+1-\Lambda_{s-\kappa_1-\mathrm{f}_{\kappa_1+1}+1})}-1}\right)  \\ &&\times {\Lambda_{s-\kappa_1-\mathrm{f}_{\kappa_1+1}+1}-1 \brack \Lambda_{s-\kappa_1-\mathrm{f}_{\kappa_1+1}}}_{2^m}  \prod\limits_{i=1}^{s-\kappa_1-\mathrm{f}_{\kappa_1+1}}{\Lambda_i \brack \lambda_i}_{2^m}\prod\limits_{b=s-\kappa_1-\mathrm{f}_{\kappa_1+1}+2}^{s+\theta_e}{\Lambda_b-\Lambda_{s-\kappa_1-\mathrm{f}_{\kappa_1+1}+1} \brack \lambda_b}_{2^m}.\end{eqnarray*} 
\end{lemma}
\vspace{-1mm}\begin{proof}
    Working as in Lemma \ref{p5.2}, the desired result follows immediately.
\end{proof}

 In the following theorem, we consider the case $e=\kappa=2,$ and provide an explicit enumeration formula for the number $\mathscr{M}_2(n;\lambda_1,\lambda_2).$  When  $e=\kappa=2,$ note that the ring $\mathcal{R}_{e,m}$ reduces to the quasi-Galois ring $\mathbb{F}_{2^m}[x]/\langle x^2\rangle.$ \begin{theorem}\label{t4.4e=k=2}
  For $e=\kappa=2,$  we have 
  \vspace{-2mm}\begin{eqnarray*} \vspace{-2mm} \mathscr{M}_2(n;\lambda_1,\lambda_2)=\displaystyle
	\displaystyle \widetilde{\mathcal{S}}(n;\lambda_1)(2^m)^{\lambda_1(n-\Lambda_2)-\frac{\lambda_1(\lambda_1-1)}{2}}{\lambda_{2}+n-\Lambda_{2}-\Lambda_1 \brack \lambda_{2}}_{2^m}, \end{eqnarray*} where 
\vspace{-3mm}\begin{equation*}\widetilde{\mathcal{S}}(n;\lambda_1)=\left\{ \begin{array}{ll}  \displaystyle \prod\limits_{i=0}^{\lambda_1-1}\left(\frac{2^{m(n-1-2i)}-1}{2^{m(i+1)}-1}\right) & \text{if }n\text{ is odd};\vspace{0.5mm}\\ \displaystyle  \left(\frac{2^{m(n-\lambda_1)}-1}{2^{m\lambda_1}-1}\right)\prod\limits_{i=0}^{\lambda_1-2}\left( \frac{2^{m(n-2-2i)}-1}{2^{m(i+1)}-1}\right) & \text{if } \lambda_1\neq 0 \text{ and } n \text{ is even; }\vspace{0.5mm}\\
    ~~~~~~1 & \text{if } \lambda_1=0 \text{ and } n \text{ is even.} \end{array}\right.\vspace{-2mm}\end{equation*}
 \vspace{-2mm}
\end{theorem}
\vspace{-2mm}\begin{proof}
    The desired result follows  by  Proposition \ref{pe=k=2} and using Theorem 2 of Pless \cite{V}.
\end{proof}
 In the following theorem, we consider the case $e=3$ and $\kappa=2,$ and provide an explicit enumeration formula for the number $\mathscr{M}_3(n;\lambda_1,\lambda_2,\lambda_3).$  
\vspace{-1mm}\begin{theorem}\label{t4.4e=3}
  For $e=3$ and $\kappa=2,$  we have 
  \vspace{-2mm}\begin{equation*} \vspace{-2mm} \mathscr{M}_3(n;\lambda_1,\lambda_2,\lambda_3)=\displaystyle
	\displaystyle \widehat{\mathcal{S}}(n;\Lambda_2)(2^m)^{\lambda_3\lambda_1+(\Lambda_1+\Lambda_{2})(n-\Lambda_{3}-\Lambda_1)+\Lambda_1^2}{\Lambda_2 \brack \lambda_1}_{2^m}{\lambda_{3}+n-\Lambda_{3}-\Lambda_1 \brack \lambda_{3}}_{2^m}, \end{equation*} where \vspace{-3mm}\begin{equation*}\vspace{-2mm}\widehat{\mathcal{S}}(n;\Lambda_2)=\left\{ \begin{array}{ll}  \displaystyle \prod\limits_{i=0}^{\Lambda_2-1}\left(\frac{2^{m(n-1-2i)}-1}{2^{m(i+1)}-1}\right) & \text{if }n\text{ is odd};\vspace{0.5mm}\\ \displaystyle  \left(\frac{2^{m(n-\Lambda_2)}-1}{2^{m\Lambda_2}-1}\right)\prod\limits_{i=0}^{\Lambda_2-2}\left( \frac{2^{m(n-2-2i)}-1}{2^{m(i+1)}-1}\right) & \text{if } \Lambda_2\neq 0 \text{ and } n \text{ is even; }\vspace{0.5mm}\\
    ~~~~~~1 & \text{if } \Lambda_2=0 \text{ and } n \text{ is even.} \end{array}\right.\end{equation*}
 \vspace{-2mm}
\end{theorem}
\begin{proof}
    The desired result follows  by  Proposition \ref{p3.3be=3} and using Theorem 2 of Pless \cite{V}.
\end{proof}
In the following theorem, we obtain  an explicit enumeration formula for the number $\mathscr{M}_e(n;\lambda_1,\lambda_2,\ldots,\lambda_e)$ when $e\geq 4$ and   $n\equiv 1,3,5,7\pmod 8.$  

\begin{theorem}\label{t4.1Keven} 
For  $e\geq 4$   and  $n\equiv 1,3,5,7\pmod 8,$   we have  
\vspace{-2mm}\begin{eqnarray*} \vspace{-2mm} \mathscr{M}_e(n;\lambda_1,\lambda_2,\ldots,\lambda_e)=\displaystyle
	 \displaystyle  (2^m)^{\sum\limits_{i=1}^{s}\Lambda_{i}(n-\Lambda_{i+1})+\sum\limits_{j=1}^{s-1+\theta_e}\Lambda_{s+j}(n-\Lambda_{s+j+1}-\Lambda_{s-j+\theta_e})-\sum\limits_{\ell=1}^{s-\kappa_1}\Lambda_\ell-(1-\theta_e)\frac{\Lambda_s(\Lambda_s-1)}{2}} \\
	 \displaystyle  \times \mathfrak{B}(\lambda_1,\lambda_2,\ldots,\lambda_{s+\theta_e})\prod\limits_{a=s+1+\theta_e}^{e}{\lambda_{a}+n-\Lambda_{a}-\Lambda_{e+1-a}\brack \lambda_{a}}_{2^m}, ~~~~~~~~~~~~~~~ \end{eqnarray*}where the number $ \mathfrak{B}(\lambda_1,\lambda_2,\ldots,\lambda_{s+\theta_e}) $ is as determined in Lemma \ref{p5.1}.
     \end{theorem}
     \vspace{-1mm}\begin{proof}
When $2\kappa\leq e,$ the desired result follows  by   Lemmas \ref{l3.1Keven} and  \ref{p5.1}.
On the other hand, when  $2\kappa>e,$ we get the desired result by applying Lemmas \ref{l3.2Keven} and \ref{p5.1}.         \vspace{-2mm} \end{proof}
\vspace{-3mm} \small
\begin{table}[h!]
\centering
\resizebox{\textwidth}{!}{%
\begin{tabular}{|c|c||c|c||c|c|}
\hline
$\{\lambda_1,\lambda_2,\lambda_3,\lambda_4\}$ & $\mathscr{M}_4(3; \lambda_1,\lambda_2,\lambda_3,\lambda_4)$ &
$\{\lambda_1,\lambda_2,\lambda_3,\lambda_4\}$ & $\mathscr{M}_4(3; \lambda_1,\lambda_2,\lambda_3,\lambda_4)$ &
$\{\lambda_1,\lambda_2,\lambda_3,\lambda_4\}$ & $\mathscr{M}_4(3; \lambda_1,\lambda_2,\lambda_3,\lambda_4)$ \\
\hline
$\{0, 0, 0, 1\} $ & $21$ & $\{0, 1, 0, 0\}$ & $1280$ & $\{0, 0, 3, 0\}$ & $1$ \\\hline
$\{0, 0, 0, 2\}$ & $21$ & $\{0, 1, 0, 1\}$ & $1600$ &  $\{1, 0, 1, 0\}$ & $80$  \\\hline
$\{0, 0, 0, 3\}$ & $1$ & $\{0, 1, 0, 2\}$ & $80$ & $\{0, 0, 2, 1\}$ & $21$ \\\hline
$\{0, 0, 1, 0\}$ & $336$ & $\{0, 1, 1, 0\}$ & $320$ & $\{1, 0, 0, 1\}$ & $320$  \\\hline
$\{0, 0, 1, 1\}$ & $420$ & $\{0, 1, 1, 1\}$ & $20$ & $\{0, 0, 2, 0\}$ & $336$   \\\hline
$\{0, 0, 1, 2\}$ & $21$ & $\{1, 0, 0, 0\}$ & $1280$  & &  \\\hline
\end{tabular}%
}
\caption{The number $\mathscr{M}_4(3; \lambda_1,\lambda_2,\lambda_3,\lambda_4)$ of self-orthogonal codes of length $3$ and type $\{\lambda_1,\lambda_2,\lambda_3,\lambda_4\}$ over $\mathcal{R}_{4,2} = GR(4,2)[x]/\langle x^2+2,2x^2\rangle$}
\label{Tab1}
\end{table}
\vspace{-3mm} \small
\begin{table}[h!]
\centering
\resizebox{\textwidth}{!}{%
\begin{tabular}{|c|c||c|c||c|c|}
\hline
$\{\lambda_1,\lambda_2,\ldots,\lambda_5\}$ & $\mathscr{M}_5(3; \lambda_1,\lambda_2,\ldots,\lambda_5)$ &
$\{\lambda_1,\lambda_2,\ldots,\lambda_5\}$ & $\mathscr{M}_5(3; \lambda_1,\lambda_2,\ldots,\lambda_5)$ &
$\{\lambda_1,\lambda_2,\ldots,\lambda_5\}$ & $\mathscr{M}_5(3; \lambda_1,\lambda_2,\ldots,\lambda_5)$ \\
\hline
$\{0,0,0,0,1\}$ & $7$ & $\{0,0,0,0,2\}$ & $7$ & $\{0,0,0,0,3\}$ & $1$ \\

\hline
$\{0,0,0,1,0\}$ & $28$ & $\{0,0,0,1,1\}$ & $42$ & $\{0,0,0,1,2\}$ & $7$ \\
\hline 
$\{0,0,0,2,0\}$ & $28$ & $\{0,0,0,2,1\}$ & $7$ & $\{0,0,0,3,0\}$ & $1$ \\
\hline 
$\{0,0,1,0,0\}$ & $48$ & $\{0,0,1,0,1\}$ & $72$ & $\{0,0,1,0,2\}$ & $12$ \\
\hline 
$\{0,0,1,1,0\}$ & $72$ & $\{0,0,1,1,1\}$ & $18$ & $\{0,0,1,2,0\}$ & $3$ \\
\hline 
$\{0,1,0,0,0\}$ & $96$ & $\{0,1,0,0,1\}$ & $144$ & $\{0,1,0,0,2\}$ & $24$ \\
\hline 

$\{0,1,0,1,0\}$ & $48$ & $\{0,1,0,1,1\}$ & $12$ & $\{1,0,0,0,0\}$ & $192$ \\
\hline 
$\{1,0,0,0,1\}$ & $96$ & $\{1,0,0,1,0\}$ & $48$ & & \\ \hline
\end{tabular}%
}
\caption{The number $\mathscr{M}_5(3; \lambda_1,\lambda_2,\ldots,\lambda_5)$ of self-orthogonal codes of length $3$ and type $\{\lambda_1,\lambda_2,\ldots,\lambda_5\}$ over $\mathcal{R}_{5,1} = GR(4,1)[x]/\langle x^4+2,2x\rangle$}
\label{Tab2}
\end{table}
\vspace{-2mm}\begin{example}\label{Example 4.1}  By carrying out computations in Magma \cite{Mag}, we see that the number $\mathscr{M}_4(3; \lambda_1,\lambda_2,\lambda_3,\lambda_4)$ of self-orthogonal codes of length $3$ and type $\{\lambda_1,\lambda_2,\lambda_3,\lambda_4\}$ over $\mathcal{R}_{4,2} = GR(4,2)[x]/\langle x^2+2,2x^2\rangle$ is given by Table \ref{Tab1}, whereas  the number $\mathscr{M}_5(3; \lambda_1,\lambda_2,\ldots,\lambda_5)$ of self-orthogonal codes of length $3$ and type $\{\lambda_1,\lambda_2,\ldots,\lambda_5\}$ over $\mathcal{R}_{5,1} = GR(4,1)[x]/\langle x^4+2,2x\rangle$ is given by Table \ref{Tab2}. In both cases, the numbers agree precisely with  Theorem \ref{t4.1Keven}. 
\end{example}  
     In the following theorem, we obtain  an explicit enumeration formula for the number $\mathscr{M}_e(n;\lambda_1,\lambda_2,\ldots,\lambda_e)$ when $e\geq 4$ and   $n\equiv 0,4\pmod 8.$ 
     \begin{theorem}\label{t4.1Kevenb} For $e\geq 4$ and $n\equiv 0,4\pmod8,$  we have  
\begin{eqnarray*} \vspace{-1mm} \mathscr{M}_e(n;\lambda_1,\lambda_2,\ldots,\lambda_e)=\displaystyle
	 \displaystyle  (2^m)^{\sum\limits_{i=1}^{s}\Lambda_{i}(n-\Lambda_{i+1})+\sum\limits_{j=1}^{s-1+\theta_e}\Lambda_{s+j}(n-\Lambda_{s+j+1}-\Lambda_{s-j+\theta_e})-\sum\limits_{\ell=1}^{s-\kappa_1}\Lambda_\ell-(1-\theta_e)\frac{\Lambda_s(\Lambda_s-1)}{2}} \\
	 \displaystyle  \times \mathcal{S}_{\theta_e}(\lambda_1,\lambda_2,\ldots,\lambda_{s+\theta_e})	 \prod\limits_{a=s+1+\theta_e}^{e}{\lambda_{a}+n-\Lambda_{a}-\Lambda_{e+1-a}\brack \lambda_{a}}_{2^m}, ~~~~~~~~~~~~~~~ \end{eqnarray*} where 
 \vspace{-2mm}
\begin{equation*}\mathcal{S}_{\theta_e}(\lambda_1,\lambda_2,\ldots, \lambda_{s+\theta_e})= \left\{\begin{array}{l}
        \displaystyle \mathfrak{B}(\lambda_1,\lambda_2,\ldots,\lambda_{s+\theta_e})+  2(2^{m})^{\kappa_1}\mathcal{X}(\lambda_1,\lambda_2,\ldots,\lambda_{s+\theta_e})  +  \sum\limits_{\omega=1}^{\kappa_1-\theta_e} (2^m)^{\omega}\mathfrak{D}_\omega(\lambda_1,\lambda_2,\ldots,\lambda_{s+\theta_e})\vspace{1mm}\\ \text{if  } 2\kappa \leq e \text{ with   either } n\equiv 0 ~(\bmod~8) \vspace{-1mm} \text{ or } n\equiv 4 ~(\bmod~8) \text{ and }  m \text{ being even;} \vspace{1mm} \\
\mathfrak{B}(\lambda_1,\lambda_2,\ldots,\lambda_{s+\theta_e})+  \sum\limits_{\omega=1}^{\kappa_1-\theta_e} (2^m)^{\omega}\mathfrak{D}_\omega(\lambda_1,\lambda_2,\ldots,\lambda_{s+\theta_e}) \vspace{1mm} \\ \text{if  } 2\kappa \leq e ,  n\equiv 4 ~(\bmod~8) \text{ and }  m \text{ is odd;}\\
\mathfrak{B}(\lambda_1,\lambda_2,\ldots,\lambda_{s+\theta_e}) +(2^{m})^{s-\kappa_1}\mathcal{Z}(\lambda_1,\lambda_2,\ldots,\lambda_{s+\theta_e}) 
 + \sum\limits_{\omega=1}^{s-\kappa_1-1} (2^m)^\omega \mathfrak{D}_\omega(\lambda_1,\lambda_2,\ldots,\lambda_{s+\theta_e}) \vspace{1mm} \\ \text{if } 2\kappa >e \text{ and } n\equiv 0,4 \pmod 8.\\

\end{array}\right.\vspace{-2mm}\end{equation*} \end{theorem}
\begin{proof}
     When $2\kappa\leq e,$ the desired result follows  by   Lemmas \ref{l3.1Keven}, \ref{p5.1} - \ref{p5.3}.
On the other hand, when  $2\kappa>e,$ we get the desired result by  applying Lemmas \ref{l3.2Keven} - \ref{p5.2}  and \ref{p5.4}.
\end{proof}

\small
\begin{table}[h!]
\centering
\resizebox{\textwidth}{!}{%
\begin{tabular}{|c|c||c|c||c|c|}
\hline
$\{\lambda_1,\lambda_2,\lambda_3,\lambda_4\}$ & $\mathscr{M}_4(4; \lambda_1,\lambda_2,\lambda_3,\lambda_4)$ &
$\{\lambda_1,\lambda_2,\lambda_3,\lambda_4\}$ & $\mathscr{M}_4(4; \lambda_1,\lambda_2,\lambda_3,\lambda_4)$ &
$\{\lambda_1,\lambda_2,\lambda_3,\lambda_4\}$ & $\mathscr{M}_4(4; \lambda_1,\lambda_2,\lambda_3,\lambda_4)$ \\
\hline
$\{0,0,0,1\}$ & $15$ & $\{0,0,0,2\}$ & $35$ & $\{0,0,0,3\}$ & $15$ \\
 \hline
 $\{ 0,0,0,4 \}$ & $ 1 $ & $\{ 0,0,1,0 \}$ & $ 120$ & $\{ 0,0,1,1 \}$ & $ 420 $ \\
 \hline
  $\{ 0,0,1,2 \}$ & $ 210 $ & $\{ 0,0,1,3 \}$ & $ 15$ & $\{ 0,0,2,0 \}$ & $ 560 $ \\
 \hline
  $\{ 0,0,2,1 \}$ & $ 420 $ & $\{ 0,0,2,2 \}$ & $35 $ & $\{ 0,0,3,0 \}$ & $ 120 $ \\
 \hline
  $\{0,0,3,1  \}$ & $  15$ & $\{ 0,0,4,0 \}$ & $1 $ & $\{ 0,1,0,0 \}$ & $ 448 $ \\
 \hline
  $\{ 0,1,0,1 \}$ & $  1568$ & $\{ 0,1,0,2 \}$ & $784 $ & $\{ 0,1,0,3 \}$ & $ 56 $ \\
 \hline
  $\{ 0,1,1,0 \}$ & $ 1344 $ & $\{ 0,1,1,1 \}$ & $1008 $ & $\{ 0,1,1,2 \}$ & $ 84 $ \\
 \hline
  $\{ 0,1,2,0 \}$ & $ 112$ & $\{ 0,1,2,1 \}$ & $14 $ & $\{ 0,2,0,0 \}$ & $ 384 $ \\
 \hline
  $\{1,0,0,0  \}$ & $ 1024 $ & $\{1,0,0,1  \}$ & $ 1536$ & $\{ 1,0,0,2 \}$ & $ 256 $ \\
 \hline
  $\{1,0,1,0  \}$ & $ 1536 $ & $\{1,0,1,1  \}$ & $ 384$ & $\{ 1,0,2,0 \}$ & $ 64 $ \\
 \hline
  $\{ 1,1,0,0 \}$ & $ 768 $ & $\{ 1,1,0,1 \}$ & $ 192$ & $\{ 2,0,0,0 \}$ & $192  $ \\
 \hline
  $\{ 0,2,0,2 \}$ & $ 24 $ & $\{ 0,2,0,1 \}$ & $288 $ &  & \\
 \hline
\end{tabular}%
}
\caption{The number $\mathscr{M}_4(4; \lambda_1,\lambda_2,\lambda_3,\lambda_4)$ of self-orthogonal codes of length $4$ and type $\{\lambda_1,\lambda_2,\lambda_3,\lambda_4\}$ over $\mathcal{R}_{4,1} = GR(4,1)[x]/\langle x^2+2,2x^2\rangle$}
\label{Tab3}
\end{table}
\vspace{-2mm}

\begin{example}\label{Example 4.1}
By carrying out computations in Magma \cite{Mag}, we see   that the number $\mathscr{M}_4(4; \lambda_1,\lambda_2,\lambda_3,\lambda_4)$ of self-orthogonal codes of length $4$ and type $\{\lambda_1,\lambda_2,\lambda_3,\lambda_4\}$ over $\mathcal{R}_{4,1} = GR(4,1)[x]/\langle x^2+2,2x^2\rangle$ is given by Table \ref{Tab3}. These numbers precisely agree with Theorem \ref{t4.1Kevenb}. 
\end{example} 
\vspace{-3mm}
\begin{theorem}\label{t4.1Kevenc}  For $e\geq 4$ and $n\equiv 2,6\pmod 8,$   we have  
\begin{eqnarray*} \vspace{-1mm} \mathscr{M}_e(n;\lambda_1,\lambda_2,\ldots,\lambda_e)=\displaystyle
	 \displaystyle  (2^m)^{\sum\limits_{i=1}^{s}\Lambda_{i}(n-\Lambda_{i+1})+\sum\limits_{j=1}^{s-1+\theta_e}\Lambda_{s+j}(n-\Lambda_{s+j+1}-\Lambda_{s-j+\theta_e})-\sum\limits_{\ell=1}^{s-\kappa_1}\Lambda_\ell-(1-\theta_e)\frac{\Lambda_s(\Lambda_s-1)}{2}} \\
	 \displaystyle  \times \mathcal{S}_{\theta_e}(\lambda_1,\lambda_2,\ldots,\lambda_{s+\theta_e})	 \prod\limits_{a=s+1+\theta_e}^{e}{\lambda_{a}+n-\Lambda_{a}-\Lambda_{e+1-a}\brack \lambda_{a}}_{2^m}, ~~~~~~~~~~~~~~~ \end{eqnarray*} where 
 \vspace{-2mm}
\begin{equation*}\mathcal{S}_{\theta_e}(\lambda_1,\lambda_2,\ldots, \lambda_{s+\theta_e})= \left\{\begin{array}{l}
		\displaystyle \mathfrak{B}(\lambda_1,\lambda_2,\ldots,\lambda_{s+\theta_e})	 \vspace{1mm}  \\ \text{if either }   \kappa=2 \text{  and }e \text{ is odd,~}   \text{or }  \kappa=2,~ (\eta_0)^{\frac{3}{2}}\neq \eta_1 \text{  and }e \text{~is even,}    \text{~or~} \kappa \geq 4\\ \text{and } \eta_1\neq 0, \text{ or }e=\kappa+1; \vspace{1mm} \\
          \mathfrak{B}(\lambda_1,\lambda_2,\ldots,\lambda_{s+\theta_e}) +  2^m\mathcal{Y}(\lambda_1,\lambda_2,\ldots,\lambda_{s+\theta_e})\vspace{1mm} \\ \text{if } 2\kappa \leq  e \text{ with  }\kappa=2,~ (\eta_0)^{\frac{3}{2}}=\eta_1 \text{ and } e \text{ is even;} \\
    \mathfrak{B}(\lambda_1,\lambda_2,\ldots,\lambda_{s+\theta_e}) + \epsilon_1 (2^m)^{\mathrm{f}_{\kappa_1+1}}\mathcal{Y}(\lambda_1,\lambda_2,\ldots,\lambda_{s+\theta_e}) + \sum\limits_{\delta=1}^{\vartheta_1} (2^m)^\delta \mathcal{W}_\delta(\lambda_1,\lambda_2,\ldots,\lambda_{s+\theta_e})\vspace{1mm} \\ \text{if } 2\kappa \leq  e,  \kappa\geq 4  \text{ and }\eta_1=0;\\
  \mathfrak{B}(\lambda_1,\lambda_2,\ldots,\lambda_{s+\theta_e}) + \epsilon_2 (2^m)^{\mathrm{f}_{\kappa_1+1}}\mathcal{Y}(\lambda_1,\lambda_2,\ldots,\lambda_{s+\theta_e}) + \sum\limits_{\delta=1}^{\vartheta_2} (2^m)^\delta \mathcal{W}_\delta(\lambda_1,\lambda_2,\ldots,\lambda_{s+\theta_e})\vspace{1mm} \\ \text{if } 2\kappa >  e,~\kappa\geq 4 \text{ and } \eta_1=0.\\
\end{array}\right.\vspace{-2mm}\end{equation*} 
Here,  in the case $\kappa\geq 4,$ $2\kappa \leq e$   and $\eta_1=0,$  $\vartheta_1$ denotes the largest positive integer satisfying $1\leq \vartheta_1 \leq \mathrm{f}_{\kappa_1}$ and $\eta_{2j-1}=0$ for $1\leq j \leq \vartheta_1,$   and \vspace{-1mm}\begin{equation*}\vspace{-1mm}
    \epsilon_1=\left\{\begin{array}{ll}1 & \text{if } \kappa \text{ is a  singly even integer,}~\eta_{2j-1}=0  \text{ for } 1\leq j\leq \mathrm{f}_{\kappa_1} \text{ and } \eta_{\kappa_1} =(\eta_0)^{\frac{3}{2}};\\ 0 & \text{otherwise.} \end{array}\right.
\end{equation*} 
On the other hand, in the case  $\kappa\geq 4,$ $2\kappa>e$  and $\eta_1=0,$  $\vartheta_2$ denotes  the largest positive integer satisfying $1\leq \vartheta_2 \leq \min\{\mathrm{f}_{\kappa_1},e-\kappa\}$ and $\eta_{2j-1}=0$ for $1\leq j \leq \vartheta_2,$   and \vspace{-1mm}\begin{equation*}\vspace{-2mm}\epsilon_2=\left\{\begin{array}{ll}1 & \text{if } \kappa \text{ is a  singly even integer,} ~\eta_{2j-1}=0  \text{ for } 1\leq j\leq \mathrm{f}_{\kappa_1},~  \eta_{\kappa_1} =(\eta_0)^{\frac{3}{2}} \text{ and }e \geq \frac{3}{2}\kappa+1+\theta_e;\\ 0 & \text{otherwise. } \end{array}\right.\end{equation*} 
\vspace{-2mm}\end{theorem}
\vspace{-2mm}\begin{proof} When $2\kappa\leq e,$ the desired result follows  by   Lemmas \ref{l3.1Keven}, \ref{p5.1},  \ref{p5.7} and \ref{p5.8}.
On the other hand, when  $2\kappa>e,$ we get the desired result by  Lemmas \ref{l3.2Keven}, \ref{p5.1},  \ref{p5.7} and \ref{p5.8}.
\end{proof}
\small
\begin{table}[h!]
\centering
\resizebox{\textwidth}{!}{%
\begin{tabular}{|c|c||c|c||c|c|}
\hline
$\{\lambda_1,\lambda_2,\lambda_3,\lambda_4\}$ & $\mathscr{M}_4(2; \lambda_1,\lambda_2,\lambda_3,\lambda_4)$ &
$\{\lambda_1,\lambda_2,\lambda_3,\lambda_4\}$ & $\mathscr{M}_4(2; \lambda_1,\lambda_2,\lambda_3,\lambda_4)$ &
$\{\lambda_1,\lambda_2,\lambda_3,\lambda_4\}$ & $\mathscr{M}_4(2; \lambda_1,\lambda_2,\lambda_3,\lambda_4)$ \\
\hline
$\{  1,0,0,0\} $ & $0$ & $\{0,1,0,1\}$ & $ 2 $ & $\{ 0,1,0,0 \}$ & $ 4 $ \\  \hline
$\{0,0,1,1 \} $ & $3$ & $\{0,0,2,0\}$ & $ 1$ & $\{ 0,0,1,0\}$ & $ 6 $ \\  \hline
$\{ 0,0,0,2\} $ & $1$ & $\{0,0,0,1\}$ & $ 3 $ &  &  \\  \hline
\end{tabular}%
}
\caption{The number $\mathscr{M}_4(2; \lambda_1,\lambda_2,\lambda_3,\lambda_4)$ of self-orthogonal codes of length $2$ and type $\{\lambda_1,\lambda_2,\lambda_3,\lambda_4\}$ over $\mathcal{R}_{4,1} = GR(4,1)[x]/\langle x^2+2,2x^2\rangle$}
\label{Tab4}
\end{table}
\small
\begin{table}[h!]
\centering
\resizebox{\textwidth}{!}{%
\begin{tabular}{|c|c||c|c||c|c|}
\hline
$\{\lambda_1,\lambda_2,\ldots,\lambda_6\}$ & $\mathscr{M}_6(2; \lambda_1,\lambda_2,\ldots,\lambda_6)$ &
$\{\lambda_1,\lambda_2,\ldots,\lambda_6\}$ & $\mathscr{M}_6(2; \lambda_1,\lambda_2,\ldots,\lambda_6)$ &
$\{\lambda_1,\lambda_2,\ldots,\lambda_6\}$ & $\mathscr{M}_6(2; \lambda_1,\lambda_2,\ldots,\lambda_6)$ \\
\hline
$\{  0,0,1,0,0,0\} $ & $8$ & $\{0,0,1,0,0,1\}$ & $ 4 $ & $\{ 0,0,1,0,1,0 \}$ & $ 2 $ \\  \hline
$\{0,1,0,0,0,0 \} $ & $8$ & $\{0,1,0,0,0,1\}$ & $ 4$ & $\{ 1,0,0,0,0,0\}$ & $ 8 $ \\  \hline
$\{0,0,0,0,0,1\} $ & $3 $ & $\{0,0,0,0,0,2\}$ & $ 1 $ & $\{0,0,0,0,1,0\}$ & $6 $  \\  \hline
$\{0,0,0,0,1,1\} $ & $ 3$ & $\{0,0,0,0,2,0\}$ & $ 1 $ & $\{0,0,0,1,0,0\}$ & $ 12$  \\  \hline
$\{0,0,0,1,0,1\} $ & $ 6$ & $\{0,0,0,1,1,0\}$ & $ 3 $ & $\{0,0,0,2,0,0\}$ & $1 $  \\  \hline\end{tabular}%
}
\caption{The number $\mathscr{M}_6(2; \lambda_1,\lambda_2,\ldots,\lambda_6)$ of self-orthogonal codes of length $2$ and type $\{\lambda_1,\lambda_2,\ldots,\lambda_6\}$ over $\mathcal{R}_{6,1} = GR(4,1)[x]/\langle x^4+2,2x^2\rangle$}
\label{Tab5}
\end{table}
\small
\begin{table}[h!]
\centering
\resizebox{\textwidth}{!}{%
\begin{tabular}{|c|c||c|c||c|c|}
\hline
$\{\lambda_1,\lambda_2,\ldots,\lambda_7\}$ & $\mathscr{M}_7(2; \lambda_1,\lambda_2,\ldots,\lambda_7)$ &
$\{\lambda_1,\lambda_2,\ldots,\lambda_7\}$ & $\mathscr{M}_7(2; \lambda_1,\lambda_2,\ldots,\lambda_7)$ &
$\{\lambda_1,\lambda_2,\ldots,\lambda_7\}$ & $\mathscr{M}_7(2; \lambda_1,\lambda_2,\ldots,\lambda_7)$ \\
\hline
$\{ 0,0,0,0,0,0,1 \} $ & $ 3$ & $\{0,0,0,0,0,0,2  \}$ & $  1$ & $\{ 0,0,0,0,0,1,0  \}$ & $  6$ \\  \hline
$\{ 0,0,0,0,0,1,1 \} $ & $ 3$ & $\{ 0,0,0,0,0,2,0 \}$ & $ 1 $ & $\{ 0,0,0,0,1,0,0  \}$ & $12  $ \\  \hline
$\{ 0,0,0,0,1,0,1 \} $ & $6 $ & $\{ 0,0,0,0,1,1,0 \}$ & $ 3 $ & $\{ 0,0,0,0,2,0,0  \}$ & $ 1 $ \\  \hline
$\{ 0,0,0,1,0,0,0 \} $ & $ 8$ & $\{0,0,0,1,0,0,1  \}$ & $ 4 $ & $\{  0,0,0,1,0,1,0 \}$ & $ 2 $ \\  \hline
$\{ 0,0,0,1,1,0,0 \} $ & $ 1$ & $\{ 0,0,1,0,0,0,0 \}$ & $ 8 $ & $\{ 0,0,1,0,0,0,1  \}$ & $ 4 $ \\  \hline
$\{ 0,0,1,0,0,1,0 \} $ & $ 2$ & $\{  0,1,0,0,0,0,0\}$ & $ 8 $ & $\{ 0,1,0,0,0,0,1  \}$ & $ 4 $ \\  \hline
$\{ 1,0,0,0,0,0,0 \} $ & $0 $ & &  &  & \\  \hline
\end{tabular}%
}
\caption{The number $\mathscr{M}_7(2; \lambda_1,\lambda_2,\ldots,\lambda_7)$ of self-orthogonal codes of length $2$ and type $\{\lambda_1,\lambda_2,\ldots,\lambda_7\}$ over $\mathcal{R}_{7,1} = GR(4,1)[x]/\langle x^4+2(x^3+x+1),2x^3\rangle$}
\label{Tab6}
\end{table}\vspace{-2mm}
\begin{example}\label{Example 4.1}
  Using Magma \cite{Mag},  we see that the number $\mathscr{M}_4(2; \lambda_1,\lambda_2,\lambda_3,\lambda_4)$ of self-orthogonal codes of length $2$ and type $\{\lambda_1,\lambda_2,\lambda_3,\lambda_4\}$ over $\mathcal{R}_{4,1} = GR(4,1)[x]/\langle x^2+2,2x^2\rangle$ is as given by Table \ref{Tab4},  the number $\mathscr{M}_6(2; \lambda_1,\lambda_2,\ldots,\lambda_6)$ of self-orthogonal codes of length $2$ and type $\{\lambda_1,\lambda_2,\ldots,\lambda_6\}$ over $\mathcal{R}_{6,1} = GR(4,1)[x]/\langle x^4+2,2x^2\rangle$ is given by   Table \ref{Tab5}, and  that the number $\mathscr{M}_7(2; \lambda_1,\lambda_2,\ldots,\lambda_7)$ of self-orthogonal codes of length $2$ and type $\{\lambda_1,\lambda_2,\ldots,\lambda_7\}$ over $\mathcal{R}_{7,1} = GR(4,1)[x]/\langle x^4+2(x^3+x+1),2x^3\rangle$ is given by  Table \ref{Tab6}. These values agree precisely with Theorem \ref{t4.1Kevenc}. 
\end{example}  
 \begin{remark}\label{rem8.1}
 When $\kappa$ is even,   Theorem 5.1  of Yadav and Sharma \cite{quasi} follows, as a special case, from  Theorems \ref{t4.4e=k=2} and  \ref{t4.1Keven} - \ref{t4.1Kevenc} by setting $\mathfrak{s}=1$ and $e=\mathtt{t}=\kappa.$ \end{remark}
 We next recall,  by Lemma  \ref{l2.2},  that a self-orthogonal code of type $\{\lambda_1,\lambda_2,\ldots,\lambda_e\}$ and length $n$ over $\mathcal{R}_{e,m}$ is self-dual if and only if $\lambda_j=\lambda_{e-j+2}$ for $1\leq j \leq e+1.$ Thus, on substituting  $\lambda_j=\lambda_{e-j+2}$ for $1\leq j \leq e+1$  in Theorems \ref{t4.4e=3} and \ref{t4.1Keven}, one can deduce enumeration formula for the number $\mathscr{B}_e(n;\lambda_1,\lambda_2,\ldots,\lambda_e).$ 
 In the following theorem, we consider the case $e=\kappa=2$ and obtain an  explicit enumeration formula for the number $\mathscr{B}_2(n;\lambda_1,\lambda_2).$ 
\begin{theorem}\label{t4.3e=k=2}
  For $e=\kappa=2,$ we have 
\begin{equation*}\mathscr{B}_2(n;\lambda_1,\lambda_2)=\left\{ \begin{array}{cl}  \displaystyle \widetilde{S}(n;\lambda_1)(2^m)^{\frac{\lambda_1(\lambda_1+1)}{2}} & \text{if    } 2\lambda_1+\lambda_2=n;\vspace{0.5mm}\vspace{0.5mm}\\ \displaystyle 0  & \text{otherwise,} \end{array}\right.\vspace{-2mm}\end{equation*} where the number $\widetilde{\mathcal{S}}(n;\lambda_1)$ is as determined in Theorem \ref{t4.4e=k=2}. 
    \end{theorem}
\begin{proof}
On substituting $n=2\lambda_1+\lambda_2$ in Theorem \ref{t4.4e=k=2} and using Lemma  \ref{l2.2}, we get the desired result. 
\end{proof}
In the following theorem, we consider the case $e=3$ and $\kappa=2$ and obtain an  explicit enumeration formula for the number $\mathscr{B}_3(n;\lambda_1,\lambda_2,\lambda_3).$ 
\begin{theorem}\label{t4.3e=3}
  For $e=3$ and $\kappa=2,$ we have 
\begin{equation*}\mathscr{B}_3(n;\lambda_1,\lambda_2,\lambda_3)=\left\{ \begin{array}{cl}  \displaystyle (2^m)^{\lambda_1(\lambda_3+\lambda_1)}{\Lambda_2 \brack \lambda_1}_{2^m} \prod\limits_{i=0}^{\Lambda_2-2}\left( \frac{2^{m(n-2-2i)}-1}{2^{m(j+1)}-1}\right)& \text{if }n=2(\lambda_1+\lambda_2) \text{ and }\lambda_2=\lambda_3;\vspace{0.5mm}\vspace{0.5mm}\\ \displaystyle 0  & \text{otherwise.} \end{array}\right.\vspace{-2mm}\end{equation*} 
    \end{theorem}
\begin{proof}
On substituting $\lambda_1=\lambda_4$ and $\lambda_2=\lambda_3$ in Theorem \ref{t4.4e=3} and using Lemma  \ref{l2.2}, we get the desired result. 
\end{proof}
In the following theorem, we provide an explicit enumeration formula for the number $\mathscr{B}_e(n;\lambda_1,\lambda_2,\ldots,\lambda_e)$ for all $e\geq 4.$   
\begin{theorem}\label{t4.2Keven} 
For $e\geq 4,$  we have  
\vspace{-2mm}\begin{eqnarray*}  \mathscr{B}_e(n;\lambda_1,\lambda_2,\ldots,\lambda_e)=\left\{ \begin{array}{cl} \displaystyle
	 \displaystyle  \mathcal{S}_{\theta_e}(\lambda_1,\lambda_2,\ldots,\lambda_{s+\theta_e}) (2^m)^{\sum\limits_{i=1}^{s}\Lambda_{i}(n-\Lambda_{i+1})-\sum\limits_{a=1}^{s-\kappa_1}\Lambda_a-(1-\theta_e)\frac{\Lambda_s(\Lambda_s-1)}{2}} & \text{if } \lambda_j=\lambda_{e-j+2}\\& \text{for }1 \leq j \leq e+1;\vspace{0.5mm}\\ 0 & \text{otherwise},\end{array}\right. \end{eqnarray*} where  $\mathcal{S}_{\theta_e}(\lambda_1,\lambda_2,\ldots, \lambda_{s+\theta_e})=\mathfrak{B}(\lambda_1,\lambda_2,\ldots,\lambda_{s+\theta_e})$ is as determined in Lemma \ref{p5.1})  if $n\equiv 1,3,5,7\pmod8,$ while in the case $n\equiv 0,2,4,6\pmod8,$ the number $\mathcal{S}_{\theta_e}(\lambda_1,\lambda_2,\ldots, \lambda_{s+\theta_e})$ is as determined in Theorems  \ref{t4.1Kevenb} and \ref{t4.1Kevenc}.
 \vspace{-1mm} \end{theorem}
 \begin{proof}
   On substituting $\lambda_{j}= \lambda_{e-j+2}$  for $1 \leq j \leq e+1$ in Theorem \ref{t4.1Keven} and using Lemma  \ref{l2.2}, we get the desired result.
 \end{proof}
\begin{remark}\label{rem8.2}
    When $\kappa$ is even,   Theorem 5.3  of Yadav and Sharma \cite{quasi}  can be deduced, as a special case, from  Theorems \ref{t4.3e=k=2} and \ref{t4.2Keven} by setting $\mathfrak{s}=1$ and $e=\mathtt{t}=\kappa.$
\end{remark} 

\section{Acknowledgements}
The authors gratefully acknowledge the support provided by the Department of Science and Technology, India, under Grant No. DST/INT/RUS/RSF/P-41/2021 with TPN 65025.

{}
\end{document}